\documentclass[a4paper, accepted=2025-12-03, onecolumn]{quantumarticle}
\pdfoutput=1

\usepackage{amsmath}
\usepackage{amssymb}
\usepackage{amsthm}
\usepackage{amsfonts}
\usepackage[numbers]{natbib}
\usepackage[pagebackref]{hyperref}
\hypersetup{
    colorlinks=true, 
    linkcolor=blue, 
    citecolor=blue, 
    urlcolor=blue 
}
\usepackage{microtype}

\usepackage{accents}
\usepackage[boxed,titlenumbered]{algorithm2e}
\usepackage[toc,page]{appendix}
\usepackage{array}
\usepackage{makecell}
\usepackage{arydshln}
\usepackage[utf8]{inputenc}
\usepackage{braket}
\usepackage{booktabs}
\usepackage{bigdelim}
\usepackage{cancel}
\usepackage[shortlabels]{enumitem}
\usepackage{float}
\usepackage{graphicx}
\usepackage[capitalize,nameinlink]{cleveref} 
\usepackage{mathtools}
\usepackage{multirow}
\usepackage{physics}
\usepackage{tabularx}
\usepackage[dvipsnames]{xcolor}
\usepackage{tikz}
\usepackage[textsize=footnotesize]{todonotes}
\usepackage{qcircuit}
\usetikzlibrary{positioning}
\usetikzlibrary{decorations.pathreplacing}
\usetikzlibrary{cd}
\usetikzlibrary{arrows.meta}
\usetikzlibrary{calc}
\usetikzlibrary{shapes.geometric}
\usepackage{bbold}
\usepackage{xparse}

\usepackage{tikzsymbols}
\usepackage{pgfplots} 
\pgfplotsset{compat=1.18}
\usepackage{caption}
\usepackage{subcaption}

\usepackage{complexity}

\usepackage{xspace}
\usepackage{ifthen}

\renewcommand{\backref}[1]{}

\renewcommand{\backrefalt}[4]{%
\ifcase #1 %
\or
[p.\ #2]%
\else
[pp.\ #2]%
\fi}

\makeatletter
\renewcommand{\paragraph}{%
 \@startsection{paragraph}{4}%
 {\z@}{2.25ex \@plus .5ex \@minus .3ex}{-1em}%
 {\normalfont\normalsize\bfseries}%
}
\makeatother

\makeatletter
\newcommand{\para}{%
 \@startsection{paragraph}{4}%
 {\z@}{2ex \@plus 3.3ex \@minus .2ex}{-1em}%
 {\normalfont\normalsize\bfseries}%
}
\makeatother

\newtheorem{lem}{Lemma}
\newtheorem*{lem*}{Lemma}
\crefname{lem}{lemma}{lemmas}
\Crefname{lem}{Lemma}{Lemmas}
\newtheorem{thm}[lem]{Theorem}
\newtheorem*{thm*}{Theorem}

\newtheorem{claim}[lem]{Claim}
\crefname{claim}{claim}{claims}
\newtheorem{cor}[lem]{Corollary}
\newtheorem{defn}[lem]{Definition}
\crefname{defn}{definition}{definitions}

\theoremstyle{definition}

\newtheorem{fact}[lem]{Fact}
\newtheorem{question}[lem]{Question}

\newcommand{\NN}{\mathbb{N}}
\newcommand{\CC}{\mathbb{C}}
\newcommand{\RR}{\mathbb{R}}
\newcommand{\ZZ}{\mathbb{Z}}
\newcommand{\zo}{\{0,1\}}
\newcommand{\zon}{\zo^n}
\newcommand{\pbra}[1]{\left({#1} \right)}
\newcommand{\sbra}[1]{\left[{#1} \right]}
\newcommand{\cbra}[1]{\left\{{#1} \right\}}
\newcommand{\wt}[1]{\widetilde{#1}}

\def\beq{\begin{equation}}
\def\eeq{\end{equation}}

\def\ben{\begin{enumerate}}
\def\een{\end{enumerate}}

\def\bem{\begin{bmatrix}}
\def\eem{\end{bmatrix}}

\newcommand{\majority}{\text{majority}\xspace}
\renewcommand{\parity}{\text{parity}\xspace}
\newcommand{\intparity}{\text{PAR}\xspace} 
\newcommand{\majmod}[1]{\text{majmod}_{#1}\xspace} 
\newcommand{\MM}[1]{\text{MM}_{#1}} 
\newcommand{\GHZ}{\text{GHZ}\xspace}
\newcommand{\EVEN}{\text{EVEN}\xspace}
\newcommand{\CNOT}{\text{CNOT}\xspace}
\newcommand{\NUrot}[1]{A_{#1}}
\newcommand{\multiNUrot}[2]{A_{#1, #2}}
\newcommand{\Urot}[2]{U_{#1, #2}}
\newcommand{\halfbitstring}{B}
\newcommand{\TVD}{\Delta}

\newcommand{\blocksum}[1]{S_{#1}}

\newcommand{\forest}[1]{F_{#1}}
\newcommand{\allsmalltrees}{\mathcal{T}}
\newcommand{\treesum}[1]{S_{#1}}
\renewcommand{\path}[1]{\mathcal{P}_{#1}} 
\newcommand{\treeneighb}{N_{\allsmalltrees}}
\newcommand{\fneighb}{N_f}
\newcommand{\blockidx}[1]{B_{#1}}
\newcommand{\bintree}[1][n]{\mathcal{B}_{#1}}
\newcommand{\pmmajmod}[1]{\text{pmmajmod}_{#1}\xspace} 

\newcommand{\at}[1]{_{#1}} 
\newcommand{\unif}{\text{Unif}} 
\newcommand{\adj}{^\dag}

\newcommand{\cktclass}[1]{\textsf{#1}}
\newcommand{\QNCZ}{\cktclass{QNC}\textsuperscript{0}\xspace}
\newcommand{\NCZ}{\cktclass{NC}\textsuperscript{0}\xspace}
\newcommand{\ACZ}{\cktclass{AC}\textsuperscript{0}\xspace}
\newcommand{\polytime}{\textsf{Poly-time}}

\newcommand{\localityf}{\ell_f}

\newcommand{\xmark}{\textsf{X}}
\newcommand{\cmark}{\checkmark}

\definecolor{orange(colorwheel)}{rgb}{1.0, 0.5, 0.0}

\definecolor{cb-blue-green} {RGB}{  0,  073,  073}
\definecolor{cb-salmon-pink}{RGB}{255, 182, 119}
\definecolor{cb-lilac}      {RGB}{182, 109, 255}

\title{Unconditional Quantum Advantage for Sampling with Shallow Circuits}

\author[1]{Adam Bene Watts}
\author[2, 1, 3]{Natalie Parham}
\affiliation[1]{Institute for Quantum Computing, University of Waterloo, Canada}
\affiliation[2]{Department of Computer Science, Columbia University}
\affiliation[3]{Perimeter Institute for Theoretical Physics, Canada}

\date{}

\begin{document}

\maketitle

\begin{abstract}
    
    Recent work by Bravyi, Gosset, and Koenig showed that there exists a search problem that a constant-depth quantum circuit can solve, but that any constant-depth classical circuit with bounded fan-in cannot. They also pose the question: Can we achieve a similar proof of separation for an input-independent sampling task? In this paper, we show that the answer to this question is yes when the number of random input bits given to the classical circuit is bounded. 
    
    We introduce a distribution $D_{n}$ over $\zon$ and construct a constant-depth uniform quantum circuit family $\{C_n\}_n$ such that $C_n$ samples from a distribution close to $D_{n}$ in total variation distance. For any $\delta < 1$ we also prove, unconditionally, that any classical circuit with bounded fan-in gates that takes as input $kn + n^\delta$ 
    i.i.d. Bernoulli random variables with entropy $1/k$
    and produces output close to $D_{n}$ in total variation distance has depth $\Omega(\log \log n)$. This gives an unconditional proof that constant-depth quantum circuits can sample from distributions that can't be reproduced by constant-depth bounded fan-in classical circuits, even up to additive error. We also show a similar separation between constant-depth quantum circuits with advice and classical circuits with bounded fan-in and fan-out, but access to an unbounded number of i.i.d. random inputs.

    The distribution $D_n$ and classical circuit lower bounds are inspired by work of Viola, in which he shows a different (but related) distribution cannot be sampled from approximately by constant-depth bounded fan-in classical circuits. 
\end{abstract}

\tableofcontents 



\section{Introduction}

At the heart of quantum information theory lies the remarkable observation that quantum devices can process information in ways that classical devices cannot. This is illustrated by the work of Bell, which showed that measurements made on spatially separated parts of a quantum system could produce non-classical correlations. More recently, much excitement surrounding quantum computers comes from the belief that there are problems, such as factoring~\cite{shor1999polynomial}, which can be solved by quantum computers in polynomial time but cannot be solved efficiently by classical computers. 

While Bell's predictions have been verified experimentally~\cite{georgescu2021bell}, there are significant challenges to demonstrating quantum advantage for complex computational problems such as factoring. The current best known quantum algorithms for these problems require construction of a full scale fault-tolerant quantum computer, which is well beyond our current experimental capabilities. Additionally, while it is widely believed that there is no efficient classical algorithm for factoring, this hardness has not been proven formally. Indeed, proving any problem can be solved efficiently by a quantum computer but not by a classical computer would also require proving $\P \neq \PSPACE$~\cite{watrous2008quantum}, constituting a major breakthrough in complexity theory. 

One approach to demonstrating quantum advantage while avoiding these problems is to study the task of sampling from the output distribution of quantum circuits. In this setting it seems possible that even shallow quantum circuits (that is, quantum circuits whose depth is much less than their length) can perform tasks which are still hard classically.
In 2004, Terhal and DiVincenzo provided evidence, later strengthened by Aaronson \cite{aaronson2005quantum}, that there is no polynomial time classical algorithm that takes as input a description of a depth-3 quantum circuit and produces samples from the output distribution of that circuit~\cite{terhal2002adaptive}. 
More recently, a series of works \cite{bouland2018quantum,aaronson2016complexity,boixo2018characterizing} studied the complexity of sampling from the output distribution of a randomly generated shallow quantum circuit (again given a description of the circuit as input) and gave evidence this task couldn't be performed by classical computers in polynomial time. We refer the reader to~\cite{harrow2017quantum} for a more complete discussion of these issues.

While these examples are striking, they do have some limitations. As in the case of factoring, the proofs of classical hardness in the results discussed above are conditional, relying on (natural) complexity-theoretic conjectures. More subtly, the presence of noise in real-world experiments means that quantum computers will not sample from the ideal output distribution of quantum circuits \textit{exactly}. Near-term (NISQ~\cite{preskill2018quantum}) devices will likely only sample from the output distribution of the idealized quantum circuits up to (likely large) \textit{additive error}. Strengthening hardness-of-sampling results of the form described above to this more real-world scenario requires more tenuous complexity-theoretic conjectures.  

An alternate approach, pioneered by Bravyi, Gosset, and Konig in~\cite{bravyi2018quantum} is to compare the computational power of shallow quantum circuits to the computational power of similarly restricted classical circuits. This allowed for an \textit{unconditional} separation: in~\cite{bravyi2018quantum} they showed that constant-depth quantum  (\QNCZ) circuits could 
solve a relational (search) problem 
 -- mapping inputs to valid outputs -- in a way 
that constant-depth, bounded fan-in, classical (\NCZ) circuits could not. Later work~\cite{watts2019exponential,grier2020interactive} improved on their result to give separations between \QNCZ circuits and more powerful classes of constant-depth classical circuits, or between quantum circuits and classical circuits even in the presence of noise~\cite{bravyi2020quantum}.

One notable feature of all of the separations discussed so far is that they are \textit{input-dependent}.  That is, they are based on the classical hardness of mapping some input to some output, e.g., a positive integer to its prime factors in the case of factoring, a circuit description to a sample from its output distribution in the case of circuit sampling problems, or a measurement setting to correlated measurement outcomes in the case of the Bell test. For each of these problems it may be easy to produce a valid output for any fixed input -- the classical hardness is in finding a classical process that takes \textit{all} valid inputs to valid outputs. 

Many important questions in quantum computing, however, concern operations that are \emph{inputless}. A major area of study is the complexity of state preparation \cite{aaronson2016states}, which asks what states can be prepared by quantum computers with bounded resources. The recently proven NLTS theorem~\cite{anshu2022nlts} states that there exist local Hamiltonians whose ground state cannot be prepared by constant-depth quantum circuits. More broadly, the longstanding open question of whether complexity classes $\QMA$ and $\QCMA$ are equal roughly amounts to asking whether every local Hamiltonian has an efficient classical description of its ground state~\cite{aharonov2002quantum}.
Beyond this, the complexity of state preparation has implications in quantum cryptography and physics, with connections to black holes and quantum money \cite{aaronson2016states}.

In this work we study unconditional separations in the style of Bravyi, Gosset, and Konig in the inputless setting. Classical \textit{input-independent sampling} problems can be thought of as the classical analog of state-preparation problems, in which the goal is to sample from a fixed $n$-bit distribution $D_n$ using a classical circuit whose input is fixed to uniformly random bits.\footnote{More formally, the goal, given a family of distributions $\{D_n\}$ that depend only on $n$, is to produce a family of circuits $\{C_n\}$, each of which samples from the appropriate distribution given random bits as input.} 
While input-dependent problems ask about a classical system's ability to \textit{process} information, input-independent problems instead study what distributions classical systems can \textit{prepare}. 

At first glance, it may appear that there is a close connection between input-dependent problems and input-independent sampling problems. If it is hard to map input $x$ to output $f(x)$ in constant-depth, is it also hard to sample from the distribution $(X, f(X))$ where $X$ is uniform? Perhaps surprisingly, the answer to this question is no! To illustrate, consider the $\parity$ function,
which requires $\Omega(\log n /\log\log(n))$ depth to implement with a classical circuit with unbounded fan-in~\cite{haastad1987computational}. Despite this, there is a simple depth-2 classical circuit which maps a random string $r\in \{0,1\}^{n-1}$ to output $(X, \parity(X))$ for uniformly random $X$. This circuit is easy to describe: simply map input $r$ to output
\begin{align*}
   (r_1, r_1 \oplus r_2, r_2 \oplus r_3, \dots, r_{n-2} \oplus r_{n-1} , r_{n-1}) 
\end{align*} 
and check that the output distribution has the desired statistics. A similar trick can be used to sample from the distribution $(X, \textsf{PHP}_n(X))$ where $\textsf{PHP}_n$ is the Parity Halving Problem, a search problem introduced in \cite{watts2019exponential} which separates \QNCZ circuits from constant-depth classical circuits with unbounded fan-in. 

Indeed, in contrast to search problems, where lower bounds against constant-depth circuits have a long history~\cite{haastad1987computational,razborov1987lower,smolensky1987algebraic}, lower bounds for input-independent problems have only been developed recently. Particularly relevant to this paper is a breakthrough result of Viola~\cite{viola2012complexity} in which he gave the first example of a distribution that could not be sampled by constant-depth classical circuits with bounded fan-in, even up to additive error. (In a follow-up work~\cite{viola2014extractors}, Viola also gave a distribution that can not be sampled by constant-depth classical circuits with \textit{unbounded} fan-in. While this result is stronger, the techniques used in~\cite{viola2014extractors} are less natural in the situation studied here). 

A natural question is whether constant-depth quantum circuits can sample from distributions that classical circuits cannot. Indeed, the authors of~\cite{bravyi2018quantum} asked exactly this question: 
\begin{question}[From \cite{bravyi2018quantum}]
\label{question:sampling_separation}
    Does there exist a family of quantum circuits $\{C_n\}_{n\in \NN}$ such that, for each $n\in \NN$, 
    any constant-depth classical circuit with bounded fan-in (\NCZ) with access to uniformly random bits produces a distribution \emph{far} from the output distribution produced by $C_n$ run on the all-zero state?
\end{question}
In the question above we understand \emph{close} and \emph{far} in the sense of additive error (or total variation distance). We quickly review the definition of this distance below. 

\begin{defn}[Total Variation Distance, $\TVD$] The Total Variation Distance (or Statistical Distance) between two distributions $D_1, D_2$ over $\{0,1\}^m$ is
    \begin{align}
        \TVD(D_1, D_2) := \max_{T \subseteq \{0,1\}^m} \bigg|\Pr[D_1 \in T] - \Pr[D_2 \in T]\bigg| = \frac{1}{2} \sum_{a\in \{0, 1\}^m} \bigg|\Pr[D_1 = a] - \Pr[D_2 = a]\bigg|
    \end{align}
\end{defn}

In the next section we discuss the main results of this paper, including a positive answer to \Cref{question:sampling_separation} when the number of random inputs given to the classical circuit is bounded.

\subsection{Results}

The main result of this paper is the following Theorem.

\begin{thm}\label{thm:main_introduction}
    For each $\delta \in [0,1)$, there exists a family of distributions $\{D_n\}$ such that for each $n\in \NN$, $D_n$ is a distribution over $\{0,1\}^n$ and
    \begin{enumerate}[label = (\arabic*)]
        \item There exists a uniform family of constant-depth quantum circuits
        $\{C_n\}$ such that for each $n$, applying $C$ to input $\ket{0^n}$ produces a distribution which has total variation distance at most $1/6 + O(n^{-c})$ from $D_n$ for some $c \in (0,1)$.
        \item Each classical circuit with fan-in 2 which takes $n + n^\delta$ random bits as input and has total variation distance at most $\frac{1}{2} - \omega(1/\log n)$ from $D_n$ has depth $\Omega(\log \log n)$.
    \end{enumerate}
    The distributions $D_n$ constructed are of the form $(X, f(X))$ for a uniformly random bitstring $X$ and function $f: \{0,1\}^{n-1} \rightarrow \{0,1\}$.  
\end{thm}
To provide context for the classical lower bound we note that a uniformly random bitstring has total variation distance $1/2$ from the distribution $D_n$ (or any other distribution of the form $(X, g(X))$ for uniformly random $X$ and function $g : \{0,1\}^{n-1} \rightarrow \{0,1\}$) and so the classical lower bound on total variation distance is near-optimal.

Considering the family of constant-depth quantum circuits that  approximately produce the distributions $\{D_n\}$, we get the following Corollary, showing the answer to \Cref{question:sampling_separation} is YES provided the number of random bits provided to the classical circuit is bounded:
\begin{cor} \label{cor:quantum_classical_TVD}
    There exists a uniform family of constant-depth quantum circuits $\{C_n\}$ such that, for each $\delta \in [0,1)$, any classical circuit with fan-in 2 which takes $n + n^\delta$ random bits as input and samples from the $n$-bit output distribution of $C_n$ to within $1/3 - \omega(1/\log n)$ additive error has depth $\Omega(\log\log n)$.
\end{cor}

While we view \Cref{thm:main_introduction} and \Cref{cor:quantum_classical_TVD} as the main results of the paper, they have some limitations which we address in part with subsequent theorems. Most significantly, the  circuits $C_n$ constructed in \Cref{thm:main_introduction} use arbitrary constant-sized unitaries. In \Cref{sec:U_m_theta_compiling} we review a standard series of arguments which shows that these unitaries can also be compiled in constant depth by circuits consisting of arbitrary single qubit gates and two-qubit CNOT gates. This shows, in particular, that the quantum circuits $\{C_n\}$ are a \emph{uniform} circuit family.

Additionally, it should be noted that even arbitrary single qubit gates have some capabilities which are beyond the reach of $\NC^0$ circuits with uniformly random input. In particular, applying a single controlled-X rotation to a qubit initially in the $\ket{0}$ state and then measuring in the computational basis results in a random bit sampled from a Bernoulli distribution with arbitrary bias (determined by the extent of the rotation). For most biases, reproducing this bias exactly with an $\NC^0$ given uniformly random input requires super-constant depth. It seems possible to build on this observation and produce a separation similar to the one appearing in \Cref{thm:main_introduction} -- indeed, independent from this observation being made here but while we were revising the paper to discuss this issue, this observation was also made formal in~\cite[Theorem 1.10]{kane2024locality}. The authors of that paper also show this observation gives a classical-quantum separation that holds even when the number of (uniformly random) input bits provided to the $\NC^0$ circuit is unbounded.

To address this issue, in \Cref{sec:biased-noise} we extend the classical lower bound of \Cref{thm:main_introduction} to $\NC^0$ circuits with biased inputs.  We also allow more input bits provided their total entropy is bounded. In particular, our lower bound holds when the circuit receives $kn+n^\delta$ independent Bernoulli random bits, each with entropy $1/k$.

Finally, we reemphasize that the classical lower bound in \Cref{thm:main_introduction}, as well as its extension in \Cref{sec:biased-noise}, applies only to $\NC^0$ circuits with at most $n + n^\delta$ random input bits for some $\delta < 1$. That is, the bound only applies to $\NC^0$ circuits with access to at most an extra $n^\delta$ bits of randomness on top of what is required to sample from the distribution $(X,f(X))$. Because of this restriction we took significant care to construct quantum circuits $C_n$ which only involve $n$ qubits, to ensure a fair comparison. Viola~\cite{viola2014extractors} proves sampling lower bounds against stronger circuit classes without the restriction on the number of random input bits, but we have not yet adapted those bounds to our setting; See \Cref{sec:discussion_opn_problems}. As a first step, in \Cref{apx:unlimited-input} we consider classical circuits with an \textit{unlimited} number of inputs but that have bounded fan-in \emph{and} fan-out. In this setting we show there is a distribution $D_n'$ that $\QNC^0$ circuits with quantum advice can approximately sample, but that classical circuits cannot sample in constant depth.

\begin{figure}[ht]
\begin{center}
    \small
    \begin{tabular}{ l| c c c c }
        \multirow{2}{4em}{\textit{Problem}}& \textit{classical} & \textit{constant} & \multirow{2}{6em}{\textit{unconditional}} & \textit{input-} \\ 
        & \textit{hardness} & \textit{depth} & & \textit{independent} \\
        \hline
        Factoring~\cite{shor1999polynomial} & \polytime & \xmark \footnotemark & \xmark & \xmark \\ 
        Sampling depth-3 quantum circuits \cite{terhal2002adaptive,aaronson2005quantum} & \polytime & \cmark & \xmark & \xmark \\
        Random Circuit Sampling \cite{bouland2018quantum,aaronson2016complexity,boixo2018characterizing} & \polytime & \cmark & \xmark & \xmark \\
        2D-HLF \cite{bravyi2018quantum}& \NCZ & \cmark & \cmark  & \xmark \\
        This work & \NCZ & \cmark & \cmark & \cmark 
    \end{tabular}
\end{center}
\caption{Table comparing a few different computational problems with either conditional or unconditional proof of quantum advantage.}
\end{figure}
\subsection{Technical Overview}
\label{subsec:tech_overview}

The distribution used to prove \Cref{thm:main_introduction} is a variation of the distribution $(X, \majmod{p}(X))$, where 
the function $\majmod{p}$ (``Majority mod $p$'') is defined as
\begin{align}
    \majmod{p}(x) =
    \begin{cases}
        0 & \text{if } |x| < p/2 \mod p \\
        1 & \text{if } |x| > p/2 \mod p
    \end{cases}
    && 
    \text{for each } x \in \{0,1\}^{n-1}, \text{ and prime } p.
\end{align}
Viola introduced the $\majmod{p}$ function in \cite{viola2012complexity}. In the same paper he showed a hardness of sampling result for the distribution $(X, \majmod{p}(X))$ similar to Item 2 of \Cref{thm:main_introduction}.

To illustrate some key ideas used in the proof of \Cref{thm:main_introduction}, we first sketch the proof of a weaker sampling separation which holds when $\QNC^0$ circuits take as input a GHZ state $\ket{\GHZ_n} = 1/\sqrt{2}\pbra{\ket{0^n} + \ket{1^n}}$. We refer to these circuits as $\QNC^0$ states with GHZ \textit{advice}. The statement of this separation is as follows. 
\footnotetext{Factoring \textit{can} be accomplished in logarithmic depth on a quantum computer with polynomial time classical post-processing~\cite{cleve2000fast}, or in constant-depth on quantum computer with unbounded fanout gates~\cite{hoyer2005quantum} or intermediate measurements~\cite{browne2010computational}, again with classical post-processing.}

\begin{thm}[Separation with \GHZ advice]\label{thm:sep_with_GHZ_introduction}
    For each $n\in \NN$, $\delta \in [0,1)$, there exists a prime $p$ such~that 
    \begin{enumerate}[label = (\arabic*)]
        \item \label{item:thm:withadvice:quantum-part} There exists a constant-depth quantum circuit that takes the $\GHZ_n$ state as input and produces a distribution which has total variation distance at most $1/6 + O(n^{-c})$ from $(X, \majmod{p}(X) \oplus \parity(X))$ for some $c\in (0,1)$.
        \item \label{item:thm:withadvice:classical-part}Each classical circuit with bounded fan-in which takes $n + n^\delta$ random bits as input and has total variation distance at most $\frac{1}{2} - \omega(1/\log n)$ from $(X, \majmod{p}(X) \oplus \parity(X))$ has depth at least $\Omega(\log\log(n))$.
    \end{enumerate}
\end{thm}

We sketch the proof of Items (1) and (2) of the above theorem in the next two subsections. 

\subsubsection{Approximately sampling from $(X, \majmod{p}(X) \oplus \parity(X))$ with GHZ advice}

We first introduce some non-unitary ``self-controlled" rotation operations, then describe a circuit that approximately samples from $(X, \majmod{p}(X) \oplus \parity(X))$ using those operations. We close this subsection by showing that it is possible to approximate those non-unitary operations with unitary gates. 



\paragraph{Introducing non-unitary gates} 
We introduce the following single-qubit non-unitary operator.
\begin{align}
    \NUrot{\theta} := \ketbra{0}{0} + e^{-i\theta X} \ketbra{1}{1}, && \theta \in \RR
\end{align}
It is straightforward to see that $\NUrot{\theta}$ is linear, but not unitary. $\NUrot{\theta}$ can be interpreted as a ``self-controlled'' $X$ rotation gate. That is, applied to the $\ket{0}$ state, it acts as the identity, and on the $\ket{1}$ state, an $e^{i\theta X}$ is applied. For this reason, it is convenient to draw the gate and its adjoint as

\begin{center}
    \begin{tikzpicture}
    \node at (-2.5, 0) {$A_\theta =$};
    \node[left] (in) at (-1.5,0) {};
    \node[right] (out) at (1.5,0) {};
    \node[draw, rectangle, minimum height=0.5cm, anchor=center] (A) at (0,0) {$e^{-i\theta X}$};
        
    \node[draw, circle, fill=black, inner sep=1pt] (c) at (-1,0) {};
  \draw (in) -- (A) -- (out);
  \draw (c) to[out=90, in=90] (A);
\end{tikzpicture} \qquad
    \begin{tikzpicture}
    \node at (-2.5, 0) {$A_\theta^\dagger =$};
    \node[left] (in) at (-1.5,0) {};
    \node[right] (out) at (1.5,0) {};
    \node[draw, rectangle, minimum height=0.5cm, anchor=center] (A) at (0,0) {$e^{i\theta X}$};
    
    \node[draw, circle, fill=black, inner sep=1pt] (c) at (1,0) {};
  \draw (in) -- (A) -- (out);
  \draw (c) to[out=90, in=90] (A);
\end{tikzpicture}.
\end{center}
Upon post-selection on the output of $\NUrot{\theta}^\dagger$ in the computational basis, we can simplify these operations as follows:
\begin{equation}\label{eq:NUcircuit_identities}
    \begin{tikzpicture}
    \node[left] (in) at (-1.5,0) {};
    \node[right] (out) at (1.5,0) {$\bra{0}$};
    \node[draw, rectangle, minimum height=0.5cm, anchor=center] (A) at (0,0) {$e^{i\theta X}$};
    
    \node[draw, circle, fill=black, inner sep=1pt] (c) at (1,0) {};
    \draw (in) -- (A) -- (out);
    \draw (c) to[out=90, in=90] (A);

    \node at (3,0) {$=$};
    \node[left] (rin) at (3.8,0) {};
    \node[right] (rout) at (6.2,0) {$\bra{0}$};
    \draw (rin) -- (rout);

    \node[left] (in) at (-1.5,-1) {};
    \node[right] (out) at (1.5,-1) {$\bra{1}$};
    \node[draw, rectangle, minimum height=0.5cm, anchor=center] (A) at (0,-1) {$e^{i\theta X}$};
    
    \node[draw, circle, fill=black, inner sep=1pt] (c) at (1,-1) {};
    \draw (in) -- (A) -- (out);
    \draw (c) to[out=90, in=90] (A);

    \node at (3,-1) {$=$};
    \node[left] (rin) at (3.8,-1) {};
    \node[right] (rout) at (6.2,-1) {$\bra{1}$};
    \node[draw, rectangle, minimum height=0.5cm, anchor=center] (rA) at (5,-1) {$e^{i\theta X}$};
    \draw (rin) -- (rA) -- (rout);

\end{tikzpicture}
\end{equation}

\paragraph{Approximate sampling with a non-unitary circuit}
\newcommand{\shrink}[1]{\scalebox{.9}{#1}}
In this section, we aim to illustrate how together these non-unitary gates, and $\GHZ$ advice can be used to produce an $n$-bit distribution where the first $(n-1)$ bits are uniformly random, and the final bit is a function of the Hamming weight of the first $(n-1)$ which approximates $(X, \majmod{p}(X) \oplus \parity(X))$.  
We proceed by walking through an example on $n=4$ qubits.

Starting with the $\ket{\GHZ_4}$ state, we apply a Hadamard gate to each qubit, and then apply our ``self-controlled'' rotation gates $\NUrot{\theta}^\dagger$ to all but the last qubit as shown on the left hand side below. To analyze the output distribution of this circuit we proceed with a series of circuit identities. The first is the fact that  $H^{\otimes n}$ maps the state $\ket{\GHZ_n}$ to the $\ket{\textrm{EVEN}_n}$ state. Here, the $\EVEN_n$ state denotes the uniform superposition over all even $n$-bit strings.

\begin{center}
    \shrink{\begin{tikzpicture}

  \foreach \y in {0,-1,-2}
  {
    \node at (-2.5, \y) {};
    \node[left] (in) at (-1.5,\y) {};
    \node[draw, rectangle] (H) at (-1, \y) {$H$};
    \node[right] (out) at (1.5,\y) {};
    \node[draw, rectangle, minimum height=0.5cm, anchor=center] (A) at (0,\y) {$e^{i\theta X}$};

    \node[draw, circle, fill=black, inner sep=1pt] (c) at (1,\y) {};

    \draw (in) -- (H) -- (A) -- (out);
    \draw (c) to[out=90, in=90] (A);
  }
  \def\y{-3}
  \node at (-2.5, \y) {};
    \node[left] (in) at (-1.5,\y) {};
    \node[draw, rectangle] (H) at (-1, \y) {$H$};
    \node[right] (out) at (1.5,\y) {};
    \node[draw, rectangle, minimum height=0.5cm, anchor=center] (coll) at (0.25,\y) {$e^{-i \pi X / 4}$};
    \draw (in) -- (H) -- (coll) -- (out);

    \node at (2, -1.5) {$=$};
    \def\xshift{6.5}
    \node[left] (inp) at (-2.2, -1.5) {$\ket{\GHZ_4}$};
    \draw[decorate,decoration={brace, amplitude=8pt, mirror}] (-2,0) -- (-2,-3);

  \draw[decorate,decoration={brace, amplitude=8pt, mirror}] (-2+\xshift,0) -- (-2+\xshift,-3);
  \node[left] (inp) at (-2.2+\xshift, -1.5) {$\ket{\mathrm{EVEN}_4}$};

  \foreach \y in {0,-1,-2}
  {
    \node at (-2.5 +\xshift, \y) {};
    \node[left] (in) at (-1.5+\xshift,\y) {};
    \node[right] (out) at (1.5+\xshift,\y) {};
    \node[draw, rectangle, minimum height=0.5cm, anchor=center] (A) at (0+\xshift,\y) {$e^{i\theta X}$};

    \node[draw, circle, fill=black, inner sep=1pt] (c) at (1+\xshift,\y) {};

    \draw (in) -- (A) -- (out);
    \draw (c) to[out=90, in=90] (A);
  }
  \def\y{-3}
  \node at (-2.5+\xshift, \y) {};
    \node[left] (in) at (-1.5+\xshift,\y) {};
    \node[right] (out) at (1.5+\xshift,\y) {};
    \node[draw, rectangle, minimum height=0.5cm, anchor=center] (coll) at (0.25+\xshift,\y) {$e^{-i \pi X / 4}$};
    \draw (in) -- (coll) -- (out);
\end{tikzpicture}} 
\end{center}
Suppose we measure the first $3$ qubits in the $\{\ket{0}, \ket{1}\}$ basis, getting outcomes $x_1, x_2, x_3, x_4 \in \zo$. Using the circuit identities in \Cref{eq:NUcircuit_identities}, we have
\begin{center}
    \shrink{\begin{tikzpicture}
  \node[left] (inp) at (-2, -1.5) {$\ket{\mathrm{EVEN}_4}$};

  \foreach \y [count=\idx] in {0, -1, -2}
  {
    \node at (-1.5, \y) {};
    \node[left] (in\y) at (-1,\y) {};
    \node[right] (out) at (1.5,\y) {$\bra{x_{\idx}}$};
    \node[draw, rectangle, minimum height=0.5cm, anchor=center] (A) at (0,\y) {$e^{i\theta X}$};

    \node[draw, circle, fill=black, inner sep=1pt] (c) at (1,\y) {};

    \draw (in\y) -- (A) -- (out);
    \draw (c) to[out=90, in=90] (A);
  }
  \def\y{-3}
  \node at (-1.5,  \y) {};
    \node[left] (in\y) at (-1,\y) {};
    \node[draw, rectangle, minimum height=0.5cm, anchor=center] (coll) at (0.25,\y) {$e^{-i \pi X / 4}$};
    \node[right] (out) at (1.5,\y) {};
    \draw (in\y)  -- (coll) -- (out);

  \draw[left,decorate,decoration={brace, amplitude=8pt, mirror}] (in0) -- (in\y);

  \node (eq) at (3, -1.5) {$=$};


  \foreach \y [count=\idx] in {0, -1, -2}
  {
    \node[left] (in\y) at (6,\y) {};
    \node[right] (out) at (8.5,\y) {$\bra{x_{\idx}}$};
    \node[draw, rectangle, minimum height=0.5cm, anchor=center] (A) at (7,\y) {$e^{i\theta x_{\idx} X}$};

    \draw (in\y) -- (A) -- (out);


}
    \def\y{-3}
  \node at (-1.5,  \y) {};
    \node[left] (in\y) at (6,\y) {};
    \node[draw, rectangle, minimum height=0.5cm, anchor=center] (coll) at (7.115,\y) {$e^{-i \pi X / 4}$};
    \node[right] (out) at (8.5,\y) {};
    \draw (in\y)  -- (coll) -- (out);

  \draw[left,decorate,decoration={brace, amplitude=8pt, mirror}] (in0) -- (in\y);
  \node[left] (inp) at (5.5, -1.5) {$\ket{\mathrm{EVEN}_4}$};
\end{tikzpicture}




}.
\end{center}
Next, we observe that if you apply the Pauli-$X$ operator to any single qubit of the $\ket{\EVEN_n}$ state, it becomes the $\ket{\mathrm{ODD}_n}$ state. Therefore, it has the same effect as if we instead applied $X$ to the last qubit. The same is true for $X$-rotation gates, so we can push all gates down to the last qubit.

\begin{center}
    \shrink{\begin{tikzpicture}
        \def\xshift{0}
        \node[left] (inp) at (-2+\xshift, -1.5) {$\ket{\mathrm{EVEN}_4}$};
        
        \foreach \y [count=\idx] in {0, -1, -2}
        {
            \node at (-1.5+\xshift, \y) {};
            \node[left] (in\y) at (-1+\xshift,\y) {};
            \node[right] (out) at (1.5+\xshift,\y) {$\bra{x_{\idx}}$};
            \node[draw, rectangle, minimum height=0.5cm, anchor=center] (A) at (0,\y) {$e^{i\theta x_{\idx} X}$};

            \draw (in\y)  -- (A) -- (out);
        }
        \def\y{-3}
        \node at (-1.5+\xshift,  \y) {};
        \node[left] (in\y) at (-1+\xshift,\y) {};
        \node[right] (out) at (1.5+\xshift,\y) {};
        \node[draw, rectangle, minimum height=0.5cm, anchor=center] (coll) at (0.115+\xshift,\y) {$e^{-i\pi X/4}$};
        \draw (in\y) -- (coll) -- (out);
        
        \draw[left,decorate,decoration={brace, amplitude=8pt, mirror}] (in0) -- (in\y);
    
        \node (eq) at (3, -1.5) {$=$};


        \def\xshift{7}
        \node[left] (inp) at (-2+\xshift, -1.5) {$\ket{\mathrm{EVEN}_4}$};
    
        \foreach \y [count=\idx] in {0, -1, -2}
        {
            \node at (-1.5+\xshift, \y) {};
            \node[left] (in\y) at (-1+\xshift,\y) {};
            \node[right] (out) at (1.5+\xshift,\y) {$\bra{x_{\idx}}$};
            \draw (in\y)  -- (out);
        }
        \def\y{-3}
        \node at (-1.5+\xshift,  \y) {};
        \node[left] (in\y) at (-1+\xshift,\y) {};
        \node[right] (out) at (1.5+\xshift,\y) {};
        \node[draw, rectangle, minimum height=0.5cm, anchor=center] (coll) at (0.25+\xshift,\y) {$e^{i (\theta |x| - \pi/4)X}$};
        \draw (in\y) -- (coll) -- (out);
    
        \draw[left,decorate,decoration={brace, amplitude=8pt, mirror}] (in0) -- (in\y);
    \end{tikzpicture}}
\end{center}

Finally, we note that the $\ket{\EVEN_n}$ state can be constructed by first initializing the first $n-1$ qubits to the $\ket{+} = \frac{1}{\sqrt{2}}\pbra{\ket{0} +\ket{1}}$ state, and the final qubit in the $\ket{0}$ state, and subsequently computing the parity of the first $n-1$ qubits into the final register
.
\begin{center}
    \shrink{\begin{tikzpicture}
        \node[left] (inp) at (-2, -1.5) {$\ket{\mathrm{EVEN}_4}$};

  \foreach \y [count=\idx] in {0, -1, -2}
  {
    \node at (-1.5, \y) {};
    \node[left] (in\y) at (-1,\y) {};
    \node[right] (out) at (1.5,\y) {$\bra{x_{\idx}}$};
    \draw (in\y)  -- (out);
  }
  \def\y{-3}
  \node at (-1.5,  \y) {};
    \node[left] (in\y) at (-1,\y) {};
    \node[right] (out) at (1.5,\y) {};
    \node[draw, rectangle, minimum height=0.5cm, anchor=center] (coll) at (0.25,\y) {$e^{i (\theta |x| - \pi/4)X}$};
    \draw (in\y) -- (coll) -- (out);

  \draw[left,decorate,decoration={brace, amplitude=8pt, mirror}] (in0) -- (in\y);

  \node (eq) at (3.3, -1.5) {$=$};

  \def\xshift{6}

  \foreach \y [count=\idx] in {0, -1, -2}
  {
    \node at (-1.5 + \xshift, \y) {};
    \node[left] (in\y) at (-1+\xshift,\y) {$\ket{+}$};
    \node[right] (out) at (1.5+\xshift,\y) {$\bra{x_{\idx}}$};
    \draw (in\y)  -- (out);
  }
  \def\y{-3}
  \node at (-1.5+\xshift,  \y) {};
    \node[left] (in\y) at (-1+\xshift,\y) {$\ket{\mathrm{parity}(x)}$};
    \node[right] (out) at (1.5+\xshift,\y) {};
    \node[draw, rectangle, minimum height=0.5cm, anchor=center] (coll) at (0.25+\xshift,\y) {$e^{i (\theta |x| - \pi/4)X}$};
    \draw (in\y) -- (coll) -- (out);

    \end{tikzpicture}}
\end{center}
And so our measurement outcomes on the first $n-1$ bits are uniformly random, as desired. As for the last qubit---let $b\in \zo$ be the outcome after measuring the last qubit in the standard basis. Then we have that
\begin{align}
    \Pr[b = \parity(x)] = \cos^2\pbra{\theta |x| - \pi / 4}.
\end{align}
This function is periodic in the Hamming weight of $x$, with periodicity $\pi/\theta$. If we set $\theta = \pi/p$ then the output bit $b$ approximately correlates with $\majmod{p}(x) \oplus \parity(x)$, as can be verified analytically or by inspection of the following figures: 

\begin{figure}[H]
    \centering
    \begin{subfigure}{0.45\textwidth}
    \begin{tikzpicture}
        \begin{axis}[
            xlabel={$|x|$},
            width=\textwidth,
            height=0.2\textheight,
            xmin=0, xmax=1,
            ymin=0, ymax=1.1,
            xtick={0.0, 0.25, 0.5, 0.75, 1.0},
            ytick={1/2, 1},
            xticklabels={$0$, $p/4$, $p/2$, $3p/4$, $p$},
            yticklabels={$1/2$, $1$},
            legend style={at={(0.5,1.1)}, anchor=south},
        ]
        \addplot [
            domain=0:1, 
            samples=100, 
            color = cb-blue-green,
            style =  ultra thick,
            ]
            {cos(deg(pi * x - pi/4))^2};
        \addlegendentry{\(\cos^2\pbra{\pi |x|/p - \pi / 4} \)}

        \addplot[
            color=cb-lilac,
            style =  ultra thick,
            ]
            coordinates {
            (0,0)(1/2, 0)(1/2,1)(1,1)
            };
        \addlegendentry{$\majmod{p}(|x|)$}
        \end{axis}
        \end{tikzpicture}
    \end{subfigure}
    \hfill
    \begin{subfigure}{0.45\textwidth}
        \begin{tikzpicture}
            \begin{axis}[
                width=\textwidth,
                height=0.2\textheight,
                xlabel={$|x|$},
                xmin=0, xmax=1,
                ymin=0, ymax=1.1,
                xtick={0.0, 0.25, 0.5, 0.75, 1.0},
                ytick={1/2, 1},
                xticklabels={$0$, $p/4$, $p/2$, $3p/4$, $p$},
                yticklabels={$1/2$, $1$},
                legend style={at={(0.5,1.1)}, anchor=south},
            ]
            \addplot [
                domain=0:1/2, 
                samples=50, 
                color = cb-salmon-pink,
                style =  ultra thick,
                ]
                {1 - cos(deg(pi * x - pi/4))^2};
            \addlegendentry{$\Pr[b\neq \majmod{p}(x)\oplus \parity(x)]$}
    
            \addplot [
                domain=1/2:1, 
                samples=50, 
                color = cb-salmon-pink,
                style =  ultra thick,
                ]
                {cos(deg(pi * x - pi/4))^2};
            \end{axis}
            \end{tikzpicture}
    \end{subfigure}
\end{figure}

\paragraph{Converting to Unitary}
Our non-unitary circuits are helpful for initial circuit design, but we need to somehow port them back over to be unitary --- while maintaining their low-depth. To this end, we make use of the following two insights to construct a unitary circuit such that the output distribution is very close to that of the non-unitary circuit.
\begin{enumerate}
    \item We do not need to find a unitary that is close to the circuit (In fact, this is likely not possible). It is sufficient to instead find a unitary that has the same behavior with respect to its action on the $\GHZ$ state.
    \item We introduce a multi-qubit non-unitary gate $\NUrot{m, \theta}$ acting on $m$-qubits that has the same action as $\NUrot{\theta}^{\otimes m}$ when applied to the $\GHZ$ state, and becomes closer to unitary as $m$ increases.
\end{enumerate}
In \Cref{sec:GHZ_state_sampling} we make this outline rigorous to construct a quantum circuit that, with advice, samples approximately from the distribution $(X, \majmod{p}(X)\oplus \parity(X))$, where $X\sim\unif\pbra{\zo^{n-1}}$.

\subsubsection{Classical circuit lower bound for $(X, \majmod{p}(X) + \parity(X))$} 

The proof of a classical lower bound for the distribution $(X, \majmod{p}(X) + \parity(X))$ closely follows Viola's techniques in \cite{viola2012complexity}. Rather than explicitly lower bounding classical circuit depth, Viola proves lower bounds for the \emph{locality} of functions. To illustrate the relationship between locality and circuit depth let $f: \{0,1\}^\ell \to \{0,1\}^{n}$ be a function implemented by a classical circuit with bounded fan-in.
We say that $f$ is $\localityf$-local if, for each $i\in [n]$, the $i$-th output bit of $f(u)$
depends on at most $\localityf$ bits of the input $u$. Any circuit with depth $d$ and fan-in $\chi$ will have locality at most $\chi^d$. 
And so, to prove a circuit lower bound of $\Omega(\log\log n)$ for sampling from the distribution $(X, \majmod{p}\oplus \parity(X))$ it suffices to prove that there exists some $k>0$ such that any function with locality at most $\Omega(\log^k n)$ cannot sample from the distribution $(X, \majmod{p}\oplus \parity(X))$ given access to uniformly random bits as input. 


Both our proof of sampling hardness for $(X, \majmod{p}(X) \oplus \parity(X))$ and Viola's original proof of hardness for $(X, \majmod{p}(X))$ begin with the observation that for any $\localityf$-local function $f: \{0,1\}^\ell \to \{0,1\}^{n}$ there exists a partition of the input $u = (x, y)$ and a permutation of output bits of $f(x,y)$ such that:\footnote{We use ``$\circ$'' to denote concatenation.}
\begin{align}
    f(x,y) = g_1(x_1, y) \circ g_2(x_2, y) \circ \dots \circ g_s(x_s, y) \circ h(y),
\end{align}
where each $g_i(x_i, y)$ is a subset (or ``block'') of the output bits that are completely determined by $y$ and a single bit of $x$, and $s = \Omega(n/\localityf^2)$. Therefore, if we fix the input bits $y$, each of the blocks $g_i$ are independent. Let $z\in \{0,1\}^{n-1}$ be the first $n-1$ outputs of $f(x,y)$ and let $b$ be the final output bit. We can also assume without loss of generality (by absorbing at most one $g_i$ into $h$) that the last output bit is not permuted, so $b$ only depends on $y$.  In order for the function $f$ to sample from the correct distribution the output bits $z$ must be uniformly distributed and, for every input $(x,y)$, we must have $\majmod{p}(z) \oplus \parity(z) = b$. After fixing the input bits $y$, the Hamming weight of $z$ is a sum of independent random variables but $b$ is fixed. Then (still following Viola) we show that if many of these independent variables are fixed the output distribution of $z$ will not have sufficiently high entropy. Alternatively, if they are unfixed, the condition $\majmod{p}(z) \oplus \parity(z) = b$ is unlikely to be satisfied. Making these observations formal completes the proof. The full details of this argument are given in ~\Cref{sec:classical_hardness_with_GHZ}.

\subsubsection{Removing the GHZ advice}

To complete the proof sketch of~\Cref{thm:main_introduction} we need to extend a proof of~\Cref{thm:sep_with_GHZ_introduction} to give a sampling separation \emph{without} $\GHZ$ advice. To do this, we replace the $\GHZ$ state in the quantum circuit used to prove \Cref{thm:sep_with_GHZ_introduction} with a ``Poor-Man's $\GHZ$ state'' (introduced in \cite{watts2019exponential}) defined over a binary tree $\bintree[]$. An $n$ qubit Poor-Man's GHZ state can be prepared by a constant-depth circuit acting on $2n-1$ qubits followed by a measurement of $n-1$ auxiliary qubits. The remaining state is equivalent to the GHZ state with some Pauli terms applied to it. We can determine which Pauli operations will ``correct'' the state back to the GHZ state as a function of our measurement outcomes. However, determining these corrections requires $\Omega(\log n)$ depth -- which we cannot afford in this shallow circuit setting. Instead, we absorb the Pauli corrections into the definition of the target distribution. 

The result is a circuit that samples approximately from a modified version of the $(X, \majmod{p}(X) \oplus \parity(X))$ distribution. Specifically, the new circuit (including the measured auxiliary qubits) approximately samples from a distribution of the form $(X, \MM{p}(S_X) \oplus \parity(X))$  where
\begin{align}
    \MM{p}(j) := \begin{cases}
    0 & \text{if } j  < p/2 \mod p \\
    1 & \text{if } j > p/2  \mod p
    \end{cases} && \text{for } j\in \ZZ.
\end{align}
and $S_X$ is a sum of terms that depends on output bits $X\in \{0,1\}^{n-1}$. In particular, $S_X$ is a weighted sum of the bits of $X$ where the sign of the weight for each bit $X_i$ may depend on many other output variables. In the body of the paper we introduce the notation $\MM{p}(S_X) \oplus \parity(X) : = \pmmajmod{p}(X)$ to describe this new function. 

Unlike the function $\majmod{p}(X) \oplus \parity(X)$, the function $\pmmajmod{p}(X)$ does not just depend on the Hamming weight of $X$. This introduces a complication when trying to show the classical hardness of sampling using Viola's previously discussed lower bounding technique, since this technique relied on the fact that the Hamming weight of $X$ could be written as a sum of the Hamming weights of disjoint blocks $g_i$ of output bits (and that these blocks became independent after fixing enough input bits). To get around this we show that, after fixing additional input bits (and hence some output bits), we can find blocks of unfixed output bits which each depend on disjoint single input bits and which contribute to disjoint terms in the sum $S_X$. After showing this additional detail, the proof of the lower bound proceeds similarly to Viola's. Although perhaps conceptually straightforward, this argument is mathematically delicate, and relies on careful counting related to the binary tree layout used to construct the poor man's $\GHZ$ state. Details of the new circuit construction (with the ``Poor Man's GHZ state'') are given in~\Cref{sec:quantum_circuit_no_GHZ}, while details of the classical lower bound for this new function are given in~\Cref{sec:classical_hardness_main}.

\section{Discussion and Open Problems}
\label{sec:discussion_opn_problems}
Our results show that $\QNCZ$ circuits can sample from distributions that $\NCZ$ circuits cannot. Below we list a few ways in which we think these results could potentially be extended. 

\begin{itemize}

\item In an experiment with the goal of demonstrating quantum advantage, one would like to not just construct a $\QNCZ$ circuit that samples from a distribution which $\NCZ$ circuits cannot, but also \textit{verify} that the distribution sampled from is indeed hard to sample from classically. How many samples are needed for this verification? Can the circuit be modified to make the verification easier? We point out here that the constant total variation distance in \Cref{cor:quantum_classical_TVD} means that only a few samples are needed to verify that the distribution produced by the described quantum circuit is not produced by a fixed $\NCZ$ circuit, for any specific choice of circuit. However, ruling out \textit{all} distributions producible by
$\NCZ$ circuits is a harder~task.

\item The procedure described in \Cref{sec:U_m_theta_compiling} for compiling the $\Urot{m}{\theta}$ unitary is unlikely to produce an ``optimal'' compilation. With careful thought it may be possible to find a more natural compilation technique that produces $\Urot{m}{\theta}$ gates while requiring many fewer elementary gates. Finding such a compilation would likely make an experimental implementation of the circuits described in this paper much more feasible. 

\item Can we get rid of the limitation on the number of inputs to the classical circuit? In \Cref{apx:unlimited-input} we make some progress in this direction. We consider classical circuits with an \textit{unlimited} number of inputs but that have bounded fan-in \emph{and} fan-out. We show that such classical circuits of depth $o(\log \log n)$ produce distributions far from $D = (X, \majmod{p}(X))$. Whereas, as we saw in \Cref{sec:GHZ_state_sampling} constant-depth quantum circuits with bounded fan-in and fan-out, when given a $\GHZ$ advice state, can sample close to~$D$. 

\item Can we prove an input-independent sampling separation between $\QNCZ$ and $\ACZ$ circuits? Notably, in \cite{viola2014extractors}, Viola proves certain distributions cannot be produced by $\ACZ$ circuits. Can these techniques be extended to $\QNCZ$ circuits? If so, we would have a novel technique for lower-bounding the circuit complexity of quantum states. If not, we should be able to find a $\QNCZ$ circuit that samples from one of these distributions, producing the desired sampling separation. 
\end{itemize}

\section{Acknowledgments}
The authors would like to thank David Gosset for helpful discussions, and Ansis Rosmanis for sharing an insightful note.
They thank Angus Lowe for insightful discussions which  motivating the study of classical circuits with biased input in \Cref{sec:biased-noise}.
They also thank  Michael Oliveira for insights regarding compilation of the $\Urot{m}{\theta}$ unitary, which are discussed in \Cref{sec:U_m_theta_compiling}.

\section{Reader's Guide}

The remainder of the main body of this paper is devoted to a formal proof of~\Cref{thm:main_introduction}. 
This proof involves the same high-level ideas as the proof sketch given in~\Cref{subsec:tech_overview}, but results are presented in a slightly different order.

\Cref{sec:GHZ_state_sampling} constructs a $\QNC^0$ with GHZ advice which samples from the distribution close to $(X,\majmod{p}(X) \oplus \parity(X))$. This is done in two steps. \Cref{ssec: non-unitary majority GHZ sampling} shows how to sample from a distribution close to $(X,\majmod{p}(X) \oplus \parity(X))$ using a shallow circuit with GHZ advice and non-unitary operations. Then \Cref{ssec: unitary majority GHZ sampling} shows how to replace the circuit's non-unitary operations with unitary operations while preserving the output distribution. This section ends with~\Cref{thm:Urot_majmod_sampling}, which gives a formal proof of Item (1) of~\Cref{thm:sep_with_GHZ_introduction}.

\Cref{sec:quantum_circuit_no_GHZ} builds on the techniques of the previous section to give a $\QNC^0$ circuit without GHZ advice which samples from a distribution close to the distribution $(Z, \pmmajmod{p}(Z))$. (Both $X$ and $Z$ are uniformly random bit strings, but we distinguish them because subsets of the $Z$ string are labeled differently than $X$). The function $\pmmajmod{p} : \{0, 1\}^n \rightarrow \{0,1\}$ is defined formally in this section (\Cref{defn:pmmajmod}). This section ends with~\Cref{thm:Urot_pmmajmod_sampling}, which gives a formal version of Item (1) of~\Cref{thm:main_introduction}. 

\Cref{sec:classical_hardness_with_GHZ} proves the classical hardness of sampling from the distribution $(X, \majmod{p}(X) \oplus \parity(X))$.  The main result in this section is~\Cref{thm:classical_LB_with_GHZ}, which gives a formal proof of Item (2) of~\Cref{thm:sep_with_GHZ_introduction}. This section is intended primarily as a ``warm-up'' for the next section, with a couple lemmas that will be reused in 
the proof of~\Cref{thm:main_introduction}. A reader could skip this section and refer back to the lemmas as needed to understand the proof of ~\Cref{thm:main_introduction}. Still, the authors suggest a reader at least skim this section before reading the more complicated proof in the next section. 

\Cref{sec:classical_hardness_main} proves lower bounds on the depth of classical circuits which sample from the distribution $(Z, \pmmajmod{p}(Z))$. The main result of this section is~\Cref{thm:classical_LB_tree}, which gives a proof of Item (2) of~\Cref{thm:main_introduction}.  

The appendixes to this paper prove some extra results which add additional context to \Cref{thm:main_introduction}. In \Cref{sec:U_m_theta_compiling}  we outline an efficient algorithm for constructing our quantum circuits in \Cref{thm:main_introduction} --- showing that they form a uniform quantum circuit family. In \Cref{apx:unlimited-input} we show a classical circuit lower bound against $(X, \majmod{p}(X))$ in the setting where the circuit has unlimited inputs, but bounded fan-in \emph{and} fan-out. In \Cref{sec:biased-noise} we prove a more general version of the classical lower bound in \Cref{thm:main_introduction} that allows for biased inputs.

\section{Sampling From \texorpdfstring{$(X,\majmod{p}(X) \oplus \parity(X))$}{(X, majmod(X) + parity(X))} Using a GHZ State}

\label{sec:GHZ_state_sampling}

In this section we consider constant-depth quantum circuits with access to an $n$-qubit GHZ state as input. We show these circuits can produce samples close to the distribution $(X,\majmod{p}(X) \oplus \parity(X))$, where $X$ is a uniformly random bitstring of length $n-1$. We will prove this result in two steps -- in \Cref{ssec: non-unitary majority GHZ sampling} we give a ``quantum-like'' circuit that samples from the correct distribution but includes non-unitary single-qubit operations. In \Cref{ssec: unitary majority GHZ sampling} we show how to replace those non-unitary operations with multi-qubit (but still constant-sized) unitaries. Before beginning these proofs we review some details about GHZ states. 

\paragraph{Review of GHZ States}
\label{ssec: GHZ review}

An $n$-qubit GHZ state is defined to be the state 
\begin{align}
    \ket{\GHZ_n} = \frac{1}{\sqrt{2}} \left(\ket{0}^{\otimes n} + \ket{1}^{\otimes n} \right).
\end{align}
It is well-known that applying a Hadamard transform to each qubit of a GHZ state produces a uniform superposition over bitstrings with even Hamming weight: 
\begin{align}
    H^{\otimes n} \ket{\GHZ_n} = 2^{-n/2} \sum_{e \in E_n} \ket{e}
\end{align}
where $E_n$ is the set containing all even parity $n$-bit strings. We can equivalently describe this state as a coherent superposition of $n-1$ random bits and a final bit whose value equals the parity of the $n-1$ other bits, so 
\begin{align}
\label{eq:Hadarmarded_GHZ}
    H^{\otimes n} \ket{\GHZ_n} = \left(\prod_{i=1}^{n-1} \CNOT_{i,n} \right)\ket{+}^{\otimes n-1} \otimes \ket{0} 
\end{align}
where $\CNOT_{i,j}$ denotes a \CNOT gate controlled on qubit $i$ and applied to qubit $j$. \Cref{eq:Hadarmarded_GHZ} will be our starting point for designing circuits that use the GHZ state as a resource state. 

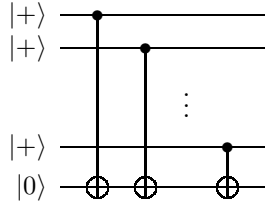
\begin{figure}[ht]
\centerline{
\Qcircuit @C=1em @R=1em {
\lstick{\ket{+}}    & \ctrl{4}  & \qw       & \qw           & \qw       & \qw \\
\lstick{\ket{+}}    & \qw       & \ctrl{3}  & \qw           & \qw       & \qw \\
                    &           &           & \push{\vdots} &           & \\
\lstick{\ket{+}}    & \qw       & \qw       & \qw           & \ctrl{1}  & \qw \\
\lstick{\ket{0}}    & \targ     & \targ     & \qw           & \targ     & \qw 
} 
}
\caption{A circuit constructing the state $H^{\otimes n} \ket{\GHZ_n}$, as described in \Cref{eq:Hadarmarded_GHZ}.} 
\label{fig:Hadamarded_GHZ_circuit}
\end{figure}

\subsection{Sampling with non-unitary operations}
\label{ssec: non-unitary majority GHZ sampling}

We now consider constant-depth quantum circuits augmented with specific single qubit non-unitary ``gates'' $\NUrot{\theta}$, which we will soon define. We show these circuits can sample (approximately) from the distribution $(X, \majority(X) \oplus \parity(X))$. While this model is non-physical, introducing it allows us to isolate some key ideas which we will reuse in the fully quantum circuit developed in the next section.

First, for each $\theta \in \mathbb{R}$, define the (non-unitary) matrix $\NUrot{\theta} \in \CC^{2\times 2}$ to be the matrix which acts on the single-qubit computational basis states as 
\begin{align}
    \NUrot{\theta} \ket{0} &= \ket{0} \\
    \NUrot{\theta} \ket{1} &= \exp(-i\theta X) \ket{1}
\end{align}
When drawing circuit diagrams in this section we sometimes include $\NUrot{\theta}$ gates, and understand that they represent the matrix $A$ acting on the qubits indicated. We also sometimes draw $\NUrot{\theta}\adj$ gates, which represent the adjoint of the matrix $\NUrot{\theta}$ acting on the qubits indicated. 

We now prove the following useful circuit identity. 
\begin{lem}
\label{lem:NUrot_identity}
For any one qubit state $\ket{\psi}$ and computational basis state $\ket{x}$ with $x \in \{0,1\}$, we have
\begin{align}
    \bra{x}_2 \left(\NUrot{\theta}\adj\right)_2 \CNOT_{2,1} \ket{\psi}_1 \ket{+}_2 = \frac{1}{\sqrt{2}} \exp(i (\theta + \pi/2) x X_1 ) \ket{\psi}_1
\end{align}
\end{lem}

\begin{proof}
Direct computation gives 
\begin{align}
    \bra{x}_2 \left(\NUrot{\theta}\adj\right)_2 \CNOT_{2,1} \ket{\psi}_1 \ket{+}_2 &= \bra{x}_2 \exp(i \theta x X_2) \CNOT_{2,1} \ket{\psi}_1 \ket{+}_2 \\
    &= \bra{x}_2 \CNOT_{2,1} \exp(i \theta x X_1X_2) \ket{\psi}_1 \ket{+}_2 \\
    &= \bra{x}_2 \CNOT_{2,1} \exp(i \theta x X_1) \ket{\psi}_1 \ket{+}_2 \\
    &= \exp(i (\theta + \pi/2) x X_1) \ket{\psi}_1 \braket{x}{+}_2 \\
    &= \frac{1}{\sqrt{2}} \exp(i (\theta + \pi/2) x X_1) \ket{\psi}_1
\end{align}
where we used on the first line that 
\begin{align}
    \NUrot{\theta}{\ket{x}} = \exp(-i \theta X x) \ket{x}
\end{align}
by definition, the commutation relation\footnote{To prove the implication, use the standard decomposition $\exp(i \theta X) = \cos( \theta) + i \sin(\theta) X$, then commute the resulting terms.} 
\begin{align}
    X_2 \CNOT_{2,1} &= \CNOT_{2,1} X_{1}X_{2} \\
    \implies \exp(i \theta X_2) \CNOT_{2,1} &= \CNOT_{2,1} \exp(i \theta X_{1}X_{2})
\end{align}
on the second line, that $\ket{+}$ is a 1-eigenstate of the $X$ operator on the third line, and then the definition of the $\CNOT$ gate and the $\ket{+}$ state on the final two lines. \Cref{fig:NUrot_circuit_identity} gives a diagrammatic version of this proof. 
\end{proof}

\begin{figure}[ht]
\begin{align*}
&\Qcircuit @C=1em @R=1em {
\lstick{\ket{\psi}}     & \targ         & \qw                       & \qw   & \qw & & 
\lstick{\ket{\psi}}     & \targ         & \qw                       & \qw   & \qw\\
&&&&& \push{\rule{0em}{0em}=\rule{2em}{0em}} & \\
\lstick{\ket{+}} & \ctrl{-2}     & \gate{\NUrot{\theta}\adj} & \qw   &  {\bra{x}} & &
\lstick{\ket{+}} & \ctrl{-2}     & \gate{\exp(i \theta x X)} & \qw   &  {\bra{x}} &
} 
\\
\\
\cdashline{1-2}
\\
&\hspace{75pt}\Qcircuit @C=1em @R=1em {
&\lstick{\ket{\psi}}     & \multigate{2}{\exp\left(i \theta  x XX\right)}     & \targ         & \qw   & \qw \\
\push{\rule{0em}{0em}=\rule{2em}{0em}} & \\
&\lstick{\ket{+}}        & \ghost{\exp\left(i \theta  x XX\right)}            & \ctrl{-2}      & \qw   &  {\bra{x}}
}
\\
\\
\cdashline{1-2}
\\
&\hspace{75pt}\Qcircuit @C=1em @R=1em {
&\lstick{\ket{\psi}}     & \gate{\exp\left(i \theta  x X\right)}    & \targ         & \qw   & \qw \\
\push{\rule{0em}{0em}=\rule{2em}{0em}} & \\
&\lstick{\ket{+}}        & \qw                                       & \ctrl{-2}      & \qw   &  {\bra{x}}
}
\\
\\
\cdashline{1-2}
\\
&\hspace{75pt}\Qcircuit @C=1em @R=1em {
&\lstick{\ket{\psi}}     & \gate{\exp\left(i x(\theta +\pi/2) X\right)}     & \qw   & \qw \\
\push{\rule{0em}{0em}=\rule{2em}{0em}} & \\
&\lstick{\ket{+}}        & \qw                                       & \qw   &  {\bra{x}}
}   
\end{align*}

\caption{A diagrammatic proof of \Cref{lem:NUrot_identity}. The equivalence between each line is explained in the proof of the lemma.}
\label{fig:NUrot_circuit_identity}
\end{figure}
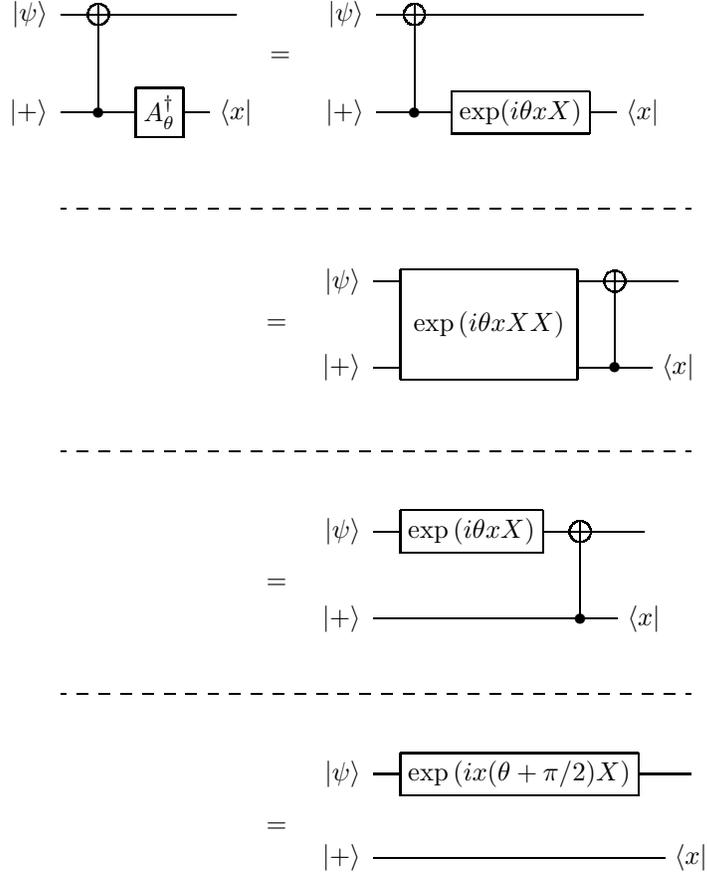

We now prove the main result of this section and construct a constant-depth circuit, with a $\GHZ$ state as input and $\NUrot{\theta}$ gates -- which samples approximately from the distribution $(X, \majmod{p}(X))$ for any $p$. The construction builds on \Cref{lem:NUrot_identity} as well as the observations about the $\GHZ$ state discussed in \Cref{ssec: GHZ review}. 

\begin{thm}
\label{thm:NUrot_majmod_sampling}
For each prime number $p$, there is a constant-depth circuit consisting of one and two-qubit unitary gates and $\NUrot{\theta}$ operations which takes a GHZ state as input and produces an output which, when measured in the computational basis, produces an output distribution $(X', Y)$ with 
\begin{align}
    \TVD((X', Y), (X, \majmod{p}(X) \oplus \parity(X))) \leq  \frac{1}{2} - \frac{1}{\pi} + \frac{1}{2p} +  O(p^{3/2}e^{-n/p^2}).
\end{align}
\end{thm}

\begin{proof}
We first describe the circuit which, when measured in the computational basis, produces output which correlates with $(X, \majmod{p}(X) \oplus \parity(X))$. Fix $\theta = \pi/p$. The circuit takes as input a $\GHZ$ state, applies a Hadamard transform to each qubit of the state, then applies a  $\NUrot{\theta}\adj$ operation to the first $n-1$ qubits in the GHZ state and a $\exp(-i\pi X/4)$ rotation to the final qubit. This circuit is indicated diagrammatically in \Cref{fig:majmod_p_NUrot_circuit}.

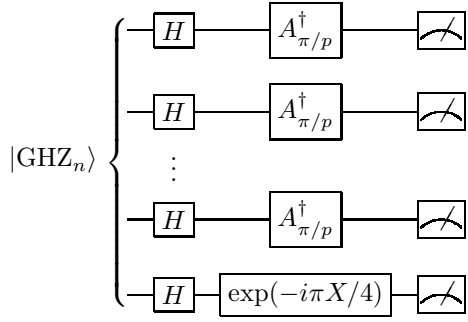
\begin{figure}

\begin{align*}
\Qcircuit @C=1em @R=1em {
\lstick{}   & \gate{H}  & \gate{\NUrot{\pi/p}\adj} & \meter \\
\lstick{}   & \gate{H}  & \gate{\NUrot{\pi/p}\adj} & \meter \\
& {\makecell{\vdots\\}} \\ 
\lstick{}   & \gate{H}  & \gate{\NUrot{\pi/p}\adj} & \meter \\  
\lstick{}   & \gate{H}  & \gate{\exp(-i \pi X / 4)}  & \meter 
\inputgroupv{1}{5}{0.8em}{4.9em}{\ket{\GHZ_n} \;\;\;\;\;\;\;\;} 
} 
\end{align*}
\caption{Constant-depth circuit producing approximate samples from the distribution $(X, \majmod{p}(X) \oplus \parity(X))$.}
\label{fig:majmod_p_NUrot_circuit}
\end{figure}

To prove this circuit samples (approximately) from the correct distribution we write the (unnormalized) output state of the circuit conditioned on the first $n-1$ qubits of the circuit being measured in computational basis state $\ket{x} = \ket{x_1} \otimes \ket{x_2} \otimes ... \otimes \ket{x_{n-1}}$ as:
\begin{align}
    &\bra{x}_{1...n-1}\left( \left(\NUrot{\pi/p}\adj\right)^{\otimes n-1} \otimes \exp(-i\pi X/4) \right) H^{\otimes n}\ket{\GHZ_n} \nonumber\\
    &\hspace{80pt}= \bra{x}_{1...n-1}\left( \left(\NUrot{\pi/p}\adj\right)^{\otimes n-1} \otimes \exp(-i\pi X/4) \right) \left(\prod_{i=1}^{n-1} \CNOT_{i,n} \right) \ket{+}^{\otimes n-1} \otimes \ket{0} \label{eq:NU_analysis1}\\
    &\hspace{80pt}= \prod_{i=1}^{n-1} \matrixel{x_i}{ \NUrot{\pi/p}\adj \left(\CNOT_{i,n}\right) }{+}_i \otimes \exp(-i\pi X/4) \ket{0}_n \label{eq:NU_analysis2}\\
    &\hspace{80pt}= 2^{-(n-1)/2} \exp(i X \left(-\frac{\pi}{4} + \sum_{i=1}^{n-1} x_i \left(\frac{\pi}{p} + \frac{\pi}{2}\right)\right) )\ket{0}_n \label{eq:NU_analysis3}
\end{align}
where we used \Cref{eq:Hadarmarded_GHZ} on the first line, reordered terms on the second (noting that $\exp(i\pi X/4)_n$ commutes with $\CNOT_{i,n}$ for any $i \in [n-1]$), and then used \Cref{lem:NUrot_identity} on the third. A diagrammatic version of this analysis is given in \Cref{fig:majmod_p_NUrot_circuit_analysis}.

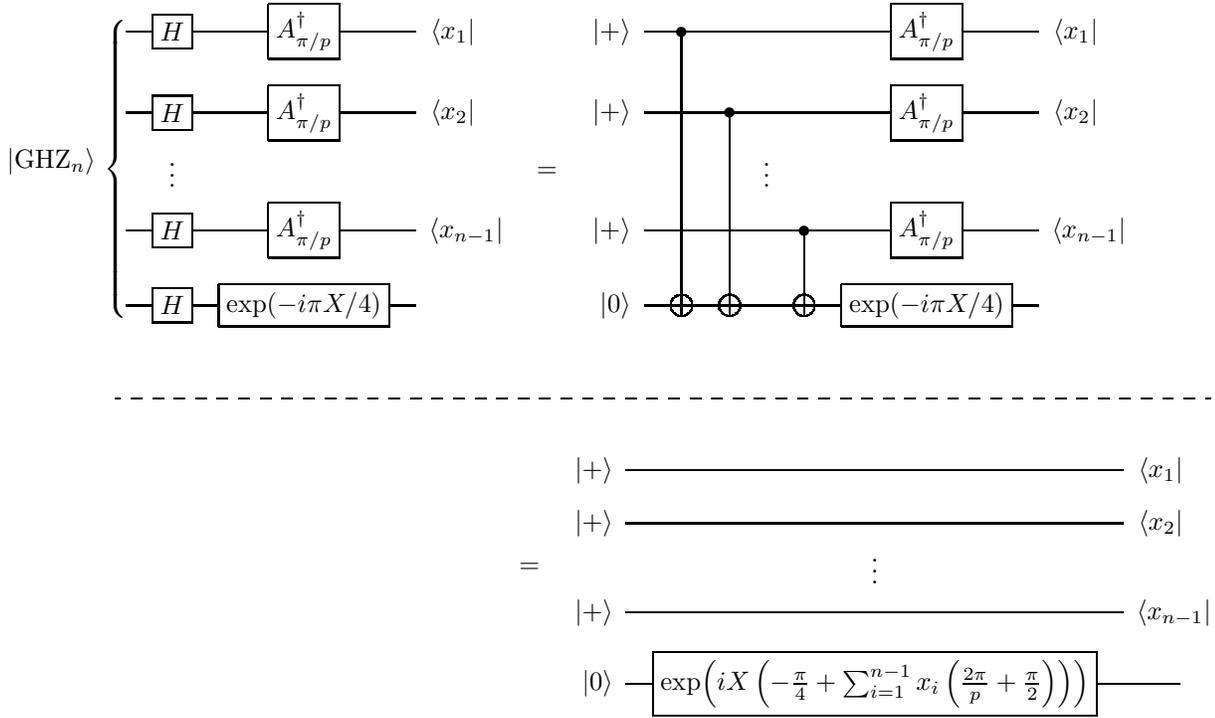
\begin{figure}[ht]
\begin{align*}
&\Qcircuit @C=1em @R=1em {
\lstick{}   & \gate{H}              & \gate{\NUrot{\pi/p}\adj} & \qw & \push{\hspace{-15pt}\bra{x_1}}  & & 
\lstick{\ket{+}}    & \ctrl{4}  & \qw       & \qw                   & \qw       & \gate{\NUrot{\pi/p}\adj} & \qw & \push{\hspace{-15pt}\bra{x_1}} 
\\
\lstick{}   & \gate{H}              & \gate{\NUrot{\pi/p}\adj} & \qw & \push{\hspace{-15pt}\bra{x_2}}  & &
\lstick{\ket{+}}    & \qw       & \ctrl{3}  & \qw                   & \qw       & \gate{\NUrot{\pi/p}\adj} & \qw & \push{\hspace{-15pt}\bra{x_2}} 
\\
\lstick{}   & {\makecell{\vdots\\}} &                           &     &                                 & \push{\rule{0em}{0em}=\rule{2em}{0em}} &
                     &           &           & {\makecell{\vdots\\}} &           &                           &     & 
\\ 
\lstick{}   & \gate{H}              & \gate{\NUrot{\pi/p}\adj} & \qw & \push{\!\!\!{\bra{x_{n-1}}}}    & &
\lstick{\ket{+}}    & \qw       & \qw       & \qw                   & \ctrl{1}  & \gate{\NUrot{\pi/p}\adj} & \qw & \push{\!\!\!{\bra{x_{n-1}}}}   
\\ 
\lstick{}   & \gate{H}              & \gate{\exp(-i \pi X / 4)}  & \qw &                                & &       
\lstick{\ket{0}}    & \targ     & \targ     & \qw                   & \targ     & \gate{\exp(-i \pi X / 4)}  & \qw 
\inputgroupv{1}{5}{0.8em}{4.9em}{\ket{\GHZ_n} \;\;\;\;\;\;\;\;} 
} 
\\
\\
\cdashline{1-2}
\\
&\hspace{150pt}\Qcircuit @C=1em @R=1em {
                                        & \lstick{\ket{+}}    &  \qw                      &  \qw & \push{\hspace{-15pt}\bra{x_1}} \\
                                        & \lstick{\ket{+}}    &  \qw                      &  \qw & \push{\hspace{-15pt}\bra{x_2}} \\
\push{\rule{0em}{0em}=\rule{2em}{0em}}  &                     &  {\makecell{\vdots\\}}    &      &                                \\
                                        & \lstick{\ket{+}}    &  \qw                      &  \qw & \push{\!\!\!{\bra{x_{n-1}}}}   \\ 
                                        & \lstick{\ket{0}}    &  \gate{\exp(iX \left(-\frac{\pi}{4} + \sum_{i=1}^{n-1} x_i \left(\frac{2 \pi}{p} + \frac{\pi}{2} \right) \right))} &  \qw & \qw 
}
\end{align*}
\caption{Diagrammatic analysis of the circuit presented in the proof of \Cref{thm:NUrot_majmod_sampling}. The first line follows from \Cref{eq:Hadarmarded_GHZ}, while the second follows from \Cref{lem:NUrot_identity}.}
\label{fig:majmod_p_NUrot_circuit_analysis}
\end{figure}

Now, tracing over the final qubit we see the probability of the first $n-1$ qubits being measured in any computational basis state $\ket{x}$ is $2^{-(n-1)}$ so the measurement of the first $n-1$ bits produces a uniformly random bit string, as desired. Additionally, conditioning on bit string $x = x_1x_2...x_{n-1}$ being measured, we see the state of the $n$-th qubit is 
\begin{align}
    &\exp(i X \left(-\frac{\pi}{4} + |x| \left(\frac{\pi}{p} + \frac{\pi}{2}\right)\right) )\ket{0}_n \label{eq:NU_analysis4}\\
    &= \exp(i X \left(-\frac{\pi}{4} + \frac{\pi}{p} |x| \right) ) \ket{\parity(x)}_n \label{eq:NU_analysis5}\\
    &= \cos(-\frac{\pi}{4} + \frac{\pi}{p} |x|) \ket{\parity(x)}_n + i\sin(-\frac{\pi}{4} + \frac{\pi}{p} |x|) \ket{1 \oplus \parity(x)}_n. \label{eq:NU_analysis6}
\end{align}
Where $|x| = \sum_{i=1}^{n-1} x_i$ denotes the Hamming weight of $x$.

Now let $Y_x$ be the random variable giving the outcome of a computational basis measurement performed on the $n$-th qubit, conditioned on a computational basis measurement of the first $n-1$ bits giving outcome $x$.  We bound the probability that this random variable does not equal $\parity(x) \oplus \majmod{p}(x)$.  Straightforward calculation gives that the probability that $Y_x$ equals $\parity(x)$ is given by
\begin{align}
    \Pr[Y_x = \parity(x)] = \cos^2\left(-\frac{\pi}{4} + \frac{\pi}{p} |x|\right).
\end{align}
It is then easy to see (see \Cref{fig:prob_parity_vs_mm}) that this function is inversely correlated with $\majmod{p}(x)$ (meaning that $Y_x$ more likely equals $\parity(x)$ when $\majmod{p}(x) = 0$ and likely does not equal $\parity(x)$ when $\majmod{p} = 1$). Expanding on this we can bound the average probability that $Y_x$ does not equal $\parity(x) \oplus \majmod{p}(x)]$:
\begin{align}
    \frac{1}{2^{n^{-1}}} \sum_{x \in \{0,1\}^{n-1}} \Pr[Y_x  \neq \parity(x) \oplus \majmod{p}(x)] \leq \frac{1}{2} - \frac{1}{\pi} + \frac{1}{2p} +  O(p^{3/2}e^{-n/p^2})
\end{align}
Details of this calculation are given after this proof, in \Cref{claim:ideal_circuit_prob_failure}. 

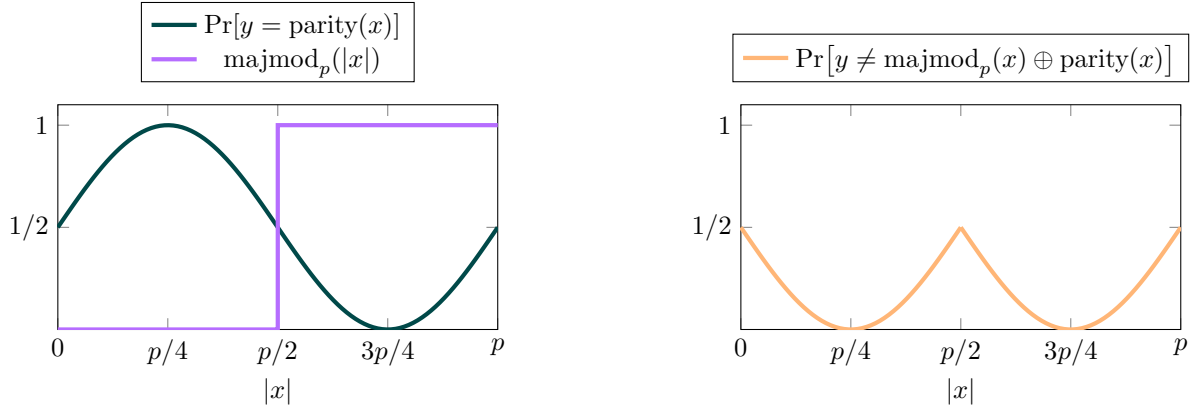
\begin{figure}[htb]
    \centering
    \begin{subfigure}[t]{0.45\textwidth}
    \begin{tikzpicture}
        \begin{axis}[
            xlabel={$|x|$},
            width=\textwidth,
            height=0.2\textheight,
            xmin=0, xmax=1,
            ymin=0, ymax=1.1,
            xtick={0.0, 0.25, 0.5, 0.75, 1.0},
            ytick={1/2, 1},
            xticklabels={$0$, $p/4$, $p/2$, $3p/4$, $p$},
            yticklabels={$1/2$, $1$},
            legend style={at={(0.5,1.1)}, anchor=south},
        ]
        \addplot [
            domain=0:1, 
            samples=100, 
            color = cb-blue-green,
            style =  ultra thick,
            ]
            {cos(deg(pi * x - pi/4))^2};
        \addlegendentry{\(\Pr[y = \parity(x)] \)}

        \addplot[
            color=cb-lilac,
            style =  ultra thick,
            ]
            coordinates {
            (0,0)(1/2, 0)(1/2,1)(1,1)
            };
        \addlegendentry{$\majmod{p}(|x|)$}
        \end{axis}
        \end{tikzpicture}
        \caption{Inverse correlation of $\Pr[Y_x = \parity(x)]$ and~$\majmod{p}(x)$}
        \label{subfig:correlation_graph_NU}
    \end{subfigure}
    \hfill
    \begin{subfigure}[t]{0.45\textwidth}
        \begin{tikzpicture}
            \begin{axis}[
                width=\textwidth,
                height=0.2\textheight,
                xlabel={$|x|$},
                xmin=0, xmax=1,
                ymin=0, ymax=1.1,
                xtick={0.0, 0.25, 0.5, 0.75, 1.0},
                ytick={1/2, 1},
                xticklabels={$0$, $p/4$, $p/2$, $3p/4$, $p$},
                yticklabels={$1/2$, $1$},
                legend style={at={(0.5,1.1)}, anchor=south},
            ]
            \addplot [
                domain=0:1/2, 
                samples=50, 
                color = cb-salmon-pink,
                style =  ultra thick,
                ]
                {1 - cos(deg(pi * x - pi/4))^2};
            \addlegendentry{$\Pr[y\neq \majmod{p}(x)\oplus \parity(x)]$}
    
            \addplot [
                domain=1/2:1, 
                samples=50, 
                color = cb-salmon-pink,
                style =  ultra thick,
                ]
                {cos(deg(pi * x - pi/4))^2};
            \end{axis}
            \end{tikzpicture}
            \caption{Probability that $Y_x$ is incorrect, $f(|x|)$}
            \label{subfig:prob_incorrect_for_x_ideal}
    \end{subfigure}
        \caption{Plots displaying the correlation of $Y_x$ and $\majmod{p}(x) \oplus \parity(x)$ where $Y_x$ is the last bit output by the circuit in \Cref{fig:majmod_p_NUrot_circuit} conditioned on the first $n-1$ measurements resulting in string $x\in \{0,1\}^{n-1}$.}
        \label{fig:prob_parity_vs_mm}
\end{figure}

Finally, we bound the total variation distance between the output of the quantum circuit depicted in \Cref{fig:majmod_p_NUrot_circuit} and the distribution $(X, \majmod{p}(X) \oplus \parity(X))$ with uniformly random $X$. Let $(X',Y)$ be the random variable giving the output of the quantum circuit. Then
    \begin{align}
        &\TVD((X, \majmod{p}(X) \oplus \parity(X)), (X',Y)) \nonumber \\
        &\hspace{20pt}= \frac{1}{2} \sum_{\substack{x \in \{0, 1\}^{n-1}\\ y \in \{0,1\}}} \Big\vert  \Pr[(X,\majmod{p}(X) \oplus \parity(X)) = (x,y)] - \Pr[(X',Y) = (x, y)]\Big\vert\\
        &\hspace{20pt}= \frac{1}{2} \sum_{\substack{x \in \{0, 1\}^{n-1}\\ y \in \{0,1\}}} \Big\vert \Pr[X = x] \Pr[\majmod{p}(x) \oplus \parity(x) = y] - \Pr[X' = x]\Pr[Y_x =  y]\Big\vert\\
        &\hspace{20pt}= \frac{1}{2^n} \sum_{\substack{x \in \{0, 1\}^{n-1}\\ y \in \{0,1\}}} \Big\vert  \Pr[\majmod{p}(x) \oplus \parity(x) = y] - \Pr[Y_x = y]\Big\vert\\
        &\hspace{20pt}= \frac{1}{2^{n-1}} \sum_{x \in \{0, 1\}^{n-1}}  \Pr[Y_x \neq \majmod{p}(x) \oplus \parity(x)] \leq \frac{1}{2} - \frac{1}{\pi} + \frac{1}{2p} +  O(p^{3/2}e^{-n/p^2})
    \end{align}
This completes the proof. 
\end{proof}

\begin{lem}\label{claim:ideal_circuit_prob_failure}
    Define the random variable $Y_x$ as in the proof of \Cref{thm:NUrot_majmod_sampling}, so $Y_x$ takes values in $\{0,1\}$ and 
    \begin{align}
    \Pr[Y_x = \parity(x)] = \cos^2\left(-\frac{\pi}{4} + \frac{\pi}{p} |x|\right). \label{eq:prob_icircuit_correct_for_x}
    \end{align}
    Then
    \begin{align}
        2^{-(n-1)} \sum_{x \in \{0, 1\}^{n-1}}  \Pr[Y_x \neq \majmod{p}(x) \oplus \parity(x)] \leq \frac{1}{2} - \frac{1}{\pi} + \frac{1}{2p} +  O(p^{3/2}e^{-n/p^2}).
    \end{align}
\end{lem}
\begin{proof}
    Let $X$ be a random variable taking value uniformly at random from $\{0,1\}^{n-1}$. Then we have
    \begin{align}\label{eq:prob_Y_incorrect_initial_eq}
        &2^{-(n-1)} \sum_{x \in \{0, 1\}^{n-1}}  \Pr[Y_x \neq \majmod{p}(x) \oplus \parity(x)] \nonumber \\
        &\hspace{80pt}= \sum_{k=0}^{p-1} \Pr[Y_X \neq \majmod{p}(X) \oplus \parity(X) \big\vert |X| = k] \cdot \Pr[|X| = k] 
    \end{align}
    Let $f(k)$ be the probability that our output measurement is incorrect given that the Hamming weight of the first $n$ bits have Hamming weight $k$.
    \begin{align}
        f(k) := \Pr[Y \neq \majmod{p}(X) \oplus \parity(X) \big| |X| = k]
    \end{align}
    It follows from \Cref{eq:prob_icircuit_correct_for_x}, that
    \begin{align}
        f(k) =
        \begin{cases}
            \sin^2\left(-\frac{\pi}{4} + \frac{\pi}{p} k \right), & k \leq p/2   \mod p\\
            \cos^2\left(-\frac{\pi}{4} + \frac{\pi}{p} k\right), &  k > p/2 \mod p
        \end{cases}
    \end{align}
    which is plotted in \Cref{subfig:prob_incorrect_for_x_ideal}.
    Let $\delta$ be the total variation distance between $|X| \mod p$ and $U_p$, the uniform distribution over $\{0,1,\dots,p-1\}$. Then $\Pr[|X| = k \mod p] \leq \frac{1}{p} + \delta$. We can upper bound \Cref{eq:prob_Y_incorrect_initial_eq}, as 
    \begin{align}
        \Pr[Y \neq \majmod{p}(X) \oplus \parity(X)] &\leq \left(\frac{1}{p} + \delta\right) \sum_{k=0}^{p-1} f(k)\\
        &=  \left(\frac{1}{p} + \delta\right)\left(\frac{1}{2} + 2 \sum_{k = 1}^{(p-1)/2} f(k) \right)\\
        &=  \left(\frac{1}{p} + \delta\right)\left(\frac{1}{2} + 2\int_{1/2}^{p/2} f(k) \right)  \, dk \label{eq:prob_incorrect_as_integral}
    \end{align}
    Where in the second line we use the fact that $f(k)$ is symmetric about $p/2$, so $\sum_{k=1}^{\frac{p-1}{2}} f(k) = \sum_{k=\frac{p+1}{2}}^{p-1} f(k)$. In the third line we used that $f(k)$ is convex over $(0, p/2)$, and therefore $\sum_{i=1}^{(p-1)/2} f(k)$ is a (midpoint-Riemann sum) over-approximation of $\int_{1/2}^{p/2} f(k)$. Next, we evaluate the integral.
    \begin{align}
        \int_{1/2}^{p/2} f(k) \, dk &= \int_{0}^{p/2} \sin^2\left(-\frac{\pi}{4} +\frac{\pi}{p}k\right) \, dk \\
        &= \int_{0}^{p/2} \frac{1}{2}\left(1 + \cos\left(\frac{2\pi}{p}k + \frac{\pi}{2}\right)\right)\, dk\\
        &= \frac{1}{2} \left.\left(k + \frac{p}{2\pi}\sin\left(\frac{2\pi}{p}k + \frac{\pi}{2} \right) \right) \right\vert_{0}^{p/2}\\
        &= \frac{p}{4}\left(1 - \frac{2}{\pi}\right) 
    \end{align}
    Combining this with \Cref{eq:prob_incorrect_as_integral}, we get the probability we measure an incorrect string is at most
    \begin{align}
        \Pr[Y \neq \majmod{p}(X) \oplus \parity(X)] &\leq \left(\frac{1}{p} + \delta\right)\left(\frac{p}{2}\left(1 - \frac{2}{\pi}\right)+ \frac{1}{2} \right)\\
        &=  \frac{1}{2} - \frac{1}{\pi} + \frac{\delta p}{2}\left(1 - \frac{2}{\pi}\right) + \frac{1}{2}\left(\frac{1}{p} + \delta\right)\\
        &= \frac{1}{2} - \left(\frac{1}{\pi} - \frac{1}{2p}\right) + O(p\delta)
    \end{align}
    All that's left is to upper bound $\delta$, the total variation distance between $|X| \mod p$ and $U_p$. For this, we use the following Fact from \cite{viola2012complexity}.
    \begin{fact} [special case of Fact 3.2 in \cite{viola2012complexity}]\label{fact:puniform_bits}
        Let $(x_1, x_2, \dots, x_t) \in \{0,1\}^n$ be sampled uniformly. Then the total variation distance between $\sum_{i=1}^t x_i \mod p$ and $U_{p}$, the uniform distribution over $\{0, 1, \dots, p-1\}$ is at most $\sqrt{p} e^{-t/p^2}$
    \end{fact}
    Using this fact, we get the upper bound $\delta \leq p^{1/2}e^{{-n/p^2}}$. The probability the measured string is incorrect is then
    \begin{align}
        \Pr[Y \neq \majmod{p}(X) \oplus \parity(X)] \leq \frac{1}{2} - \frac{1}{\pi} + \frac{1}{2p} +  O(p^{3/2}e^{-n/p^2}).
    \end{align}
\end{proof}

\subsection{Removing non-unitary operations}
\label{ssec: unitary majority GHZ sampling}

We now construct a fully quantum circuit that takes a $\GHZ$ state as input and produces a state which, when measured in the computational basis, samples approximately from the distribution $(X, \majmod{p}(X) \oplus \parity(X))$. Our starting point is the non-unitary circuit constructed in \Cref{ssec: non-unitary majority GHZ sampling}. First, we modify this circuit by replacing the non-unitary $\NUrot{\theta}$ gates with a different set of non-unitary gates and show the classical distributions output by the two circuits after measurement are identical. Then we show these new non-unitary gates are close to unitary gates, and hence the circuit can be made fully unitary with minimal change to the output distribution. 

\subsubsection{Introducing multi-qubit non-unitary operations}
\label{sssec:multi_qubit_NU_sampling_circuit}

We start by defining the $m$-qubit non-unitary operation $\multiNUrot{\theta}{m}$ whose action on the $m$ qubit basis state $\ket{x} = \ket{x_1x_2...x_m}$ is given by:
\begin{align}
    \multiNUrot{\theta}{m}\ket{x_1x_2...x_m} = \exp(i\theta x_m) \ket{x_1} \otimes \exp(i \theta x_1) \ket{x_2} \otimes ... \otimes \exp(i \theta x_{m-1}) \ket{x_m}.
\end{align}
Intuitively, we can think of the $\multiNUrot{\theta}{m}$ operation as consisting of $m$ distinct $\NUrot{\theta}$ operations, just with the qubits they act on ``shifted'' away from the qubits controlling the gate by 1 modulo $m$. 

Now we observe that, in certain situations, an $\multiNUrot{\theta}{m}$ operation can replace a tensor product of $m$ different $\NUrot{\theta}$ operations. 
\begin{lem}
\label{lem:multi_NUrot-NUrot_circuit_equivalence}
For any $m$-qubit computational basis state $\ket{x} = \ket{x_1x_2...x_m}$ and arbitrary one qubit state $\ket{\psi}$, the following equivalence holds: 
\begin{align}
    &\bra{x}_{1 ... m} \left(\multiNUrot{\theta}{m} \adj \right)_{1...m} \left(\prod_{i=1}^m \CNOT_{i, m+1}\right) \ket{+}^{\otimes m} \otimes \ket{\psi} \nonumber \\
    &\hspace{80pt}= \bra{x}_{1 ... m} \left(\prod_{i=1}^m \left(\NUrot{\theta}\adj\right)_i \CNOT_{i, m+1}\right) \ket{+}^{\otimes m} \otimes \ket{\psi}
\end{align}
\end{lem}

\begin{proof}
The proof is similar to the proof of \Cref{lem:NUrot_identity}. In what follows we identify indices mod $m$ so, in particular, we have $x_0 = x_m$. Then we see: 
\begin{align}
    &\bra{x}_{1 ... m} \left(\multiNUrot{\theta}{m} \adj \right)_{1...m} \left(\prod_{j=1}^m \CNOT_{j, m+1}\right) \ket{+}^{\otimes m} \otimes \ket{\psi} \nonumber \\
    &\hspace{80pt}=\bra{x}_{1 ... m} \left(\prod_{j=1}^m \exp(i \theta X_j x_{j-1} ) \CNOT_{j, m+1}\right) \ket{+}^{\otimes m} \otimes \ket{\psi} \\
    &\hspace{80pt}= \bra{x}_{1...m} \left(\prod_{j=1}^m \CNOT_{j, m+1} 
    \exp(i \theta X_j X_{m+1} x_{j-1})
    \right) \ket{+}^{\otimes m} \otimes \ket{\psi} \\
    &\hspace{80pt}= \bra{x}_{1...m} 
    \left(\prod_{j=1}^m \CNOT_{j, m+1} 
    \right) \ket{+}^{\otimes m} \otimes \exp(i \theta X \sum_{j=1}^m x_{j-1}) \ket{\psi} \\
    &\hspace{80pt}= \bra{x}_{1...m} 
    \left(\prod_{j=1}^m \CNOT_{j, m+1} 
    \right) \ket{+}^{\otimes m} \otimes \exp(i \theta X \sum_{j=1}^m x_{j}) \ket{\psi} \\
    &\hspace{80pt}= \bra{x}_{1...m} 
    \left(\prod_{j=1}^m \exp(i \theta X_j x_{j}) \CNOT_{j, m+1} 
    \right) \ket{+}^{\otimes m} \otimes \ket{\psi} \\
    &\hspace{80pt}= \bra{x}_{1...m} 
    \left(\prod_{j=1}^m \left(\NUrot{\theta}\adj\right)_j \CNOT_{j, m+1} 
    \right) \ket{+}^{\otimes m} \otimes \ket{\psi}. 
\end{align}
Here the first line follows from the definition of $\multiNUrot{\theta}{m}$, the second line follows from commuting an $\exp(i \theta X)$ gate past a $\CNOT$ gate as in the proof of \Cref{lem:NUrot_identity}, the third line follows because $\ket{+}$ is a 1 eigenstate of the $X$ operator and the fourth line follows from a simple relabeling of indices. The fifth line follows from applying the same argument as in the second and third lines, just in the reverse direction, and the sixth line follows by definition of~$\NUrot{\theta}$. \Cref{fig:multi_NUrot-NUrot_circuit_equivalence} gives a diagrammatic version of this proof. 
\end{proof}

\begin{figure}[htp]
\begin{align*}
&\Qcircuit @C=1em @R=1em {
\lstick{\ket{+}}    & \ctrl{4}  & \qw       & \qw                   & \qw       & \multigate{3}{\multiNUrot{\theta}{m}\adj} & \qw & \push{\!\!\!\bra{x_1}}
&&
\lstick{\ket{+}}    & \ctrl{4}  & \qw       & \qw                   & \qw       & \gate{\exp(i\theta X x_{m})}         & \qw & \push{\!\!\!\bra{x_1}}
\\
\lstick{\ket{+}}    & \qw       & \ctrl{3}  & \qw                   & \qw       & \ghost{\multiNUrot{\theta}{m}\adj}        & \qw & \push{\!\!\!\bra{x_2}} 
&&
\lstick{\ket{+}}    & \qw       & \ctrl{3}  & \qw                   & \qw       & \gate{\exp(i\theta X x_{1})}         & \qw & \push{\!\!\!\bra{x_2}} 
\\
                    &           &           & {\makecell{\vdots\\}} &           &                                           &     & 
&\push{\rule{0em}{0em}=\rule{2em}{0em}} &
                    &           &           & {\makecell{\vdots\\}} &           &                                      &     & 
\\
\lstick{\ket{+}}    & \qw       & \qw       & \qw                   & \ctrl{1}  & \ghost{\multiNUrot{\theta}{m}\adj}        & \qw & \push{\!\!\!\bra{x_{m}}}   
&&
\lstick{\ket{+}}    & \qw       & \qw       & \qw                   & \ctrl{1}  & \gate{\exp(i\theta X x_{m-1})}       & \qw & \push{\!\!\!\bra{x_{m}}}   
\\
\lstick{\ket{\psi}}    & \targ     & \targ     & \qw                   & \targ     & \qw                                       & \qw &
&&
\lstick{\ket{\psi}}    & \targ     & \targ     & \qw                   & \targ     & \qw                                  & \qw &
} 
\\
\\
\cdashline{1-2}
\\
&\hspace{150pt}\Qcircuit @C=1em @R=1em {
&\lstick{\ket{+}}                                               & \qw & \ctrl{4}  & \qw       & \qw                   & \qw       & \qw & \push{\!\!\!\bra{x_1}}
\\
&\lstick{\ket{+}}                                               & \qw & \qw       & \ctrl{3}  & \qw                   & \qw       & \qw & \push{\!\!\!\bra{x_2}} \\
\push{\rule{0em}{0em}=\rule{2em}{0em}} 
&                                                               &     &           &           & {\makecell{\vdots\\}} &           &     & 
\\
&\lstick{\ket{+}}    & \qw                                      & \qw       & \qw       & \qw                   & \ctrl{1}  & \qw & \push{\!\!\!\bra{x_{m}}}   
\\ 
&\lstick{\ket{\psi}}    & \gate{\exp(i \theta X \sum_{j=1}^m x_j)} & \targ     & \targ     & \qw                   & \targ     & \qw &
}
\\
\\
\cdashline{1-2}
\\
&\hspace{150pt}\Qcircuit @C=1em @R=1em {
&\lstick{\ket{+}}    & \ctrl{4}  & \qw       & \qw                   & \qw       & \gate{\exp(i\theta X x_{1})}         & \qw & \push{\!\!\!\bra{x_1}}
\\
&\lstick{\ket{+}}    & \qw       & \ctrl{3}  & \qw                   & \qw       & \gate{\exp(i\theta X x_{2})}         & \qw & \push{\!\!\!\bra{x_2}} \\
\push{\rule{0em}{0em}=\rule{2em}{0em}} 
&                    &           &           & {\makecell{\vdots\\}} &           &                                      &     & 
\\
&\lstick{\ket{+}}    & \qw       & \qw       & \qw                   & \ctrl{1}  & \gate{\exp(i\theta X x_{m})}       & \qw & \push{\!\!\!\bra{x_{m}}}   
\\
&\lstick{\ket{\psi}}    & \targ     & \targ     & \qw                   & \targ     & \qw                                  & \qw &
}
\\
\\
\cdashline{1-2}
\\
&\hspace{150pt}\Qcircuit @C=1em @R=1em {
&\lstick{\ket{+}}    & \ctrl{4}  & \qw       & \qw                   & \qw       & \gate{\NUrot{\theta}\adj}         & \qw & \push{\!\!\!\bra{x_1}}
\\
&\lstick{\ket{+}}    & \qw       & \ctrl{3}  & \qw                   & \qw       & \gate{\NUrot{\theta}\adj}         & \qw & \push{\!\!\!\bra{x_2}} \\
\push{\rule{0em}{0em}=\rule{2em}{0em}} 
&                    &           &           & {\makecell{\vdots\\}} &           &                                      &     & 
\\
&\lstick{\ket{+}}    & \qw       & \qw       & \qw                   & \ctrl{1}  & \gate{\NUrot{\theta}\adj}       & \qw & \push{\!\!\!\bra{x_{m}}}   
\\
&\lstick{\ket{\psi}}    & \targ     & \targ     & \qw                   & \targ     & \qw                                  & \qw &
}
\end{align*}
\caption{Diagrammatic proof of \Cref{lem:multi_NUrot-NUrot_circuit_equivalence}. $\ket{\psi}$ is an arbitrary single qubit state. The equivalence between lines is explained in the proof of the lemma.}
\label{fig:multi_NUrot-NUrot_circuit_equivalence}
\end{figure}
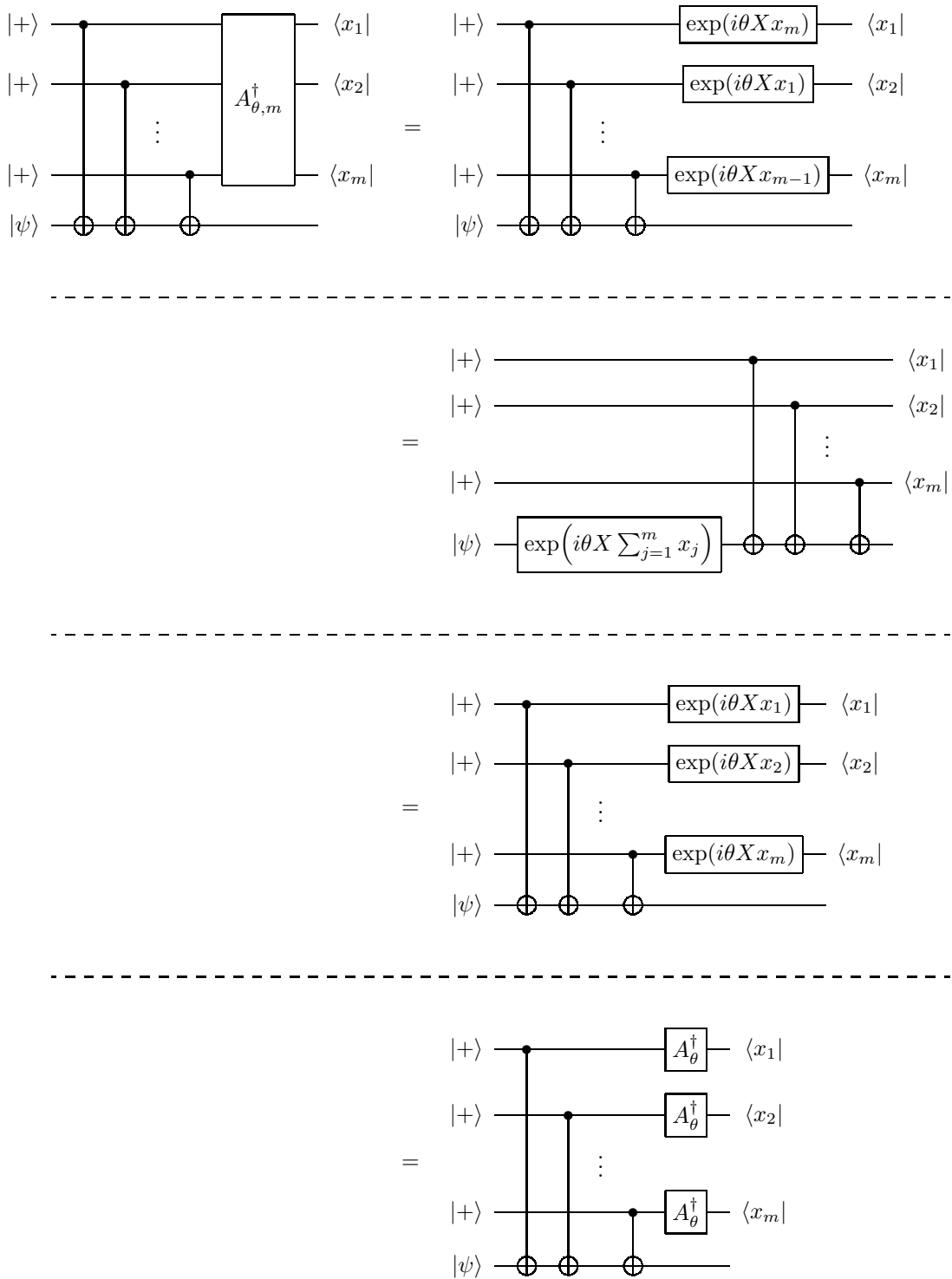

A straightforward consequence of \Cref{lem:multi_NUrot-NUrot_circuit_equivalence} and the arguments of \Cref{ssec: non-unitary majority GHZ sampling} is that constant-depth quantum circuits augmented with $\multiNUrot{\theta}{m}$ gates and acting on a GHZ state can also approximately sample from the distribution $(X, \majmod{p}(X) \oplus \parity(X))$.

\begin{cor}

\label{cor:multiNUrot_majmod_sampling}
Let $m$ and $D$ be integers, and $n = Dm + 1$. Then the state 
\begin{align}
    \left(\left(\multiNUrot{\pi/p}{m}\adj \right)^{\otimes D} \otimes \exp(-i \pi X/4)\right) H^{\otimes n} \ket{GHZ_n},
\end{align}
when measured in the computational basis, produces an output distribution $(X', Y)$ with 
\begin{align}
    \TVD((X', Y), (X, \majmod{p}(X) \oplus \parity(X))) \leq  \frac{1}{2} - \frac{1}{\pi} + \frac{1}{2p} +  O(p^{3/2}e^{-n/p^2}).
\end{align}
\end{cor}

\begin{proof}
By \Cref{lem:multi_NUrot-NUrot_circuit_equivalence} and \Cref{eq:Hadarmarded_GHZ} we have 
\begin{align}
    &\left(\left(\multiNUrot{\pi/p}{m}\adj \right)^{\otimes D} \otimes \exp(-i \pi X/4)\right) H^{\otimes n} \ket{GHZ_n} \nonumber \\
    &\hspace{80pt}= 
     \left(\left(\multiNUrot{\pi/p}{m}\adj \right)^{\otimes D} \otimes \exp(-i \pi X/4)\right) \left(\prod_{i=1}^{n-1} \CNOT_{i,n} \right)\ket{+}^{\otimes n-1} \otimes \ket{0} \\
    &\hspace{80pt}= 
     \left(\left(\NUrot{\pi/p}\adj \right)^{\otimes n-1} \otimes \exp(-i \pi X/4)\right) \left(\prod_{i=1}^{n-1} \CNOT_{i,n} \right)\ket{+}^{\otimes n-1} \otimes \ket{0} \\
     &\hspace{80pt}= 
     \left(\left(\NUrot{\pi/p}\adj \right)^{\otimes n-1} \otimes \exp(-i \pi X/4)\right) H^{\otimes n} \ket{GHZ_n}
\end{align}
In the proof of \Cref{thm:NUrot_majmod_sampling} we show this state, when measured in the computational basis, is close to the distribution $(X, \majmod{p}(X) \oplus \parity(X))$.
\end{proof}

\subsubsection{Replacing multi-qubit non-unitary operations with unitary operations}

In this section, we construct a fully unitary circuit which takes a GHZ state as input and produces an output which, when measured in the computation basis, samples for a distribution close in Total Variation Distance to the distribution $(X, \majmod{p}(X) \oplus \parity(X))$. We do this by proving that we can replace the non-unitary operations $\multiNUrot{m}{\theta}$ introduced in the previous section with unitary operations while causing minimal change to a circuit using these elements. 

To make these statements formal, we first recall some definitions and useful standard facts about matrix norms. 

\begin{defn}
The Frobenius norm of a matrix $M$, denoted $\norm{M}_F$, is defined by
\begin{align}
    \norm{M}_F = \sqrt{\tr[M^*M]}
\end{align}
\end{defn}

\begin{defn}
The infinity (or operator) norm of a matrix M, denoted $\norm{M}_\infty$,
is defined by
\begin{align}
    \norm{M}_\infty = \max_{\ket{\psi} : \norm{\ket{\psi}} = 1} \norm{M \ket{\psi}},
\end{align}
where $\norm{\ket{\psi}}$ denotes the regular Euclidean norm of any vector $\ket{\psi}$.
\end{defn}

\begin{fact} \label{fact:Frobenius_Operator_bound}
For any matrix $M$, the Frobenius norm upper bounds the operator norm
\begin{align}
    \norm{M}_\infty \leq \norm{M}_F.
\end{align}
\end{fact}

\begin{proof}
For an arbitrary matrix $M$, let $\lambda_1, ..., \lambda_d$ denote the eigenvalues of $M^* M$, with $\lambda_1 \geq \lambda_2 \geq ... \lambda_d$. Note all $\lambda_i$ are positive. Then we have 
\begin{align}
\norm{M}_\infty^2 = \lambda_1 \leq \sum_{i=1}^d \lambda_i = \norm{M}_F^2
\end{align} 
as desired. \qedhere
\end{proof}

\begin{fact}
\label{fact:Operator_Tensor_Product_Bound}
Given matrices $A_1, A_2, ... A_s$ and $B_1, B_2, ... , B_s$ with 
\begin{align}
    \norm{A_i - B_i}_\infty &\leq \epsilon, \\
    \norm{A_i}_\infty &\leq 1
\end{align}
for all $i \in [s]$, and 
\begin{align}
    s \epsilon < 1,
\end{align}
we also have 
\begin{align}
    \norm{ \bigotimes_{i \in [s]} A_i - \bigotimes_{i \in [s]} B_i}_\infty \leq 2s\epsilon. 
\end{align}
\begin{proof}
First note that $\norm{M}_\infty$ is equal to the largest singular value of the matrix $M$, from which it follows that 
\begin{align}
\norm{M \otimes N}_\infty = \norm{M}_\infty \norm{N}_\infty
\end{align}
for any matrices $M$ and $N$. Then an inductive argument gives
\begin{align}
        \norm{ \bigotimes_{i = 1}^s A_i - \bigotimes_{i = 1}^s B_i}_\infty &= \norm{ \bigotimes_{i = 1^s} A_i - B_1 \bigotimes_{i = 2}^s A_i +  B_1 \bigotimes_{i = 2}^s A_i - \bigotimes_{i = 1}^s B_i}_\infty \\ 
        &\leq \norm{\left(A_1 - B_1\right)\bigotimes_{i = 2}^s A_i} + \norm{B_1 \otimes \left(\bigotimes_{i = 2}^s A_i - \bigotimes_{i = 2}^s B_i \right)} \\
        &\leq \epsilon + (1+\epsilon)\norm{\bigotimes_{i = 2}^s A_i - \bigotimes_{i = 2}^s B_i} \\
        &= \epsilon + (1+ \epsilon)(2\epsilon(s-1)) \leq 2s\epsilon
\end{align}
as desired.
\end{proof}
\end{fact}

\begin{fact} \label{fact:l2_distance_trace_distance_bound}
Given two states $\ket{\rho}$ and $\ket{\sigma}$, let $p(x)$ and $q(x)$ denote the resulting classical distributions when $\ket{\rho}$ and $\ket{\sigma}$ are measured in some basis $\{\ket{x}\}$. Then we have 
\begin{align}
    \sum_x \abs{p(x) - q(x)} \leq  4 \norm{\ket{\rho} - \ket{\sigma}}
\end{align}
\end{fact}

\begin{proof}
First, we note that for any two states $\ket{\rho}$ and $\ket{\sigma}$ and PSD matrix $M \leq I$ we have 
\begin{align}
    2 \norm{\ket{\rho} - \ket{\sigma}} &\geq 2 \norm{M (\ket{\rho} - \ket{\sigma})} \\
    &\geq 2\left(\norm{M \ket{\rho}} - \norm{M \ket{\sigma}}\right) \\
    &\geq \left(\norm{M \ket{\rho}} - \norm{M \ket{\sigma}}\right)\left(\norm{M \ket{\rho}} + \norm{M \ket{\sigma}}\right) \\
    &= \norm{M \ket{\rho}}^2 - \norm{M \ket{\sigma}}^2 
\end{align}
Then defining probability distributions $p(x)$ and $q(x)$ and the basis $\{\ket{x}\}$ as above, let 
\begin{align}
    P_x := \{x : p(x) \geq q(x)\}
\end{align}
and 
\begin{align}
    M_x = \sum_{x \in P_x} \ketbra{x}.
\end{align}
Then note 
\begin{align}
    \norm{M_x \ket{\rho}}^2 - \norm{M_x \ket{\sigma}}^2 &= \sum_{x \in P_x} \abs{\braket{x}{\rho}}^2 - \abs{\braket{x}{\sigma}}^2 \\
    &= \sum_{x \in P_x} (p(x) - q(x)) \\
    &= \frac{1}{2} \sum_x \abs{p(x) - q(x)} 
\end{align}
with the final inequality holding because both $p(x)$ and $q(x)$ must sum to one. Combining the two inequalities above proves the result. 
\end{proof}
Next, we recall the definition of the matrix $\multiNUrot{m}{\theta}$ in terms of its action on computational basis states. 
\begin{align}
    \multiNUrot{m}{\theta} \ket{x_1 x_2 ... x_m} := \exp(i \theta X x_m) \ket{x_1} \otimes \exp(i \theta X x_1) \ket{x_2} \otimes ... \otimes \exp(i \theta X x_{m-1}) \ket{x_m}.
\end{align}
The matrix $\multiNUrot{m}{\theta}$ would be a unitary matrix iff it mapped computational basis states to some set of orthonormal basis states.\footnote{More generally it is unitary iff it maps any set of orthonormal basis states to some other orthonormal basis.} The following lemma shows that this condition is close to being satisfied. In what follows, for any bitstring $x = x_1 x_2 ... x_m \in \{0,1\}^m$ we let $\overline{x}$ denote the bitwise compliment of $x$. We also interpret all subscripts in the remainder of this section mod $m$ so, in particular, $x_0 = x_m$. 

\begin{lem}
\label{lem:multiNUrot_inner_products}
For any $\theta \in \mathbb{R}, m \in \mathbb{Z}^+$ and $x = x_1x_2...x_m \in \{0,1\}^m$ the matrix $\multiNUrot{\theta}{m}$ satisfies the following properties: 
\begin{enumerate}
    \item $\matrixel{x}{\multiNUrot{\theta}{m}\adj\multiNUrot{\theta}{m}}{x} = 1$. \label{lemitem:multiNUrot_normalized} 
    \item $\matrixel{\overline{x}}{\multiNUrot{\theta}{m}\adj\multiNUrot{\theta}{m}}{x} = -i^{m + 2\abs{x}} \sin^m(\theta)$. 
    \label{lemitem:multiNUrot_compliments}
    \item $\matrixel{y}{\multiNUrot{\theta}{m}\adj\multiNUrot{\theta}{m}}{x} = 0$ for any $y \in \{0,1\}^m \backslash \{ \overline{x}, x\} $.  
    \label{lemitem:multiNUrot_orthogonality}
\end{enumerate}
\end{lem}

\begin{proof}
The proof of \Cref{lemitem:multiNUrot_normalized,lemitem:multiNUrot_compliments} are purely computational. For any $x = x_1x_2...x_m \in \{0,1\}^m$ we have
\begin{align}
      \bra{x} \multiNUrot{m}{\theta}\adj \multiNUrot{m}{\theta} \ket{x} 
    &= \prod_{j \in [m]} \bra{x_j} \exp(-i \theta x_{j-1}) 
    \exp(i \theta x_{j-1})
    \ket{x_j} \\
    &= \prod_{j \in [m]} \braket{x_j} = 1,
\end{align}
proving \Cref{lemitem:multiNUrot_normalized}. A similar calculation gives 
\begin{align}
    \matrixel{\overline{x}}{\multiNUrot{m}{\theta}\adj \multiNUrot{m}{\theta}}{x} 
    &= \prod_{j \in [m]} \matrixel{\overline{x}_j}{
    \exp(-i \theta X \overline{x}_j) 
    \exp(i \theta X x_j)}{x_j} \\
    &= \prod_{j \in [m]} \matrixel{\overline{x}_j}{
    \exp(i^{1 + 2\overline{x}_j} \theta X)}{x_j} \\
    &= \prod_{j \in [m]} \matrixel{\overline{x}_j}{\cos(\theta) + i^{1 + 2\overline{x}_j} \sin(\theta) X}{x_j} \\
    &= \prod_{j \in [m]} i^{1 + 2\overline{x}_j} \sin(\theta) \\
    &= i^{m + 2\abs{\overline{x}}} \sin^m(\theta) \\
    &= -i^{m + 2\abs{x}} \sin^m(\theta),
\end{align}
where we used that $X \ket{\overline{x}_j} = \ket{x_j}$ by definition of the compliment on the fourth line and that $\abs{\overline{x}} + \abs{x} = m$ for any $x$ in the final line. This proves \Cref{lemitem:multiNUrot_compliments}.

To prove \Cref{lemitem:multiNUrot_orthogonality} note that for any $m$ bit strings $x$ and $y$ with $x \notin \{\overline{y}, y\}$ there exists a $k \in [m]$ with $x_{k-1} = y_{k-1}$ and $x_{k} \neq y_{k}$. Fixing $k$ to be that value we find:
\begin{align}
    \matrixel{y}{\multiNUrot{m}{\theta}\adj \multiNUrot{m}{\theta}}{x} &= \prod_{j=1}^m \matrixel{x_j}{\exp(-i\theta X y_{j-1})\exp(i\theta X x_{j-1})}{y_j} \\
    &= \matrixel{y_{k}}{\exp(i\theta X (x_{k} - y_{k}))}{x_{k}} 
    \times \prod_{j \in [m]\backslash \{k\}} \matrixel{y_j}{\exp(i\theta X (x_{j-1} - y_{j-1}))}{x_j} \\
    &= \braket{y_{k}}{x_{k}} 
    \times \prod_{j \in [m]\backslash \{k\}} \matrixel{y_j}{\exp(i\theta X (x_{j-1} - y_{j-1}))}{x_j} \\
    &= 0 
\end{align}
since $y_k \neq x_k$ by definition. This completes the proof of \Cref{lemitem:multiNUrot_orthogonality}. 
\end{proof}

We show that, as a consequence of \Cref{lem:multiNUrot_inner_products}, there exists an $m$ qubit unitary matrix which is close (in Frobenius norm) to the non-unitary matrix $\multiNUrot{\theta}{m}$. We construct this unitary by applying Gram-Schmidt orthnomalization applied to the state's output by $\multiNUrot{m}{\theta}$ acting on computational basis states. 

\begin{lem}
\label{lem:Frobenius_Urot_NUrot_bound}
For any $m$, there exists unitary matrices $\Urot{m}{\theta}$ satisfying
\begin{align}
    \norm{\multiNUrot{m}{\theta} - \Urot{m}{\theta}}_F \in O\left(\theta^{-m}\right)
\end{align}
as $\theta \rightarrow 0$.
\end{lem}

\begin{proof}
We will define $\Urot{m}{\theta}$ by its action on computational basis states. First, fix $\halfbitstring^m$ to be any set containing half the bit strings of length $m$ with the property that for any $x \in \{0,1\}^m$ either $x \in \halfbitstring^m$ or $\overline{x} \in \halfbitstring^m$.  
(That is, $\halfbitstring^m$ contains one representative element from the equivalence classes of the set $\{0,1\}^m$ induced by the equivalence relation $x \sim y$ if $x = y$ or $\overline{x} = y$). Then define: 
\begin{align}
    \Urot{m}{\theta} \ket{x} := 
    \begin{cases}
    \multiNUrot{m}{\theta} \ket{x} &\text{ if } x \in \halfbitstring^m \\
    C^{-1}\left(\multiNUrot{m}{\theta} \ket{x} + i^{m + 2 \abs{x}}\sin^{m}(\theta)\multiNUrot{m}{\theta}\ket{\overline{x}}\right) &\text{ otherwise.}
    \end{cases}
\end{align}
with $C:= \sqrt{1 - \sin^{2m}(\theta)}$ a normalizing constant. Observe that, by \Cref{lemitem:multiNUrot_compliments} of \Cref{lem:multiNUrot_inner_products}, for $x \notin \halfbitstring^m$ we can also write 
\begin{align}
    \Urot{m}{\theta} \ket{x} &=  C^{-1}\left(\multiNUrot{m}{\theta} \ket{x} - \matrixel{\overline{x}}{\multiNUrot{m}{\theta}\adj\multiNUrot{m}{\theta}}{x} \multiNUrot{m}{\theta}\ket{\overline{x}}\right) \label{eq:Urot_action_as_inner_product}
\end{align}
and 
\begin{align}
    C = \left(1 - \abs{\matrixel{\overline{x}}{\multiNUrot{m}{\theta}\adj\multiNUrot{m}{\theta}}{x}}^2 \right)^{1/2}.
    \label{eq:Urot_normalization_as_inner_product}
\end{align}

We now prove that $\Urot{m}{\theta}$ is unitary. To do this, we prove $\Urot{m}{\theta}$ maps computational basis states to an orthonormal basis. First note that \Cref{lemitem:multiNUrot_normalized} of \Cref{lem:multiNUrot_inner_products} gives that for any $x \in \halfbitstring^m$:
\begin{align}
    \matrixel{x}{\Urot{m}{\theta}\adj \Urot{m}{\theta}}{x} = 
    \matrixel{x}{\multiNUrot{m}{\theta}\adj \multiNUrot{m}{\theta}}{x} = 1
\end{align}
while a similar calculation gives for any $x \notin \halfbitstring^m$:
\begin{align}
    \matrixel{x}{\Urot{m}{\theta}\adj \Urot{m}{\theta}}{x} &= C^{-2} \left(\bra{x} \multiNUrot{m}{\theta}\adj - \matrixel{\overline{x}}{\multiNUrot{m}{\theta}\adj\multiNUrot{m}{\theta}}{x}\adj \bra{\overline{x}}\multiNUrot{m}{\theta}\adj \right)\left(\multiNUrot{m}{\theta} \ket{x} - \matrixel{\overline{x}}{\multiNUrot{m}{\theta}\adj\multiNUrot{m}{\theta}}{x}\multiNUrot{m}{\theta}\ket{\overline{x}}\right) \\
    &= C^{-2} \left(1 -  \abs{\matrixel{\overline{x}}{\multiNUrot{m}{\theta}\adj\multiNUrot{m}{\theta}}{x}}^2 \right)  = 1\label{eq:Urot_normalization_halfway}.
\end{align}
Where we used \Cref{eq:Urot_action_as_inner_product,eq:Urot_normalization_as_inner_product} on the first and second lines, respectively. Then we see the states $\{\Urot{m}{\theta}\ket{x}\}$ for $x \in \{0,1\}^m$ acting on computational basis states are correctly normalized. 

It remains to show that these states are orthogonal. First, we note that \Cref{lemitem:multiNUrot_orthogonality} of \Cref{lem:multiNUrot_inner_products} gives that for any $x, y \in \{0,1\}^m$ with $y \notin \{x, \overline{x}\}$ we have 
\begin{align}
    \matrixel{y}{\multiNUrot{\theta}{m}\adj\multiNUrot{\theta}{m}}{x} 
    = \matrixel{\overline{y}}{\multiNUrot{\theta}{m}\adj\multiNUrot{\theta}{m}}{x}
    = \matrixel{y}{\multiNUrot{\theta}{m}\adj\multiNUrot{\theta}{m}}{\overline{x}} 
    = \matrixel{\overline{y}}{\multiNUrot{\theta}{m}\adj\multiNUrot{\theta}{m}}{\overline{x}} 
    = 0
\end{align}
and then a quick proof by cases shows that $\matrixel{y}{\Urot{\theta}{m}\adj\Urot{\theta}{m}}{x} = 0$ for any $x \in \{0,1\}^m$ and $y \notin \{x, \overline{x}\}$. Finally, we consider the inner product $\matrixel{\overline{x}}{\Urot{\theta}{m}\adj\Urot{\theta}{m}}{x}$. By definition of $\halfbitstring^m$, exactly one of $x$ or $\overline{x}$ is in $\halfbitstring^m$. Assume for the moment that $x \notin \halfbitstring^m$. Then using \Cref{eq:Urot_action_as_inner_product} we have 
\begin{align}
    \matrixel{\overline{x}}{\multiNUrot{\theta}{m}\adj\multiNUrot{\theta}{m}}{x} &= C^{-1} \left( \bra{\overline{x}} \multiNUrot{m}{\theta}\adj \right) \left(\multiNUrot{m}{\theta} \ket{x} - \matrixel{\overline{x}}{\multiNUrot{m}{\theta}\adj\multiNUrot{m}{\theta}}{x} \multiNUrot{m}{\theta}\ket{\overline{x}}\right) \\
    &= C^{-1}\left(\matrixel{\overline{x}}{\multiNUrot{m}{\theta}\adj\multiNUrot{m}{\theta}}{x} - \matrixel{\overline{x}}{\multiNUrot{m}{\theta}\adj\multiNUrot{m}{\theta}}{x} \matrixel{\overline{x}}{\multiNUrot{m}{\theta}\adj\multiNUrot{m}{\theta}}{\overline{x}} \right) \\
    &= C^{-1} \left(\matrixel{\overline{x}}{\multiNUrot{m}{\theta}\adj\multiNUrot{m}{\theta}}{x} - \matrixel{\overline{x}}{\multiNUrot{m}{\theta}\adj\multiNUrot{m}{\theta}}{x} \right) = 0
\end{align}
as desired. We conclude $\Urot{m}{\theta}$ is unitary. 

Finally, to show $\Urot{m}{\theta}$ is close to $\multiNUrot{m}{\theta}$ we compute 
\begin{align}
    \norm{\multiNUrot{m}{\theta} - \Urot{m}{\theta}}_F^2 
    &= \sum_{x \in \{0,1\}^m} \abs{\left(\multiNUrot{m}{\theta} - \Urot{m}{\theta}\right) \ket{x} }^2 \\
    &= \sum_{x \in \halfbitstring^m} \abs{\left(1 - C^{-1}\right)\multiNUrot{m}{\theta}\ket{x} - i^{m + 2 \abs{x}}C^{-1}\sin^{m}(\theta)\multiNUrot{m}{\theta} \ket{\overline{x}} }^2 \\
    &\leq \sum_{x \in \halfbitstring^m} \left( 1 - C^{-1} \right)^2 + C^{-2} \sin^{2m}(\theta) \\
    &\leq 2^{m/2} \left(\frac{\sin^{4m}(\theta)}{2}  + \frac{\sin^{2m}(\theta)}{1 - \sin^{2m}(\theta)}  \right) \in O\left(\theta^{2m}\right)
\end{align}
where the final big $O$ approximation holds for any fixed $m$ as $\theta \rightarrow 0$. Taking a square root then completes the proof. 
\end{proof}

Finally, we are in a position to describe the fully unitary $(X, \majmod{p}(X) \oplus \parity(X))$ sampling circuit.

\begin{thm}
\label{thm:Urot_majmod_sampling}
For $n$ sufficiently large and $p = n^c$ for any constant $c \in (0,1/2)$ there is a constant-depth circuit consisting of one and two-qubit unitary gates and $\Urot{m'}{\theta'}$ gates with $m' = \lceil c^{-1} + 1 \rceil$ and $\theta' = \pi/p$ which takes an $n$ qubit GHZ state as input and produces an output which, when measured in the computational basis, produces an output distribution $(X', Y)$ with 
\begin{align}
    \TVD((X', Y), (X, \majmod{p}(X) \oplus \parity(X))) \leq  \frac{1}{2} - \frac{1}{\pi} +  O(1/p).
\end{align}
\end{thm}

\begin{proof}
For convenience, we assume $n = Dm' + 1$ for some constant $D$. This circuit consists of a Hadamard gate applied to each qubit of the $\GHZ$ state, followed by $\Urot{m'}{\theta'}\adj$ gates applied to all qubits except the final qubit and an $\exp(-i \pi X/4)$ rotation applied to the final qubit. \Cref{fig:majmod_p_Urot_circuit} illustrates this circuit. Note the quantum state produced by this circuit pre-measurement is 
\begin{align}
    \left(\left(\Urot{\theta'}{m'}\adj \right)^{\otimes D} \otimes \exp(-i \pi X/4)\right) H^{\otimes n} \ket{\psi}.
\end{align}

To prove this circuit samples from the correct distribution first note that \Cref{lem:Frobenius_Urot_NUrot_bound} and \Cref{fact:Frobenius_Operator_bound} give that \begin{align}
    \norm{\Urot{\pi/p}{m} -  \multiNUrot{\pi/p}{m}}_\infty \in O(\theta'^m) = O(n^{-mc}) \leq O(n^{-(1 + c)})
\end{align}
Them, \Cref{fact:Operator_Tensor_Product_Bound} gives that 
\begin{align}
    \norm{\left(\left(\Urot{\theta'}{m'}\adj \right)^{\otimes D} \otimes \exp(-i \pi X/4)\right) H^{\otimes n} - \left(\left(\multiNUrot{\pi/p}{m}\adj \right)^{\otimes D} \otimes \exp(-i \pi X/4)\right) H^{\otimes n}}_\infty 
    &\in O(Dn^{-(1+c)}) \\
    &\leq O(n^{-c}).
\end{align}
Combining this observation with \Cref{fact:l2_distance_trace_distance_bound} and the definition of the operator norm $\norm{}_\infty$ gives that the classical distributions resulting from computation basis measurements of the states
\begin{align}
    \left(\left(\Urot{\theta'}{m'}\adj \right)^{\otimes D} \otimes \exp(-i \pi X/4)\right) H^{\otimes n} \ket{\psi}.
\end{align}
and 
\begin{align}
    \left(\left(\multiNUrot{\pi/p}{m}\adj \right)^{\otimes D} \otimes \exp(-i \pi X/4)\right) H^{\otimes n} \ket{\psi}
\end{align}
are $O(n^{-c})$ in total variation distance away from each other. Then \Cref{cor:multiNUrot_majmod_sampling}, together with the fact that $O(p^{3/2} e^{-n/p^2}) \leq O(1/p)$ since $p = n^{-c}$ for $c< 1/2$ completes the proof. 
\end{proof}

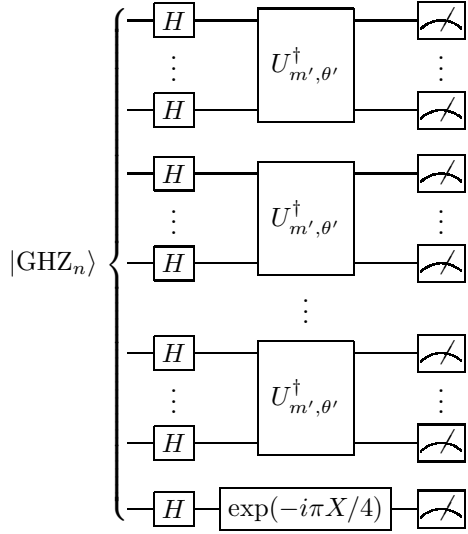
\begin{figure}

\begin{align*}
\Qcircuit @C=1em @R=1em {
\lstick{}   & \gate{H}  & \multigate{2}{\Urot{m'}{\theta'}\adj} & \meter \\
& {\makecell{\vdots\\}} & & {\makecell{\vdots\\}} \\ 
\lstick{}   & \gate{H}  & \ghost{\Urot{m'}{\theta'}\adj}        & \meter \\
\lstick{}   & \gate{H}  & \multigate{2}{\Urot{m'}{\theta'}\adj} & \meter \\
& {\makecell{\vdots\\}} & & {\makecell{\vdots\\}} \\ 
\lstick{}   & \gate{H}  & \ghost{\Urot{m'}{\theta'}\adj}        & \meter \\
& & {\makecell{\mbox{}\\\vdots\\ \mbox{}\\}} \\ 
\lstick{}   & \gate{H}  & \multigate{2}{\Urot{m'}{\theta'}\adj} & \meter \\
& {\makecell{\vdots\\}} & & {\makecell{\vdots\\}} \\ 
\lstick{}   & \gate{H}  & \ghost{\Urot{m'}{\theta'}\adj}        & \meter \\
\lstick{}   & \gate{H}  & \gate{\exp(-i \pi X / 4)}  & \meter 
\inputgroupv{1}{11}{0.8em}{9.2em}{\ket{\GHZ_n} \;\;\;\;\;\;\;\;} 
} 
\end{align*}
\caption{Constant-depth fully unitary circuit producing approximate samples from the distribution $(\majmod{p}(X) \oplus \parity(X), X)$. Here $p = n^c$ for some constant $c \in (0,1]$, $\theta' = \pi/p$, $m' = \left\lceil c^{-1} + 1 \right \rceil$ and $n = Dm' + 1$ for some large integer $D$.  }
\label{fig:majmod_p_Urot_circuit}
\end{figure}

\section{Sampling From $(Z, \pmmajmod{p}(Z))$ Without a GHZ State} \label{sec:quantum_circuit_no_GHZ}

\newcommand{\PM}[1]{\operatorname{PM}(#1)}
\newcommand{\BPM}[1]{\operatorname{PM}_{#1}} 
\newcommand{\pathsum}[1]{\operatorname{h}(#1)} 
\newcommand{\bpathsum}{\operatorname{h}} 

In this section we define sampling tasks related to the $(X, \majmod{p}(X) \oplus \parity(X))$ sampling task considered in \Cref{sec:GHZ_state_sampling}, but which can be performed (approximately) by a constant-depth quantum circuit without access to a GHZ input state. At a high level, the approach we use to construct these tasks mirrors the approach used in~\cite{watts2019exponential} to find a relational problem that can be solved by a $\QNC^0$ circuit without access to a GHZ state. First, we review ``Poor Man's GHZ States'': GHZ-like states which (unlike the GHZ state) can be constructed by $\QNC^0$ circuits. Then we modify the circuit constructed in \Cref{ssec: unitary majority GHZ sampling} by replacing the GHZ input state with a circuit constructing a poor man's GHZ state. Finally, we define a new sampling task based on the output of these modified circuits. 

\subsection{Review of Poor Man's GHZ States}

\begin{defn}\label{defn:binary_tree}
For any integer $n$ let $\bintree[n]$ be the balanced binary tree on $n$ vertices. Label its edges $e_1, ..., e_{n-1}$ and vertices $v_0, ..., v_{n-1}$ (note the vertex labels start at $0$), with $v_0$ the root of $T$. For every non-root vertex $v_i \in \{v_1, ..., v_{n-1}\}$ define $P(v_i)$ to be the set of edges contained in the (unique) path going from $v_0$ to $v_i$. Finally, define the function $\pathsum{d} : \{0,1\}^{n-1} \rightarrow \{0,1\}^{n-1}$ by
\begin{align}
    \pathsum{d}_i = \bigoplus_{j :\ e_j \in P(v_i)} d_j && i \in \{1, 2, \dots, n-1\}.
\end{align}
That is, thinking of the bitstring $d$ as assigning values to the edges of $\bintree[n]$, $\pathsum{d}$ assigns a value to every non-root vertex $v_i$ of $\bintree[n]$ equal to the parity of the edge values going from $v_0$ to $v_i$.
\end{defn}

\begin{defn} 
\label{defn:btpmGHZ}
Define the (binary tree) Poor Man's GHZ state: 
\begin{align}
    \ket{\BPM{n}} = \sum_{d \in \{0,1\}^{n-1}} \frac{1}{2^{(n-1)/2}} \ket{d} \otimes \frac{1}{\sqrt{2}} \left(\ket{\bpathsum(d) 0  \vphantom{\overline{\bpathsum(d)}}} + \ket{\overline{\bpathsum(d)} 1} \right) 
\end{align}
We call the first $n-1$ qubits of $\ket{\BPM{n}}$ ``edge'' qubits, and the last $n$ qubits ``vertex'' qubits. Note that the $n$ in $\ket{\BPM{n}}$ gives the number of vertex qubits in the state, not the total number of qubits. 
\end{defn}

Intuitively, it is occasionally helpful to think of the $n$ vertex qubits of the state $\ket{\BPM{n}}$ as being in an ``almost-\GHZ state'', or a \GHZ state with additional Pauli $X$ type ``error'' terms specified by the edge qubits. 
To explain this intuition, note that we can also write the state $\ket{\BPM{n}}$ as 
\begin{align}
    \ket{\BPM{n}} = \frac{1}{2^{(n-1)/2}} \sum_{d \in \{0,1\}^{n-1}}  \left( \ket{d} \otimes \left(\left(\bigotimes_{i=1}^{n-1} X^{\bpathsum(d)_i} \right) \otimes I_2 \right) \ket{\GHZ_n} \right)  \label{eq:PMGHZ_X_error}
\end{align}
We will make use of \Cref{eq:PMGHZ_X_error} when working with the state $\ket{\BPM{n}}$ later in this section. 

\begin{thm} \label{thm:Binary_Tree_Poor_Man_Construction}
For any $n$, the state $\ket{\BPM{n}}$ can be constructed by a depth-3 circuit consisting of $1$ and $2$ qubit gates acting on $2n-1$ qubits. 
\end{thm}

\begin{proof}
    This state can be constructed by following the procedure outlined in Theorem 17 of~\cite{watts2019exponential}, but omitting the measurement of the edge qubits. We recap this procedure here. 
    
    Begin with $2n-1$ qubits, $n$ of which we identify with the vertices $v_0, ..., v_{n-1}$ of the tree $B_n$ and $n-1$ of which we identify with edges $e_1, ... e_{n-1}$ of the same tree. Apply a Hadamard gate to each vertex qubit. Then, for every pair of vertices $v_i$ and $v_j$ connected by an edge $e_{k}$, apply CNOT gates with controls on vertex qubits $v_i$ and $v_j$ and target on the edge qubit $e_{k}$. Order the edge qubits as in the tree $B_n$; these form the first $n-1$ qubits of $\ket{\BPM{n}}$. Order the vertex qubits $v_1...v_{n-1}v_0$ (note the qubit identified with the root vertex comes last in this ordering); these form remaining $n$ qubits of the state $\ket{\PM{n}}$.
    
    To see that this circuit produces the correct state first observe that after the Hadamard gates are applied and before the CNOT gates are applied, the vertex qubits are in a uniform superposition over all computational basis states. 
    We order the vertex qubits as in the state $\ket{\BPM{n}}$, so the final vertex qubit is associated with the root vertex of the graph $B_n$. 
    It is then straightforward to check that, for any $n-1$ bit string $x = x_1...x_{n-1}$, if the vertex qubits are in state $\ket{x 0}$  then applying the $\CNOT$ gates puts the edge qubits in the state $\bpathsum^{-1}(x)$. Similarly, if vertex qubits are in the state $\ket{x 1}$, applying the $\CNOT$ gates puts the edge qubits in the state $\bpathsum^{-1}(\overline{x})$. Then we can write the state produced by our circuit as
    \begin{align}
        &\frac{1}{2^{n/2}} \left( \sum_{x \in \{0,1\}^{n-1}} \ket{\bpathsum^{-1}(x)} \otimes \ket{x 0} + \sum_{x \in \{0,1\}^{n-1}} \ket{\bpathsum^{-1}(\overline{x})} \otimes \ket{x 1} \right) \\
        &\hspace{30pt}= \frac{1}{2^{n/2}} \left( \sum_{d \in \{0,1\}^{n-1}} \ket{d} \otimes \ket{\bpathsum(d) 0} + \sum_{d \in \{0,1\}^{n-1}} \ket{d} \otimes \ket{\overline{\bpathsum(d)} 1} \right) \\
        &\hspace{30pt}=\frac{1}{2^{(n-1)/2}} \left( \sum_{d \in \{0,1\}^{n-1}} \ket{d} \otimes \left( \frac{1}{\sqrt{2}} \ket{\bpathsum(d) 0} + \ket{\overline{\bpathsum(d)} 1}\right) \right) = \ket{\BPM{n}}
    \end{align}
    where we used on the second line that the function $\bpathsum$ was one-to-one. 
    
    Finally, we show this circuit can be implemented in depth 3. Consider the $2n-1$ vertex graph obtained from $B_n$ by bifurcating each edge of $B_n$ -- that is, replacing each edge of $B_n$ connecting vertices $v_i$ and $v_j$ with a new vertex connected to both $v_i$ and $v_j$. This graph is still a tree, hence 2-colorable, and the edges of this graph are in one-to-one correspondence with $\CNOT$ gates which need to be implemented in the circuit described above. All $\CNOT$ gates in the same color class touch disjoint qubits and be applied simultaneously, so we see all $\CNOT$ gates can be applied in depth 2. Adding the layer of Hadamard gates required as the first step shows this whole circuit can be implemented in depth 3. 
    
\end{proof}
\subsection{Sampling with \texorpdfstring{$\QNC^0$}{QNC0} Circuits}

We begin with a description of the distribution which we will show can be sampled from (approximately) by a $\QNC^0$ circuit. Like the distributions considered in \Cref{sec:GHZ_state_sampling}, it will be a distribution of the form $(Z, f(Z))$ where $Z$ is a uniformly random bitstring and $f(Z) : \{0,1\}^n \rightarrow \{0,1\}$ is some function. However, the function $f$ considered here is substantially more complicated than the functions considered in \Cref{sec:GHZ_state_sampling}. We define this function next.

\begin{defn}\label{defn:pmmajmod}
For any prime $p$ define the function $\pmmajmod{p} : \{0,1\}^{2n-2} \rightarrow \{0,1\}$ to act on a $2n-2$ bit string $z$ via the following procedure:
\begin{enumerate}
    \item Associate the first $n-1$ bits of $z$ with edges of the complete binary tree $\bintree[n]$ and the next $n-1$ bits with the non-root vertices $v_1...v_{n-1}$, following the same ordering as in \Cref{defn:binary_tree}. Label bits associated with edges $d$ and the bits associated with vertices $x$. 
    \item For any integer $a$ define
    \begin{align}
    \MM{p}(a) :=
    \begin{cases}
    0 \text{ if } a < p/2 \\
    1 \text{ otherwise.}
    \end{cases}
    \end{align}
    \item Set 
    \begin{align}
        \pmmajmod{p}(z) = \MM{p}\left(\sum_{i=1}^{n-1} x_i (-1)^{\bpathsum(d)_i} \right) \bigoplus \parity(x)
    \end{align}
\end{enumerate}
\end{defn}

Now we construct a quantum circuit that samples approximately from $(Z, \pmmajmod{p}(Z))$ without requiring a $\GHZ$ state input. As in \Cref{sec:GHZ_state_sampling}, we begin by describing a circuit that performs the sampling task and involves single qubit non-unitary rotations $\NUrot{\theta}$. 

\begin{thm} \label{thm:pmmajmod_ideal_sampling}
For any $p \in \mathbb{Z}^+$ there is a constant-depth circuit consisting of one and two-qubit unitary gates and $\NUrot{\theta}$ operations which takes the $(2n-1)$-qubit all zeros state as input and produces an output which, when measured in the computational basis, produces an output distribution $(Z', Y)$ with 
\begin{align}
    \TVD((Z', Y), (Z, \pmmajmod{p}(Z))) \leq  \frac{1}{2} - \frac{1}{\pi} + \frac{1}{2p} + O(p^{3/2}e^{-n/4p^2}).
\end{align}
\end{thm}

\begin{proof}
The first step is preparing the state $\ket{\BPM{n}}$, which can be done in constant-depth by \Cref{thm:Binary_Tree_Poor_Man_Construction}. After that, the same non-unitary circuit as described in the proof of \Cref{thm:NUrot_majmod_sampling} is applied to the vertex qubits of the poor man's GHZ state. This is illustrated in \Cref{fig:pmmajmod_NUrot_circuit}.

\begin{figure}[ht]

\begin{align*}
\Qcircuit @C=1em @R=1em {
\lstick{}   & \qw  & \qw & \meter \\
& {\makecell{\vdots\\}} \\ 
\lstick{}   & \qw  & \qw & \meter \\
&&&\\
\lstick{}   & \gate{H}  & \gate{\NUrot{\pi/p}\adj} & \meter \\
& {\makecell{\vdots\\}} \\ 
\lstick{}   & \gate{H}  & \gate{\NUrot{\pi/p}\adj} & \meter \\  
\lstick{}   & \gate{H}  & \gate{\exp(-i \pi X / 4)}  & \meter 
\gategroup{1}{1}{3}{4}{1.5em}{--}
\gategroup{5}{1}{8}{4}{1.5em}{--}
\inputgroupv{1}{8}{3em}{6.5em}{\ket{\BPM{n}} \hspace{35pt}} 
} 
\end{align*}
\caption{Constant-depth non-unitary circuit producing approximate samples from the distribution $(Y, \pmmajmod{p}(Y))$. The upper box indicates the $n-1$ ``edge'' qubits of the state $\ket{\BPM{n}}$. The lower box indicates the $n$ ``vertex'' qubits of the same state.}
\label{fig:pmmajmod_NUrot_circuit}
\end{figure}
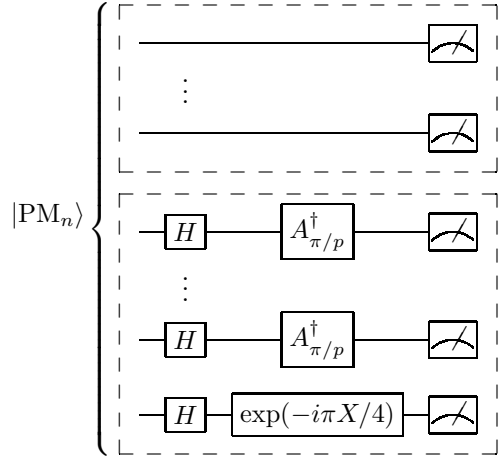

\newcommand{\altNUrot}[1]{A\left(#1\right)}

To see that this circuit approximately samples from the correct distribution we write the state $\ket{\BPM{n}}$ as a GHZ state with additional controlled $X$ ``error'' terms, then commute those through the rest of the circuit. In the following argument we will need to pay close attention to the rotation angle $\theta$ in the non-unitary operator $\NUrot{\theta}$. For this reason, \textit{for the remainder of this section only}, we change notation and write $\NUrot{\theta}$ as $\altNUrot{\theta}$. 

The key observation is the operator identity 
\begin{align}
     \altNUrot{\theta}\adj = \altNUrot{-\theta}\adj Z \label{eq:NUrot_X_commute}
\end{align}
which holds for any $\theta$ and can quickly be verified by checking the action of $Z \altNUrot{\theta}$ and $\altNUrot{-\theta} Z$ on $\ket{0}$ and $\ket{1}$ basis states. Then (using \Cref{eq:PMGHZ_X_error} as a starting point) we can write the pre-measurement state produced by the circuit above as:
\begin{align}
    &\frac{1}{2^{(n-1)/2}} \sum_{d \in \{0,1\}^{n-1}} \left(I_{2^{n-1}} \otimes \left( \bigotimes_{j=1}^{n-1} \altNUrot{\frac{\pi}{p}} \adj H \right) \otimes \exp(\frac{-i \pi X}{4}) H   \right) \left( \ket{d} \otimes \left(\left(\bigotimes_{j=1}^{n-1} X^{\bpathsum(d)_j} \right) \otimes  I_2 \right) \ket{\GHZ_n} \right) \nonumber \\
    &\hspace{10pt}=\frac{1}{2^{(n-1)/2}} \sum_{d \in \{0,1\}^{n-1}} \left(I_{2^{n-1}}  \otimes \left(\bigotimes_{j=1}^{n-1} \altNUrot{\frac{\pi}{p}} \adj H X^{\bpathsum(d)_j}\right) \otimes \exp(\frac{-i \pi X}{4}) H \right) \left(\ket{d} \otimes \ket{\GHZ_n} \right) \\
    &\hspace{10pt}=\frac{1}{2^{(n-1)/2}} \sum_{d \in \{0,1\}^{n-1}} \left( \ket{d} \otimes \left( \left( \bigotimes_{j=1}^{n-1} Z^{\bpathsum(d)_j}  \altNUrot{(-1)^{\bpathsum(d)_j} {\frac{\pi}{p}}} \adj \right) \otimes   \exp(\frac{-i \pi X}{4})  \right) H ^{\otimes n}\ket{\GHZ_n}  \right).
\end{align}
Where the rearrangement on the third line used the operator identity discussed above (\Cref{eq:NUrot_X_commute}).

From this it is clear that the measurement of the first $n-1$ edge qubits produces a uniformly random bitstring. We assume that such a measurement has been carried out, producing some bitstring $d$. Then, following the same analysis as used in the proof of \Cref{thm:NUrot_majmod_sampling}, we consider the (unnormalized) state of the first vertex qubit when the first $n-1$ vertex qubits have been measured and bitstring $x = x_1 x_2 ... x_{n-1}$ is observed:
\begin{align}
        &\bra{x}_{1...n-1}  \left(\bigotimes_{j=1}^{n-1} Z^{\bpathsum(d)_j}  \altNUrot{(-1)^{\bpathsum(d)_j} {\frac{\pi}{p}}} \adj \right)  \otimes \exp(\frac{-i \pi X}{4})  \left( H ^{\otimes n}\ket{\GHZ_n}\right) \nonumber \\
        &\hspace{30pt}= (-1)^{\abs{x}} \bra{x}_{1...n-1} \left(\bigotimes_{j=1}^{n-1} \altNUrot{(-1)^{\bpathsum(d)_j} {\frac{\pi}{p}}} \adj   \right) \otimes \exp(\frac{-i \pi X}{4})  \left( H ^{\otimes n}\ket{\GHZ_n}\right)  \\
        &\hspace{30pt}= (-1)^{\abs{x}} 2^{-(n-1)} \exp(iX \left(- \frac{\pi}{4} +   \frac{\pi}{p}  \sum_{j=1}^{n-1} \left( x_j(-1)^{\bpathsum(d)_j} \right)  \right))\ket{\parity(x)} \label{eq:pmmajmod_last_qubit_state}, 
\end{align}
where the final line followed from the same series of identities as used in \Cref{eq:NU_analysis1,eq:NU_analysis2,eq:NU_analysis3,eq:NU_analysis4,eq:NU_analysis5,eq:NU_analysis6}. The key features of this argument are illustrated in \Cref{fig:pmmajmod_NUrot_circuit_analysis}, where we focus just on the analysis of the vertex qubits when the edge qubits are measured and classical bitstring $d$ is observed. 

\begin{figure}[htp]

\begin{align*}
&\Qcircuit @C=1em @R=1em {
\lstick{}   & \gate{X^{\bpathsum(d)_1}} & \gate{H}  & \gate{\altNUrot{\pi/p}\adj} & \meter & \cw & \rstick{x_1} \\
\lstick{}& {\makecell{\vdots\\}} &&&\\
\lstick{}   & \gate{X^{\bpathsum(d)_{n-1}}} & \gate{H}  & \gate{\altNUrot{\pi/p}\adj} &\meter & \cw & \rstick{x_{n-1}} \\  
\lstick{}   & \qw & \gate{H}  & \gate{\exp(-i \pi X / 4)} & \qw
\inputgroupv{1}{4}{1em}{3.25em}{\ket{\GHZ_n} \;\;\;\;\;\;\;}
}
\\
\\
\cdashline{1-2}
\\
&\hspace{45pt} = \hspace{45pt}
\begin{aligned}
\Qcircuit @C=1em @R=1em {
\lstick{}   & \gate{H}  & \gate{Z^{\bpathsum(d)_1}}     &  \gate{\altNUrot{\pi/p}\adj} & \meter & \cw & \rstick{x_1} \\
\lstick{}   &  {\makecell{\vdots\\}}     &       && \\
\lstick{}   & \gate{H}  & \gate{Z^{\bpathsum(d)_{n-1}}} &  \gate{\altNUrot{\pi/p}\adj} & \meter & \cw & \rstick{x_{n-1}} \\  
\lstick{}   & \gate{H}  & \qw                           &  \gate{\exp(-i \pi X / 4)}  & \qw
\inputgroupv{1}{4}{1em}{3.25em}{\ket{\GHZ_n} \;\;\;\;\;\;\;}
}
\end{aligned}
\\
\\
\cdashline{1-2}
\\
&\hspace{45pt} = \hspace{45pt}
\begin{aligned}
\Qcircuit @C=1em @R=1em {
\lstick{}   & \gate{H}  &  \gate{\altNUrot{(-1)^{\bpathsum(d)_1} \pi/p}\adj} & \gate{Z^{\bpathsum(d)_1}} & \meter & \cw & \rstick{x_1} \\
\lstick{}   & {\makecell{\vdots\\}}  &    && \\
\lstick{}  & \gate{H}  &  \gate{\altNUrot{(-1)^{\bpathsum(d)_{n-1}} \pi/p}\adj} & \gate{Z^{\bpathsum(d)_{n-1}}} & \meter & \cw & \rstick{x_{n-1}} \\  
\lstick{}   & \gate{H}  &  \gate{\exp(-i \pi X / 4)} & \qw & \qw
\inputgroupv{1}{4}{1em}{3.25em}{\ket{\GHZ_n} \;\;\;\;\;\;\;}
}
\end{aligned}
\\
\\
\cdashline{1-2}
\\
&\hspace{45pt} = \hspace{45pt}
\begin{aligned}
\Qcircuit @C=1em @R=1em {
\lstick{}   & \gate{H}              &  \qw                              & \meter & \cw & \rstick{x_1} \\
\lstick{}   & {\makecell{\vdots\\}} &                                   & \\
\lstick{}   & \gate{H}              &  \qw                              & \meter & \cw & \rstick{x_{n-1}} \\  
\lstick{}   & \gate{H}              &  \gate{\exp(-i X \left(\pi /4 + \pi/p \sum_j x_j (-1)^{\bpathsum(d)_j} \right) )}        & \qw
\inputgroupv{1}{4}{1em}{3.25em}{\ket{\GHZ_n} \;\;\;\;\;\;\;}
}
\end{aligned}
\end{align*}
\caption{The state of the final vertex qubit of the circuit described in \Cref{fig:pmmajmod_NUrot_circuit} when all other vertex qubits (and edge qubits) are measured in the computational basis. Equivalence between lines is explained in the proof of \Cref{thm:pmmajmod_ideal_sampling}.}
\label{fig:pmmajmod_NUrot_circuit_analysis}
\end{figure}
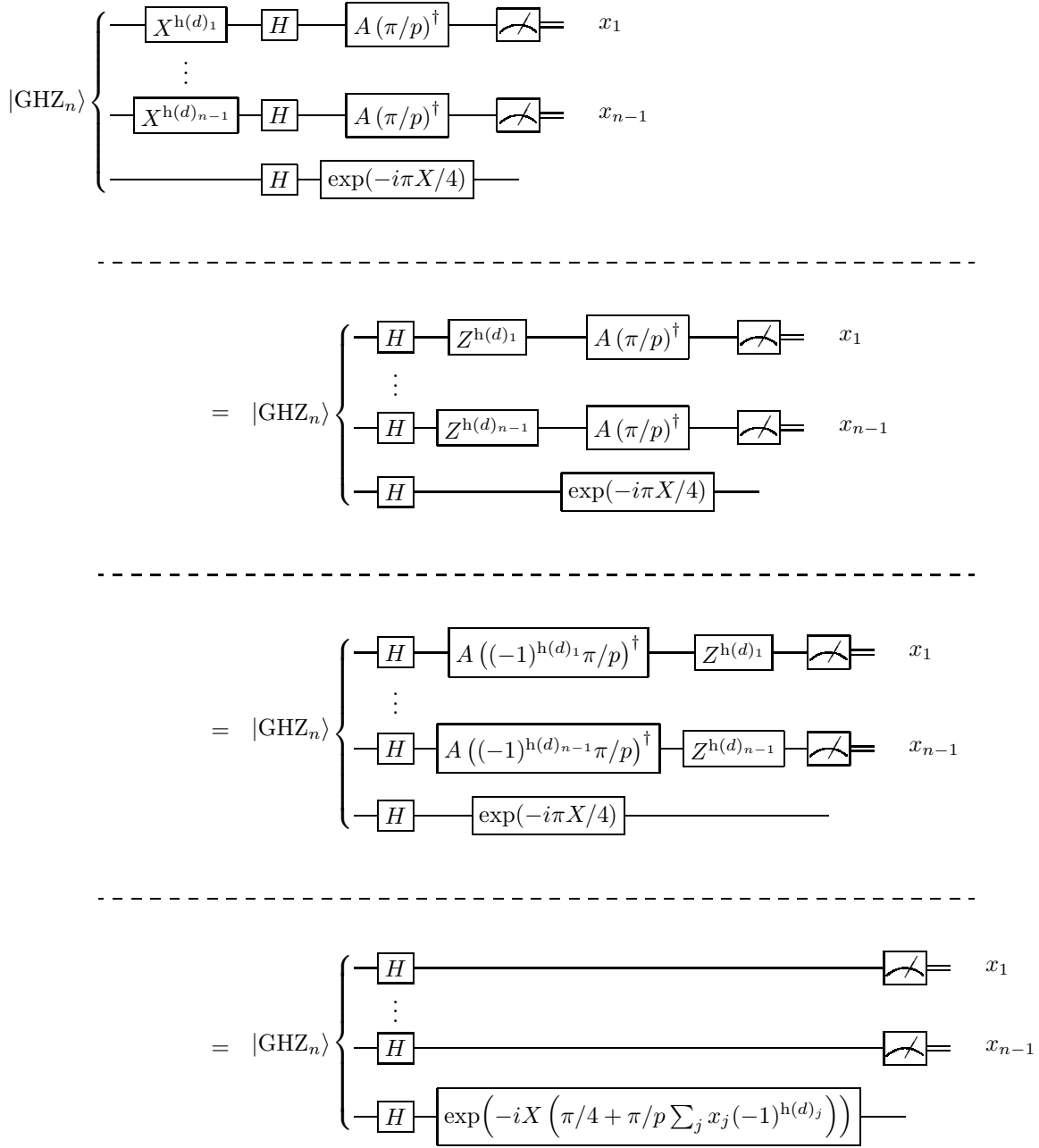

Next (still following the analysis used in \Cref{ssec: non-unitary majority GHZ sampling}) we note that the vector above has norm $2^{-(n-1)}$ for any string $x$, and hence the bitstring $x$ observed when measuring the first $n-1$ vertex qubits is uniformly random. Additionally, we let $Y_{d,x}$ be the random variable representing the outcome measurement applied to the final qubit of the circuit depicted in \Cref{fig:pmmajmod_NUrot_circuit}, conditioned on the measurement of the previous $2n-2$ qubits giving the bitstring $(d,x)$. Straightforward calculation applied to \Cref{eq:pmmajmod_last_qubit_state} gives
\begin{align}
    \Pr[Y_{d,x} = \parity(x)] = \cos^2\left( - \frac{\pi}{4} + \frac{\pi}{p} \left(\sum_i x_i (-1)^{\bpathsum(d)_i} \right) \right)
\end{align}
Then, small extension of \Cref{claim:ideal_circuit_prob_failure} (proven next, in \Cref{lem:pmmajmod_circuit_failure}) gives 
\begin{align}
    \frac{1}{2^{2n-2}} \sum_{(d,x) \in \{0,1\}^{2n-2}}\Pr[Y_{d,x} \neq \pmmajmod{p}(d,x)] \leq \frac{1}{2} - \frac{1}{\pi} + \frac{1}{2p} +  O(p^{3/2}e^{-n/4p^2}).
\end{align}

Finally, we let $D', X'$ be random variables representing the output of measuring the edge qubits and first $n-1$ vertex qubits of the circuit depicted in \Cref{fig:pmmajmod_NUrot_circuit}, respectively. We have already shown that the marginal distributions of $D'$ and $X'$ are uniformly random and so we find 
\begin{align}
    \TVD((D',X',Y_{D',X'}), (Z, \pmmajmod{p}(Z))) \leq \frac{1}{2} - \frac{1}{\pi} + \frac{1}{2p} +  O(p^{3/2}e^{-n/4p^2})
\end{align}
by exactly the same argument as used to finish the proof of \Cref{thm:NUrot_majmod_sampling}.

\end{proof}

\begin{lem}
    \label{lem:pmmajmod_circuit_failure} Define the random variable $Y_{d,x}$ as in the proof of \Cref{thm:pmmajmod_ideal_sampling}, so 
    \begin{align}
        \Pr[Y_{d,x} = \parity(x)] = \cos^2\left( - \frac{\pi}{4} + \frac{\pi}{p} \left(\sum_i x_i (-1)^{\bpathsum(d)_i} \right) \right)
    \end{align}
    Then 
    \begin{align}
        \frac{1}{2^{2n-2}} \sum_{(d,x) \in \{0,1\}^{2n-2}}\Pr[Y_{d,x} \neq \pmmajmod{p}(d,x)] \leq \frac{1}{2} - \frac{1}{\pi} + \frac{1}{2p} +  O(p^{3/2}e^{-n/4p^2}).
    \end{align}
\end{lem}

\begin{proof}
    Let $D,X$ be random variables each taking value uniformly at random from $\{0,1\}^{n-1}$. Then we can write 
    \begin{align}
            &\frac{1}{2^{2n-2}} \sum_{(d,x) \in \{0,1\}^{2n-2}}\Pr[Y_{d,x} \neq \pmmajmod{p}(d,x)] \nonumber \\
            &\hspace{20pt}= \Pr[Y_{D,X} \neq \parity(x) \oplus \MM{p}\left(\sum_i x_i (-1)^d_i \right)] \\
            &\hspace{20pt}= \sum_k \Pr[Y_{D,X} \neq \parity(x) \oplus \MM{p}\left( k \right) \Big\vert \sum_i X_i (-1)^D_i = k] \Pr[\sum_i X_i (-1)^D_i = k] \label{eq:ideal_majmod_circuit_failure_prob_initial}
    \end{align}
    We compare this equation to \Cref{eq:prob_Y_incorrect_initial_eq}, and note that (after rewriting $\majmod{p}(X) = \MM{p}(\abs{X})$) the two probabilities are identical except that the random variable $\abs{X}$ has been replaced by $\sum X_i (-1)^D_i$. Then the proof of the bound proceeds identically to the proof of bound in \Cref{claim:ideal_circuit_prob_failure}, except that we need a bound on the total variation distance between the distribution of the random variable $\sum_i X_i (-1)^{D_i} \pmod{p}$ and the uniform distribution over $\{0,1,...,p-1\}$.
    
    To do this, we write 
    \begin{align}
        \sum_i X_i (-1)^{D_i} = \sum_i X_i - 2 \sum_{i : X_i = 1} D_i
    \end{align}
    and note that both terms in the right-hand side equation give uniform distributions mod $p$ by \Cref{fact:puniform} (provided that close to half the bits of $X_i$ are ones, which happens with high probability). 
    
    Formally, let $\tilde{X}$ be the random variable taking value uniformly at random from the set of $n$-bit strings with Hamming weight at least $n/4$. Then we have 
    \begin{align}
        \TVD\left(\sum_i X_i - 2 \sum_{i : X_i = 1} D_i, \sum_i \tilde{X}_i - 2 \sum_{i : \tilde{X}_i = 1} D_i
        \right) \leq \TVD(X, \tilde{X}) \leq \exp(-n/8), \label{eq:bpath_sum_uniform_dist_pt1}
    \end{align}
    where the first inequality follows because for any distributions $X$ and $\tilde{X}$ and (possibly random) function $f$ we have $\TVD(X, X') \geq \TVD(f(X), f(X'))$, and the second inequality follows from Hoeffding's. Then, letting $U_p$ denote the uniform distribution mod $p$, for any $\tilde{x}$ in the support of $\tilde{X}$ we have, by \Cref{fact:puniform}, that
    \begin{align}
        \TVD \left( 2 \sum_{i : \tilde{x}_i = 1} D_i \pmod{p}, U_p \right) \leq \sqrt{p} \exp(-n/4p^2)
    \end{align}
    and hence 
    \begin{align}
        \TVD\left(\abs{\tilde{x}} - 2 \sum_{i : \tilde{x}_i = 1} D_i \pmod{p}, U_p 
        \right) \leq \sqrt{p} \exp(-n/4p^2)
    \end{align}
    since shifting a distribution doesn't change its distance from the uniform distribution. Then summing over all possible $\tilde{x}$ we see
        \begin{align}
        \TVD\left(\abs{\tilde{X}} - 2 \sum_{i : \tilde{X}_i = 1} D_i \pmod{p}, U_p 
        \right) \leq \sqrt{p} \exp(-n/4p^2). \label{eq:bpath_sum_uniform_dist_pt2}
    \end{align}
    Combining \Cref{eq:bpath_sum_uniform_dist_pt1,eq:bpath_sum_uniform_dist_pt2} gives 
    \begin{align}
        \TVD\left(\sum_i X_i - 2 \sum_{i : X_i = 1} D_i \pmod{p}, U_p
        \right) \leq \exp(-n/8) + \sqrt{p} \exp(-n/4p^2) = O(\sqrt{p} \exp(-n/4p^2)).
    \end{align}
    Then, following the same proof as in \Cref{claim:ideal_circuit_prob_failure} and plugging the above inequality in place of \Cref{fact:puniform} gives the desired bound. 
\end{proof}

Then, following the same arguments as used in \Cref{ssec: unitary majority GHZ sampling}, we show that we can replace the non-unitary rotation gates used in the circuit described above with actual unitary gates, while causing small disturbance to the output distribution. The result of this procedure is a $\QNC^0$ circuit that takes the all zeros state as input and whose output samples approximately from the distribution $(Z, \pmmajmod{p}(Z))$.

\begin{thm}
\label{thm:Urot_pmmajmod_sampling}
For $n$ sufficiently large and $p = n^c$ for some constant $c \in (0,1/2)$ there is a constant-depth circuit consisting of one and two qubit unitary gates and $\Urot{m'}{\theta'}$ gates with $m' = \lceil c^{-1} + 1 \rceil$ and $\theta' = \pi/p$ which takes the $(2n-1)$-qubit all zeros state as input and produces an output which, when measured in the computational basis, produces a distribution $(Z', Y)$ with an $n$-bit output which correlates approximately with the distribution $(Z, \pmmajmod{p}(Z))$. 
\end{thm}

\begin{proof}

The desired circuit can be constructed from the circuit presented in \Cref{fig:pmmajmod_NUrot_circuit} following the same procedure as used in \Cref{ssec: unitary majority GHZ sampling}. Specifically, we first replace blocks of $m$ parallel $\NUrot{\theta}$ gates with $\multiNUrot{\theta}{m}$ gates, then replace those with $\Urot{\theta}{m}$ gates. The only additional complication we encounter is that we must apply a final permutation to our output bits to accommodate a ``shuffling effect" caused by replacing blocks of the $\NUrot{\theta}$ gates by $\multiNUrot{\theta}{m}$. The final circuit is presented in \Cref{fig:pmmajmod_unitary_circuit}, where the $C_m$ gate denotes a permutation whose action on the $m$ qubit computational basis state $\ket{x_1 x_2 ... x_m}$ is given by
\begin{align}
    C_m \ket{x_1 x_2 ... x_m} = \ket{x_2 x_3 ... x_m x_1}.
\end{align}

\begin{figure}[htb]

\begin{align*}
\Qcircuit @C=1em @R=1em {
\lstick{}   & \qw  & \qw & \qw & \meter \\
& & {\makecell{\vdots\\}} \\ 
\lstick{}   & \qw  & \qw & \qw & \meter \\
&&&\\
\lstick{}   & \gate{H}  & \multigate{2}{\Urot{m'}{\theta'}\adj} & \multigate{2}{C_m} & \meter \\
& {\makecell{\vdots\\}} & &  & &  \\ 
\lstick{}   & \gate{H}  & \ghost{\Urot{m'}{\theta'}\adj}        & \ghost{C_m}       & \meter \\
& & {\makecell{\mbox{}\\\vdots\\ \mbox{}\\}} & {\makecell{\mbox{}\\\vdots\\ \mbox{}\\}} \\ 
\lstick{}   & \gate{H}  & \multigate{2}{\Urot{m'}{\theta'}\adj} & \multigate{2}{C_m} & \meter \\
& {\makecell{\vdots\\}} & &  & & \\ 
\lstick{}   & \gate{H}  & \ghost{\Urot{m'}{\theta'}\adj}         & \ghost{C_m}      & \meter\\
\lstick{}   & \gate{H}  & \gate{\exp(-i \pi X / 4)}  & \qw& \meter  
\gategroup{1}{1}{3}{5}{1.5em}{--}
\gategroup{5}{1}{12}{5}{1.5em}{--}
\inputgroupv{1}{12}{3em}{9.7em}{\ket{\BPM{n}} \hspace{35pt}} 
} 
\end{align*}
\caption{Constant-depth unitary circuit producing approximate samples from the distribution $(Y, \pmmajmod{p}(Y))$. Note that $m$ is constant, and so the unitaries acting on $m$ qubits have constant size. The upper box indicates the $n-1$ ``edge'' qubits of the state $\ket{\BPM{n}}$. The lower box indicates the $n$ ``vertex'' qubits of the same state.}
\label{fig:pmmajmod_unitary_circuit}
\end{figure}
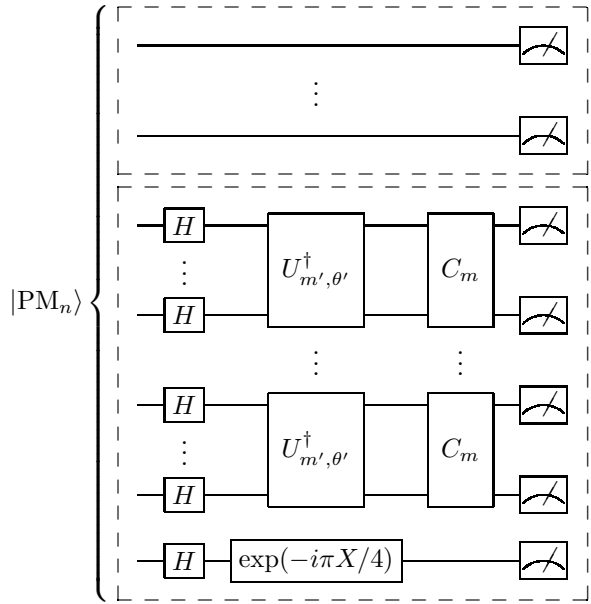

As a first step towards showing this circuit samples from the desired distribution, we show that replacing the parallel $\NUrot{\theta}$ gates in the circuit of \Cref{fig:pmmajmod_NUrot_circuit} with $\multiNUrot{\theta}{m}$ gates followed by a $C_m$ gates doesn't change the post-measurement distribution produced by the circuit. To see why, we consider the state of the final vertex qubit in both circuits after a measurement is performed on all edge qubits, producing bitstring $d$, and the first $m$ vertex qubits, producing bitstring $x_1x_2...x_m$. In the circuit described in \Cref{fig:pmmajmod_NUrot_circuit}, the state of the final qubit is given by
\begin{align}
    &\bra{x_1 x_2 ... x_m} \bigotimes_{i=1}^m \NUrot{\theta} Z^{\bpathsum(d)_i} \prod_i \CNOT_{i,n} \ket{+}^{\otimes m} \otimes \ket{0} \\
    &= \bra{x_1 x_2 ... x_m } \bigotimes_{i=1}^m \exp(i \theta X x_{i}) Z^{\bpathsum(d)_i}\prod_i \CNOT_{i,n} \ket{+}^{\otimes m} \otimes \ket{0} \\
    &= \bra{x_1 x_2 ...  x_m} \bigotimes_{i=1}^m Z^{\bpathsum(d)_i} \ket{+}^{\otimes n} \otimes \exp(i \theta X \sum_i x_i (-1)^{\bpathsum(d)_i}) \ket{\parity(x_1x_2...x_m)}
\end{align}
and, if the $\NUrot{\theta}$ gates are replaced by a $C_m$ gate and $\multiNUrot{\theta}{m}$ gate the state of the final qubit is given by 
\begin{align}
    &\bra{x_1 x_2 ... x_m} C_m \multiNUrot{\theta}{m} \bigotimes_{i=1}^m Z^{\bpathsum(d)_i} \prod_i \CNOT_{i,n} \ket{+}^{\otimes m} \otimes \ket{+}_n \\
    &=\bra{x_2 ... x_m x_1} \multiNUrot{\theta}{m} \bigotimes_{i=1}^m Z^{\bpathsum(d)_i} \prod_i \CNOT_{i,n} \ket{+}^{\otimes m} \otimes \ket{+}_n \\
    &= \bra{x_2 ... x_m x_1} \bigotimes_{i=1}^m \exp(i \theta X x_{i}) Z^{\bpathsum(d)_i}\prod_i \CNOT_{i,n} \ket{+}^{\otimes m} \otimes \ket{+}_n \\
    &= \bra{x_2 ... x_m x_1} \bigotimes_{i=1}^m Z^{\bpathsum(d)_i} \ket{+}^{\otimes n} \otimes \exp(i \theta X \sum_i x_i (-1)^{\bpathsum(d)_i}) \ket{\parity(x_2...x_mx_1)}.
\end{align}
Since these states are the same up to an overall phase we see the change does not affect the probability of observing outcomes $d$ and $x_1,...,x_m$ or the state of the unmeasured qubit. 
It is straightforward to extend this analysis to the case where the same replacement is made to all $D$ blocks of $\NUrot{\theta}$ gates in the circuit of \Cref{fig:pmmajmod_NUrot_circuit}.

It remains to show that replacing the $\multiNUrot{\theta}{m}$ gates (in the circuit produced by the replacement discussed above) with $\Urot{\theta}{m}$ gates causes a negligible change to the distribution output by the circuit after a computational basis measurement. Following the same argument as used to prove \Cref{thm:classical_LB_with_GHZ} we see
\begin{align}
    &\Biggl\lVert I_2^{\otimes (n-1)} \otimes \left(\left(C_m \Urot{\theta'}{m'}\adj \right)^{\otimes D} \otimes \exp(-i \pi X/4)\right) H^{\otimes n} \nonumber \\
    &\hspace{50pt}- I_2^{\otimes (n-1)} \otimes \left(\left(C_m \multiNUrot{\pi/p}{m}\adj \right)^{\otimes D} \otimes \exp(-i \pi X/4)\right)H^{\otimes n} \Biggr\rVert_\infty 
    \in O(Dn^{-(1+c)}) \leq O(n^{-c}).
\end{align}
and so the classical distributions produced by computational basis measurements of the states 
\begin{align}
    I_2^{\otimes n-1} \otimes \left(\left(C_m \Urot{\theta'}{m'}\adj \right)^{\otimes D} \otimes \exp(-i \pi X/4)\right) H^{\otimes n}  \ket{\BPM{n}}
\end{align}
and 
\begin{align}
    I_2^{\otimes n-1} \otimes \left(\left(C_m \multiNUrot{\pi/p}{m}\adj \right)^{\otimes D} \otimes \exp(-i \pi X/4)\right)H^{\otimes n} \ket{\BPM{n}}
\end{align}
also differ by at most $O(n^{-c})$ in total variation distance. Combining \Cref{thm:pmmajmod_ideal_sampling} with the fact that $O(p^{3/2}e^{n/4p^2}) \leq O(1/p)$ for $p=n^{-c}$ with $c<1/2$ completes the proof.

\end{proof}

\section{Classical Hardness of Sampling   \texorpdfstring{$(X,\majmod{p}(X) \oplus \parity(X))$}{(X, majmod(X) + parity(X))}}
\label{sec:classical_hardness_with_GHZ}
In this section we prove the classical hardness of sampling from the distribution $(X, \majmod{p}(X)\oplus \parity(X))$ for each prime number $p$, where $X$ is sampled from the uniform distribution over $\{0,1\}^n$.
Recall that the total variation distance distributions $D_1, D_2$ over $\{0,1\}^m$ is
\begin{align}
    \TVD(D_1, D_2) := \max_{T \subseteq \{0,1\}^m} \bigg|\Pr[D_1 \in T] - \Pr[D_2 \in T]\bigg|
\end{align}
By the definition of $\TVD$, each set $T\subseteq \{0,1\}^m$, witnesses a lower bound on $\TVD(D_1, D_2)$ of $\big|\Pr[D_1 \in T] - \Pr[D_2 \in T]\big|$. To prove a lower bound on $\TVD(D_1, D_2)$, we construct a particular $T\in \{0, 1\}^m$ and refer to it as our \emph{statistical test}, and we say that $D_i$ ``passes'' the statistical test with probability $\Pr[D_i \in T]$.

We are interested in the total variation distance between the true distribution $D = (X, \majmod{p}(X) \oplus \parity(X))$, and the output distribution of some local function $f: \{0,1\}^\ell \to \{0,1\}^{n+1}$ that takes a uniformly random $\ell$-bit string $U$ as input. That is, we aim to lower bound $\TVD(f(U), D)$. We prove such a lower bound in the following theorem.

\begin{thm} \label{thm:classical_LB_with_GHZ}
For all $\delta <1$ there exists an $\epsilon >0$ such that for all sufficiently large $n$ and prime number $p = \Theta(n^\alpha)$ for $\alpha\in (\delta/3, 1/3)$: Let $f: \{0, 1\}^{\ell} \to \{0, 1\}^{n+1}$ be an $\epsilon \log(n)$-local function, with $\ell \leq n + n^\delta$. Then $\TVD(f(U), (X, \majmod{p}(X) \oplus \parity(X))) \geq 1/2 - O(1/\log n)$
\end{thm}

\begin{proof}
    This proof follows closely to the analogous proof for $(X, \majmod{p}(X))$ in \cite{viola2012complexity}, with similar notation. Let $\localityf$ be the locality of $f$, $\localityf = \epsilon\log(n)$.
    We start by permuting the outputs, as shown in \cite{viola2012complexity}.  Note that $\circ$~denotes concatenation.
    \begin{lem}[\cite{viola2012complexity}]
        There exists a partition of the input $u\in \{0,1\}^\ell$ into $u=(x,y)$, and permutation of the output bits such that
        \begin{align}
            f(x,y) =  g_1(x_1, y) \circ g_2(x_2, y) \circ \dots \circ g_s(x_s, y) \circ h(y).
        \end{align}
        With $g_i: \{0,1\}\times \{0,1\}^{\ell-s} \to \{0, 1\}^{|\blockidx{i}|}$, $|\blockidx{i}| \leq O(\localityf)$ and $s \geq \Omega(n/\localityf^2)$. 
    \end{lem}    
    We will refer to each $g_i(x_i, y)$ as the $i$th \emph{block} of the output, indexed by $\blockidx{i} \subseteq [n+1]$ in the initial permutation, for $i\in [s]$. Note that if we fix $y$, each block is independent, and block $i \in [s]$ only depends on $x_i$. We say that $g_i$ is \emph{$y$-fixed} for some $y\in \{0,1\}^{\ell-s}$ if $g_i(0, y) = g_i(1, y)$.

    Without loss of generality, and for simplicity of notation, let's assume that the last output bit does not get permuted, so $f(x,y)_{n+1}$ is still the output bit which should (ideally) correspond to $\majmod{p} \oplus \parity$ of the first $n$ outputs, and that it only depends on $y$. Next we define our statistical test.
    
    \paragraph{Statistical Test:} Let $N_0 := 3n^{3\alpha}, N_F := 2n^{3\alpha}$, we define our statistical test as $T := T_0 \cup T_F \cup T_S$, with
    \begin{align}
        T_0 &:=  \{z \in \{0,1\}^{n+1} : z\at{\blockidx{i}} = 0^{|\blockidx{i}|}  \text{ for} \ \leq N_0 \ \text{blocks}\ i\in [s] \}\\
        T_F &:= \{z \in \{0,1\}^{n+1} : \exists (x,y): f(x,y) = z \ \text{and} \ \geq N_F \ \text{blocks} \ g_i(x_i, y) \ \text{are} \ y\text{-fixed} \}\\
        T_S &:= \{(z', b) \in \{0,1\}^n \times \{0,1\} : b \neq \majmod{p}(z') \oplus \parity(z') \} \qquad \text{(``incorrect strings'')}
    \end{align}

    We will show that $f(U)$ passes the statistical test ($f(U)\in T$) with probability at least $1/2 -O(1/\log n)$ and $(X, \majmod{p}(X) \oplus \parity(X))$ passes with probability at most $1/n$. 

    Since both of the functions $\majmod{p}$ and $\parity$ only depend on the Hamming weight of their input, it is useful to define $\MM{p}$ and $\intparity$ as functions over integers, such that $\majmod{p}(z) = \MM{p}(|z|)$ and $\parity(z) = \intparity(|z|)$ for any $z\in \{0,1\}^n$, where we use $|\cdot|$ to denote Hamming weight $|z| = \sum_{i=1}^n z_i$.
        \begin{align}
            \MM{p}(j) := \begin{cases}
            0 & \text{if } j  < p/2 \mod p \\
            1 & \text{if } j > p/2  \mod p
            \end{cases},
            &&
            \intparity(j) := j \mod 2,
            &&
            \text{for } j\in \mathbb{Z}.
        \end{align}

    Upon fixing $y$, the Hamming weight $|f(x,y)|_{1:n}$ is a sum of independent random variables $|g_i(x_i, y)|$ which take on at most 2 different values. The following Fact, Corollary, and Lemma will be useful in analyzing this independent sum of random variables in the context of the $\majmod{p}\oplus \parity$ function. 
    \begin{fact} [Fact 3.2 in \cite{viola2012complexity}]\label{fact:puniform}
        Let $a_1, a_2, \dots a_t$ be nonzero integers modulo $p$, and let $(x_1, x_2, \dots, x_t) \in \{0,1\}^n$ be sampled uniformly. Then the total variation distance between $\sum_{i=1}^t a_ix_i \mod p$ and $U_{p}$, the uniform distribution over $\{0, 1, \dots, p-1\}$ is at most $\sqrt{p} e^{-t/p^2}$
    \end{fact}
    \begin{cor}\label{cor:prob_sum_in_set}
        For each prime $p = \Theta(n^\alpha)$ with $\alpha <1$,  $t = \Omega(p^3)$, $a_0, a_1, \dots a_t$ nonzero integers modulo $p$, and $A \subseteq \{0, 1, \dots p-1\}$
        \begin{align}
            \frac{|A|}{p} - O(1/n) \leq \Pr_{x\in \{0, 1\}^t}\left[a_0 + \sum_{i=1}^t a_i x_i \in A\right] \leq \frac{|A|}{p} + O(1/n)
        \end{align}
    \end{cor}
    \begin{proof}
        By the definition of total variation distance, it is sufficient to prove that $\TVD(U_p, a_0 + \sum_{i=1}^t a_i x_i) \leq O(1/n)$. 
        \begin{align}
            \TVD(U_p, a_0 + \sum_{i=1}^t a_i x_i) \leq \sqrt{p}e^{-t/p^2} = \sqrt{p} e^{-\Omega(p)} = \Theta(n^{\alpha/2}) e^{-\Omega(n^{\alpha})} \leq O(1/n).
        \end{align}
    \end{proof}
    
    \begin{lem} \label{lem:MMP_of_sum} 
        For each $\alpha \in (0, 1)$, and prime number $p=\Theta(n^\alpha)$, define the sums $S = a_0 + \sum_{i=1}^t a_i x_i$ and $U = u_0 + \sum_{i=1}^t u_i x_i$. Also let $t= \Omega(p^3)$ and $a_0, a_1, \dots, a_t$ and $u_0, u_1, \dots, u_t$ be integers with $0 < a_i \leq O(p/\log n)$ for each  $i\in [t]$. Then
        \begin{align}
            \Pr_{x} [\MM{p}(S) \oplus 
            \intparity(U) = b] \geq \frac{1}{2} - O(1/\log n). 
        \end{align}
    \end{lem}
    \begin{proof}
    Let's consider the case that at least $1/2$ of the $u_i$ for $i\in [t]$ are even. Then we arbitrarily fix all $x_i$ such that $u_i$ is odd, and let $E = \{i \in [t] : u_i \text{ even}\}$. Note that now the parity is fixed to $c := \intparity(u_0 + \sum_{i \in [t]\setminus E} u_ix_i)$. Let $a_i' = a_{E_i}$ for each $i\in \{1, 2, \dots, |E|\}$, and $a_0' = a_0 + \sum_{i \notin E} a_i$.
    \begin{align}
        \Pr_{x_E} \left[\MM{p}(S) \oplus 
        \intparity(U) = b\right] &= \Pr_{r\in \{0,1\}^{|E|}} \left[\majmod{p}(a'_0 + \sum_{i=1}^{|E|} a_i'r_i ) \oplus c = b\right]\\
        &= \Pr_r\left[a'_0 + \sum_{i=1}^{|E|} a_i'r_i \in M_{c\oplus b}\right]
    \end{align}
    Where $M_0 = \{0, 1, \dots, (p-1)/2\}$ and $M_1 = \{(p+1)/2, \dots, p-2, p-1\}$. Since $|M_0| = (p+1)/2$, $|M_1|=(p-1)/2$, and $|E| = \Theta(n^{\alpha})$, it follows from \Cref{cor:prob_sum_in_set} that
    \begin{align}
        \Pr_{x_E} \left[\MM{p}(S) \oplus 
        \intparity(U) = b\right] \geq (p-1)/2p - O(1/n)= 1/2 - O(1/n^\alpha).
    \end{align}
    All that's left is to consider the case where more than half of the $u_i$ for $i \in [t]$ are odd. In this case we will fix $x_i$ for each $i\in [t]$ with $u_i$ even, setting $a_0' := a_0 + \sum_{i\in E} S_i$, and $u_0' = u_0 + \sum_{i\in E} u_i$. We denote the set of indices of such ``odd'' elements as $O = \{i\in [t] : u_i \text{ odd} \}$, and set $a_i' = a_{O_i}$ and $u_i' = u_{O_i}$ for each $i\in [|O|]$. Note that since each $u_i'$ is odd, we have $\intparity(u_0'+ \sum_{i\leq t} u_i'r_i) = u_0' + (\parity(r_1, \dots, r_{|O|})) \mod 2$, which is denoted as $\parity(r) \oplus u_0'$.
   \begin{align}
       \Pr_{x_O}\left[\MM{p}(S) \oplus 
       \intparity(U) = b\right] =& \Pr_{r\in \{0,1\}^{|O|}} \left[\majmod{p}\bigg(a_0'{+}\sum_{i\leq t} a_i'r_i\bigg) \oplus \parity(r) = b\oplus u_0'\right]\\
        =& \frac{1}{2}\Pr_r\left[\MM{p}\bigg(a_0'{+}\sum_{i\leq t} a_i'r_i \bigg) = b \oplus u_0'\bigg| \parity(r) = 0\right] \\
        &+  \frac{1}{2}\Pr_r\left[\MM{p}\bigg(a_0'+ \sum_{i\leq t} a_i'r_i \bigg) \neq b \oplus u_0'\bigg| \parity(r) = 1\right]
    \end{align}
    Sampling a uniformly random $t$ bit string $r_1r_2\dots r_t$ with even Hamming weight is equivalent to sampling the first $t-1$ bits uniformly at random, and setting the last bit to $r_t =\parity(r_1, \dots, r_{t-1})$. So the equation above is equal to
    \begin{align}
        =&\frac{1}{2} \Pr_{r_1, \dots r_{t-1}}\left[\majmod{p}\bigg(a_0'{+}\sum_{i=1}^{|O|-1} a_i'r_i  + a'_{|O|}\cdot\parity(r_1,\dots ,r_{t-1})\bigg)= b \oplus u_0'\right] \label{eq:MMP_of_sum_eq1} \\
        &+  \frac{1}{2}\Pr_{r_1, \dots r_{t-1}}\left[\majmod{p}\bigg(a_0'+ \sum_{i=1}^{|O|-1} a_i'r_i  + a'_{|O|}\cdot\parity(1,r_1,\dots, r_{t-1})\bigg) \neq b \oplus u_0'\right]. \label{eq:MMP_of_sum_eq2}
    \end{align}
    For any fixed $r_1, ..., r_{t-1}$ at least one of the expressions above must be satisfied unless we have
    \begin{align}
        \majmod{p}\bigg(a_0'+ \sum_{i=1}^{|O|-1} a_i'r_i \bigg) &\neq \majmod{p}\bigg(a_0'+ \sum_{i=1}^{|O|-1} a_i'r_i  + a'_{|O|} \bigg) \\
        \implies\left(a_0' +\sum_{i=1}^{|O| -1} a_i' r_i \right) &\in [p_/2 - a_{|O|}', p/2] \cup [p - 1 - a_{|O|}', p -1].
    \end{align}
    Then we can lower bound the sum in \Cref{eq:MMP_of_sum_eq1,eq:MMP_of_sum_eq2} as
    \begin{align}
        \labelcref{eq:MMP_of_sum_eq1} + \labelcref{eq:MMP_of_sum_eq2}\geq& \frac{1}{2} \left(\Pr\left[ a_0'+ \sum_{i=1}^{|O|-1} a'_ir_i \in [0, p/2 - a'_{|O|})\right] + \Pr\left[ a_0'+ \sum_{i=1}^{|O|-1} a'_ir_i \in (p/2, p-1 - a'_{|O|}]\right] \right)\\
        \geq& \frac{1}{2p}((p+1)/2 - a'_{|O|} + (p-1)/2 - a'_{|O|}) - O(1/n)\\
        =& \frac{1}{2} - \frac{a'_{|O|}}{2p} - O(1/n)\\
        =& \frac{1}{2} - \frac{O(p/\log n)}{2p} - O(1/n) \geq \frac{1}{2} - O(1/\log n).
   \end{align}
   Where we used \Cref{cor:prob_sum_in_set}, and the Lemma's assumption that $ 0< a_i \leq p/\log n$ for each $i\in [t]$ and $p = \Theta(n^\alpha)$. 
\end{proof}

We are now ready to prove the following claims. 
    \begin{claim}\label{claim:f_passes}
    $\Pr[f(U)\in T] \geq 1/2 - O(1/\log n)$
    \end{claim}
    
    \begin{proof}
        We will show that for each $y$, $\Pr_x[f(x, y) \in T] \geq 1/2 - 1/\log n$. Suppose we fix $y$ arbitrarily.
        
        If $y$ fixes at least $N_F$, blocks $g_i(x_i, y)$, then $\Pr_x[f(x,y) \in T_F] =1$. Moreover, if there are $\leq N_0$ blocks $g_i$ such that $g_i(x_i, y) = 0^{|\blockidx{i}|}$ for some $x_i\in \{0,1 \}$, then for each $x$, there will also be $\leq N_0$ blocks with $g_i(x_i, y) = 0^{|\blockidx{i}|}$, so $\Pr_x[f(x,y) \in T_0] = 1$.
        
        Therefore, we assume that there are $< N_F$ blocks $g_i$ that are $y$-fixed, and $>N_0$ blocks with $g_i(x_i, y)=0^{|\blockidx{i}|}$ for some $x\in \{0,1\}^s$. Thus, there are more than $N_0 - N_F = n^{3\alpha}$ blocks $g_i$ such that for some $x_i\in \{0,1\}$, $g_i(x_i, y) = 0^{|\blockidx{i}|}$ and $g_i(1-x_i, y) \neq 0^{|\blockidx{i}|}$. 
        Let $J\subseteq [s]$ denote this subset of blocks, with $|J|\geq n^{3\alpha}$. 
        We arbitrarily fix the $x_i$ for $i\in [s]\setminus J$. Now, the total Hamming weight of the first $n$ bits of $f(x,y)$ (denoted as $|f(x,y)_{1:n}|$) only depends on the $x_i$ for $i\in J$.

        Let $S_i$ denote the Hamming weight of the $i$th block for each $i\in [s]$. Note that for each $i\in J$, $S_i =0$ with probability $1/2$, and $S_i$ is some positive integer modulo $p$, with probability $1/2$, since $|\blockidx{i}| \leq O(\localityf) = O(\epsilon \log n) < p$. Moreover, for each $i\in [s]\setminus J$, $S_i$ is fixed. Therefore,
        \begin{align}
            |f(x,y)_{1:n}| = a + \sum_{j\in J} |g_i(x_i, y)| = a + \sum_{i\in J} S_i
        \end{align}
        for some positive integer $a$ that does not depend on $\{x_i\}_{i\in J}$.
        
        Since the last bit $b:= f(x,y)_{n+1}$ is fixed, the correctness of the output is determined by the $\majmod{p}$  and $\parity$ of $f(x,y)_{1:n}$. 
        We have that $f(x,y)\in T_S \iff \MM{p}(a + \sum_{i\in J} S_i) \oplus \intparity(a + \sum_{i\in J} S_i) \neq b$. Note that we can write $a + \sum_{i\in J} S_i = a + \sum_{i\leq |J|} a_ir_i$ for some uniformly random $r\in\{0,1\}^{|J|}$, and for each $a_i$ a fixed positive integer mod $p$.  Therefore,
        \begin{align}
            \Pr_{x_J}[f(x,y) \in T_S] &= \Pr_{r\in \{0, 1\}^{|J|}}[\majmod{p}(a + \sum_{i=1}^{|J|} a_i r_i) \oplus \intparity(a + \sum_{i=1}^{|J|} a_i r_i) \neq b].
        \end{align}
        
        Furthermore, each $a_i$ is at most $O(\localityf) = O(\epsilon\log n)$ since $|B_j|\leq O(\localityf)$ for each $j\in [s]$.
        By \Cref{lem:MMP_of_sum}, it immediately follows that $\Pr_{x_J}[f(x,y) \in T_S] \geq \frac{1}{2} - O(1/\log n)$. 
       In conclusion, we've showed that after arbitrarily fixing $y$, $\Pr_x[f(x,y) \in T]\geq \frac{1}{2} - O(1/\log n) $. Therefore, $\Pr_{x,y}[f(x,y) \in T] \geq \frac{1}{2} - O(1/\log n)$, as desired.
        
    \end{proof}
    
    \begin{claim}\label{claim:mmp_fail_stat_test}
        $\Pr[(X, \majmod{p}(X) \oplus \parity(X)) \in T] \leq O(1/n)$
    \end{claim}
    \begin{proof}
        This proof follows that of Claim 3.3 in \cite{viola2012complexity}.
        Let $D := (X, \majmod{p}(X) \oplus \parity(X))$. By the union bound $\Pr[D\in T] \leq \Pr[D\in T_0] + \Pr[D\in T_F] + \Pr[D\in T_S]$. Clearly $\Pr[D\in T_S] = 0$, since $T_S$ is the set of invalid strings. Therefore, it is sufficient for us to show that $\Pr[D\in T_F], \Pr[D\in T_0] \leq \frac{1}{2n}$. 
        
        $\Pr[D\in T_F] = |T_F|/2^{n+1} \leq |T_F|/2^n$, so it is sufficient to upper bound $|T_F|$. Recall that $z\in T_F$ if $z= f(x,y)$ for some $x,y$ such that at least $N_F$ blocks are $y$-fixed. Thus each $z\in T_F$ is characterized by $y$, and the bits of $x$ that do not belong to fixed blocks. That is, we need at most $\ell - N_F$ bits to characterize $z$. Since $\ell\leq n + n^\delta$ and $N_F = 2n^{3\alpha}$,
        \begin{align}
            |T_F| &\leq 2^{ n + n^{\delta} - 2n^{3\alpha}}\\
            &\leq 2^{n - n^{3\alpha}}
        \end{align}
        since $\delta < 3\alpha$. So
        \begin{align}
            \Pr[D\in T_F] \leq 2^{-n^{3\alpha}} \leq \frac{1}{2n}.
        \end{align}
        All that's left is to bound $\Pr[D\in T_0]$, the probability that at most $N_0 = 3n^{3\alpha}$ blocks $i$ are all zero, $D_{\blockidx{i}} = 0^{|\blockidx{i}|}$. Since the first $n$ bits of $D$ are independently random, the probability that the block $D_{\blockidx{i}}$ is all zero is independent of other blocks $D_{\blockidx{j}}$ for $i\neq j \in [s]$. The probability that block $i\in [s]$ is all zero is
        \begin{align}
            \Pr[D_{\blockidx{i}} = 0^{|\blockidx{i}|}] = (1/2)^{|\blockidx{i}|}\geq (1/2)^{O(\localityf)}= (1/2)^{O(\epsilon \log n)} = \left(\frac{1}{n}\right)^{O(\epsilon)}.
        \end{align}
        Now, the probability that at most $N_0 = 3n^{3\alpha}$ are all zero is 
        \begin{align}
            \Pr[D \in T_0] &= \Pr[\bigcup_{\substack{T\subseteq [s]:\\|T|=N_0}}  \{ D_{\blockidx{i}} \neq 0^{|\blockidx{i}|} \text{ for each } i\in [s]\setminus T \} ]\\
            &\leq \binom{s}{N_0} \left(1 - \frac{1}{n^{O(\epsilon)}}\right)^{s-N_0}\\
            &\leq \binom{s}{N_0} e^{-\frac{s-N_0}{n^{O(\epsilon)}}}.
        \end{align}
        Since $s\geq \Omega(N/\localityf^2) = \Omega(\frac{n}{\epsilon^2 \log^2 n})$, $s\leq n$ and $N_0 = 3n^{3\alpha}$,
        \begin{align}
            &\leq \binom{n}{3n^{3\alpha}} e^{-n^{-O(\epsilon)}(\frac{n}{\epsilon^2 \log^2 n}-3n^{3\alpha})}\\
            &\leq \left(\frac{n}{3n^{3\alpha}}\right)^{3n^{3\alpha}} e^{-n^{1-O(\epsilon)}/\log^2 n} e^{3n^{3\alpha}}\\
            &\leq n^{3n^{3\alpha}}e^{-n^{1-O(\epsilon)}/\log^2 n}\\
            &\leq \frac{1}{2n}
        \end{align}
        for sufficiently large $n$ and small $\epsilon$. In conclusion, $\Pr[D\in T]\leq \frac{1}{n}$, as desired.
    \end{proof}
    Using \Cref{claim:f_passes,claim:mmp_fail_stat_test}, we can lower bound the total variation distance between the target distribution $D = (X, \majmod{p}(X) \oplus \parity(X))$ and $f(U)$.
    \begin{align}
        \TVD(D, f(U)) &\geq \left|\Pr[f(U) \in T] - \Pr[D \in T]\right|\\
        &\geq \frac{1}{2} - O(1/\log n),
    \end{align}
    completing the proof of \Cref{thm:classical_LB_with_GHZ}. 
\end{proof}

\section{Classical Hardness of Sampling \texorpdfstring{$(Z,\pmmajmod{p}(Z))$}{(Z, pmmajmod(Z))}} 

\label{sec:classical_hardness_main}
This section concerns the hardness of classically sampling from the distribution $(Z, \pmmajmod{p}(Z))$, where $Z$ is a random variable $Z\sim \unif(\{0,1\}^N)$ and the function $\pmmajmod{p}$ is defined in \Cref{defn:pmmajmod}, and recalled below.

\paragraph*{$\pmmajmod{p}$}
The input to $\pmmajmod{p}$ is an $N = 2n-2$ bit string, $(x_1, x_2, \dots x_{n-1}, d_1, d_2, \dots, d_{n-1})$. Each $x_i$ corresponds to the vertex $v_i$ of the balanced binary tree $\bintree[n]$, and each $d_i$ corresponds to the edge $e_i$ of $\bintree[n]$ per the ordering in \Cref{defn:binary_tree}. 
\begin{align}
    \pmmajmod{p}(x, d) = \MM{p}\left(\sum_{i=1}^{n-1} x_i (-1)^{\pathsum{d}_i}\right) \oplus \parity(x) && x, d\in \{0,1\}^{n-1}. 
\end{align}
Where $\MM{p}$ is defined in \Cref{defn:btpmGHZ} and $\pathsum{d}$ is defined in \Cref{defn:binary_tree}.

In \Cref{sec:classical_hardness_with_GHZ} we proved the classical hardness of sampling from the slightly different distribution $(X, \majmod{p}(X) \oplus \parity(X))$ where $X\sim \unif(\{0,1\}^n)$. For the sake of comparing with $\pmmajmod{p}$ we list this function below.
\paragraph*{$\majmod{p}\oplus \parity$}
\begin{align}
    \majmod{p}(x) \oplus \parity(x) = \MM{p}\left(\sum_{i =1}^n x_i \right) \oplus \parity(x) && x\in \{0,1\}^n
\end{align}
Both of these distributions have the form $(Y, \MM{p}(S_Y) \oplus \parity(Y))$ for a uniformly random bitstring $Y$, and $S_Y$ a sum that depends on $Y$. For the $\majmod{p}(S_x)\oplus \parity(x)$ function, the relevant sum is simply the Hamming weight of the input $x\in \{0,1\}^n$, denoted as $|x|$. A nice property of the Hamming weight, $|x| = \sum_{i} x_i$ is that each of the terms in the sum depends on a different bit of the input, and thus if many of the bits of $x_i$ are sampled independently, then so are their corresponding terms in the sum. The key challenge in applying the framework from the proof of \Cref{thm:classical_LB_with_GHZ} is that the terms in $S = \sum_i x_i (-1)^{\pathsum{d}_i}$ no longer depend on disjoint variables. In particular, flipping the bit $d_j$ corresponding to edge $e_j$ flips the sign of all terms $x_i (-1)^{\pathsum{d}_i}$ for $v_i$ downstream from $e_j$ in the balanced binary tree $\bintree[n]$. To accommodate for this dependence, we will partition the tree $\bintree[n]$ into subtrees, then identify subtrees corresponding to output variables which are independent when a large chunk of the input variables are fixed. 

We show that for some choice of $p$, any function $f$ which takes as input a uniformly random $(N + N^\delta)$-bit string and is $(1/2 - \omega(1/\log n))$-close in total variation distance with $(Z, \pmmajmod{p}(Z))$, must have locality $\localityf\geq \Omega(\log^{1/2} N)$. If we consider $f$ as a classical circuit with fan-in 2, this corresponds to a circuit depth lower bound of $\Omega(\log \log N)$. 

\begin{thm} \label{thm:classical_LB_tree}
    For each $\delta < 1$, there exists an $\epsilon >0$ such that for all sufficiently large even integer $N$ and prime number $p = \Theta(N^\alpha)$ for $\alpha \in (\delta/3, 1/3)$: Let $f:\{0,1\}^\ell \to \{0,1\}^{N+1}$ be an $(\epsilon \log N)^{1/2}$-local function, with $\ell \leq N + N^\delta$. Then $\TVD(f(U), (Z, \pmmajmod{p}(Z))) \geq 1/2 - O(1/\log N)$.
\end{thm}
\begin{proof}
    The function $f$ takes input an $\ell$-bit string we label as $(u_1, u_2, \dots, u_\ell)$ and outputs $(N+1)$-bit output string we label as $(z_1, \dots, z_{N}, b)$. Let $n$ be the integer such that $N = 2n-1$. Just as in the definition of $\pmmajmod{p}$ in \Cref{defn:pmmajmod}, we consider the partition of $z = (x, d) \in \{0, 1\}^{n-1} \times \{0,1\}^{n-1}$, where $x_1, \dots, x_{n-1}$ are the first $n-1$ bits of $z$, and $d_1, \dots d_{n-1}$ are the next $n-1$ bits of $z$, and $b\in \{0, 1\}$ is the last bit which is considered  ``correct'' if $b = \pmmajmod{p}(z)$. 

    The output variables $x_1, \dots, x_{n-1}$ are associated with $v_1, \dots, v_{n-1}$, the non-root vertices of the balanced binary tree $\bintree[n]$. The output variables $d_1, \dots, d_{n-1}$ are associated with the edges $e_1, \dots, e_{n-1}$, where we use the ordering as defined in \Cref{defn:binary_tree}. As is standard in graph theory, for any graph $G$ we use $V(G)$ and $E(G)$ to denote $G$'s vertices and edges respectively. To understand the correlations between each of the output bits $z_i$, it is useful to partition $\bintree[n]$ as follows. 
    
    \begin{defn}[$\bintree$ partition $(T_0, T_1, \dots, T_k)$]
        Let $D:= \log(2\localityf)$, we partition the vertices of the balanced binary tree $\bintree[n]$ into the bottom $D$ layers and the top $\log n - D$ layers as shown in  \Cref{fig:tree_partition}. Let the \emph{top tree} $T_0$ be the tree induced by the top $\log(n) - \log(2\localityf)$ layers of vertices in $\bintree[n]$. The subgraph induced by the bottom $D$ layers is a forest of trees which we denote as $\allsmalltrees = \{T_1, T_2, \dots, T_k\}$ and refer to as the \emph{small trees}. In order to make sure that each edge and vertex of $\bintree[n]$ is accounted for in $\{T_0\} \cup \allsmalltrees$, for each $i\in [k]$ we consider the edge which connects the root of $T_i$ to a leaf of $T_0$ as in the small tree $T_i$. Thus, each small tree $T\in \allsmalltrees$ has an edge with the root of $T$ as its only endpoint as shown in \Cref{fig:tree_partition}.  
    \end{defn}

    \newcommand{\drawtree}[1][]{
        \node [draw, isosceles triangle, isosceles triangle apex angle=50, shape border rotate=90, inner sep = -1.1cm, anchor=north, rounded corners =1.9cm, fill=none] {
            \begin{tikzpicture}[every node/.style={circle,draw, fill, inner sep=0.7mm},level 1/.style={sibling distance=3cm,level distance=2cm},level 2/.style={sibling distance=1.5cm,level distance=1.2cm},level 3/.style={sibling distance=0.8cm,level distance=1.2cm}, level 4/.style={sibling distance=0.5cm,level distance=0cm}, level 5/.style={level distance=0cm}, level 6/.style={level distance=2.3cm}]
            \node (main) {} 
                    child {
                        node {} 
                        child {
                            node {}   
                            child {
                                node {} 
                                child {
                                    node[fill=none, draw=none] {} edge from parent[dotted]
                                }
                                child {
                                    node[fill=none, draw=none] {} edge from parent[dotted]
                                    child {
                                        node[fill=none, draw=none] {$\vdots$} edge from parent[draw=none]
                                        child {
                                        node[fill=none, draw=none] {} edge from parent[draw=none] 
                                    }
                                    }
                                }
                            }
                            child {
                                node {} 
                                child {
                                    node[fill=none, draw=none] {} edge from parent[dotted]
                                }
                                child {
                                    node[fill=none, draw=none] {} edge from parent[dotted]
                                }
                            }
                        }
                        child {
                            node {}  
                            child {
                                node {} 
                                child {
                                    node[fill=none, draw=none] {} edge from parent[dotted]
                                }
                                child {
                                    node[fill=none, draw=none] {} edge from parent[dotted]
                                    child {
                                        node[fill=none, draw=none] {$\vdots$} edge from parent[draw=none]
                                    }
                                }
                            }
                            child {
                                node {} 
                                child {
                                    node[fill=none, draw=none] {} edge from parent[dotted]
                                }
                                child {
                                    node[fill=none, draw=none] {} edge from parent[dotted]
                                }
                            }
                        }
                    }
                    child {
                        node {} 
                        child {
                            node {}  
                            child {
                                node {} 
                                child {
                                    node[fill=none, draw=none] {} edge from parent[dotted]
                                }
                                child {
                                    node[fill=none, draw=none] {} edge from parent[dotted]
                                    child {
                                        node[fill=none, draw=none] {$\vdots$} edge from parent[draw=none]
                                    }
                                }
                            }
                            child {
                                node {} 
                                child {
                                    node[fill=none, draw=none] {} edge from parent[dotted]
                                }
                                child {
                                    node[fill=none, draw=none] {} edge from parent[dotted]
                                }
                            }
                            }
                        child {
                            node {}  
                            child {
                                node {} 
                                child {
                                    node[fill=none, draw=none] {} edge from parent[dotted]
                                }
                                child {
                                    node[fill=none, draw=none] {} edge from parent[dotted]
                                    child {
                                        node[fill=none, draw=none] {$\vdots$} edge from parent[draw=none]
                                    }
                                }
                            }
                            child {
                                node {} 
                                child {
                                    node[fill=none, draw=none] {} edge from parent[dotted]
                                }
                                child {
                                    node[fill=none, draw=none] {} edge from parent[dotted]
                                }
                            }
                        }
                    };
                \ifthenelse{\equal{#1}{small}}{
                    \draw (main) -- +(0,2);
                }{}
                \end{tikzpicture}
                };
    }
    \begin{figure}[hbtp]
        \begin{center}
            \begin{tikzpicture}
                \node[scale=0.7]{
                \begin{tikzpicture}[fill=none, draw=none]
                    \node[label={[yshift=-.6cm]\Large{$T_0$}}] (T0) {
                        \begin{tikzpicture}
                            \drawtree
                        \end{tikzpicture}
                    };

                    \foreach \n/\treelabel in {0/1, 1/2, 2/3, 4.4/k} {
                        \node[scale=0.2, label={\Large $T_\treelabel$}] (T\treelabel) at (-6.5 + 3*\n, -7) {
                            \begin{tikzpicture}
                                \drawtree[small]
                            \end{tikzpicture}
                        };
                    }
                    \node at (-6.5 + 3*3.2, -7) {\ldots};
                    \draw [|-|] ([yshift=-1.8cm, xshift=0.2cm]T0.north east) -- ([yshift=.5cm, xshift=0.2cm]T0.south east) node [midway,xshift=0.6cm,right] {\Large $\log n - D$};
                    \draw [|-|] ([yshift=-.5cm, xshift=0.4cm]Tk.north east) -- ([yshift=.2cm, xshift=0.4cm]Tk.south east) node [midway,xshift=0.6cm,right] {\Large $D$};
                    \draw[decorate,decoration={brace,amplitude=10pt,mirror}]
                    ([yshift=.8cm, xshift=-0.2cm]T1.north west) -- ([yshift=-.5cm, xshift=-0.2cm]T1.south west) node [midway,xshift=-0.4cm,left] {\Large $\allsmalltrees$};
                \end{tikzpicture}
                };
        \end{tikzpicture}
        \end{center}
        \caption{Partition of the balanced binary tree $\bintree[n]$ into $k+1$ subtrees. The top tree $T_0$ consists of the subtree induced by the first $\log n - D$ layers of $\bintree$. The $k$ bottom trees $\allsmalltrees = \{T_1, T_2, \dots, T_k\}$ include all vertices in the bottom $D$ layers of $\bintree$ and all incident edges. Note that for each $i\in [h]$, $T_i$ contains a single edge that only has one endpoint, this edge corresponds to the edge in $\bintree$ that connects the root of $T_i$ with its parent in $T_0$.}
        \label{fig:tree_partition}
    \end{figure}
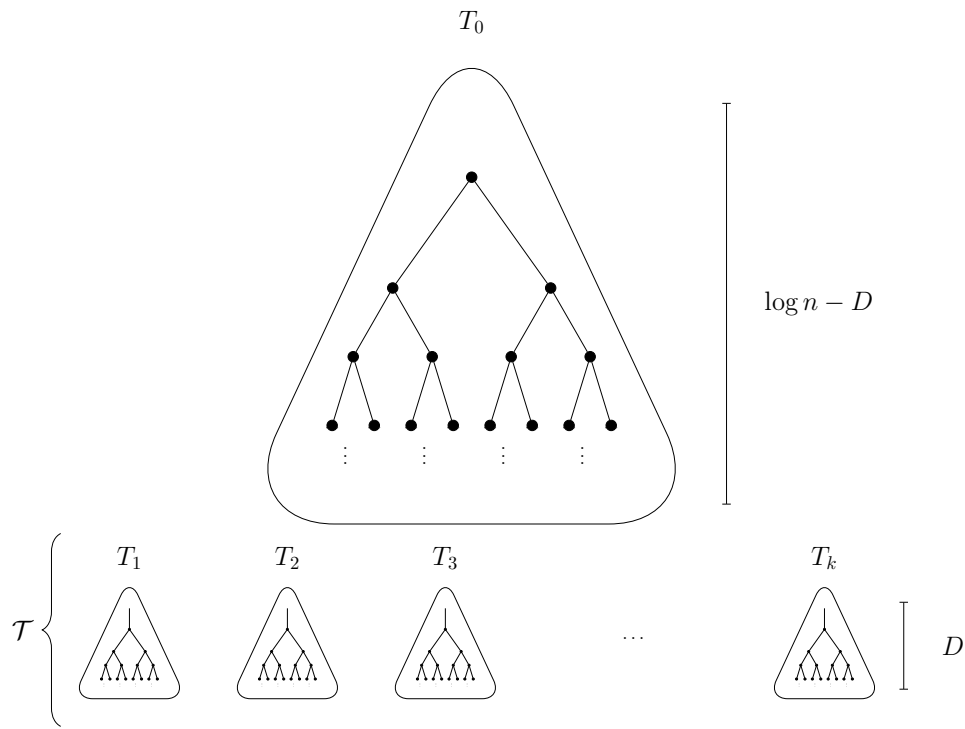

    Although a subtree $T$ of $\bintree[n]$ consists of vertices and edges labeled as $\{v_i\}_i$ and $\{e_i\}_i$, we slightly abuse notation and say that the output variable $z_i$ is ``in'' $T$ (denoted $z_i \in T$) if the edge or vertex which is associated with $z_i$ is in $E(T) \cup V(T)$. And will sometimes use $T$ to denote the subset of variables $\{z_i\}$ which are associated with the tree $T$. Moreover, we define the \emph{size} of a subtree $T$ of $\bintree[n]$ be $|T| = |V(T)| + |E(T)|$. Note that since each $T\in \allsmalltrees$ has an extra edge at the root, with no other endpoint, $|E(T)| = |V(T)| \leq 2\localityf$. 

    The top tree $T_0$ has $|V(T_0)| = 2^{\log n - D} - 1 = \frac{n}{2\localityf} - 1$ vertices, and $|E(T_0)| = |V(T_0)| -1 = \frac{n}{2\localityf} -2$ edges. For each $i\in [k]$ the small tree $T_i$ has at most $2^D -1 = 2\localityf - 1$ vertices $V(T_i)$, and the same number of edges $|E(T_i)| = |V(T_i)| = 2\localityf - 1$. In total, the top tree has size $|T_0| \leq n/\localityf - 3$ and each bottom tree $T_i\in \allsmalltrees$ has size at most $|T_i| \leq 4\localityf$. Since the root vertex of each small tree is at the $(\log n - D + 1)$-level of the balanced binary tree $\bintree[n]$, there are $k = 2^{\log n - D} = n/\localityf$ small trees.

    For each output variable $z_i$ in the string $z$, we consider the other output variables which are in the same tree as $z_i$ as the \emph{tree neighborhood} of $z_i$. 
    \begin{defn}[Tree Neighbors, $\treeneighb$] \label{defn:tree_neighbours}
        For each variable $z_i$ for $i\in[N]$, let $\treeneighb(z_i) \subseteq \{z_i\}_{i\in [N]}$, be the subset of outputs in the same tree $T\in \allsmalltrees \cup \{T_0\}$ as $z_i$. Moreover, for any subset of outputs $S \subseteq \{z_i\}_{i\in[N]}$, let $\treeneighb(S) := \bigcup_{z_i \in S} \treeneighb(z_i)$.
    \end{defn}
    Recall that the variables $\{z_i\}_{i\in[N]}$ only correspond to the \emph{non-root} vertices of $\bintree[n]$, but the root vertex $v_0$ is in the top tree $T_0$. Thus for vertices $v_j, v_k$ corresponding to the left and right children of root $v_0$, we have that $z_j \in \treeneighb(z_k)$, despite there being no variable in $\treeneighb(z_k)$ associated with the root.
    Note that for any output in a small tree $z_i\in \bigcup_{T \in \allsmalltrees} T$, $\treeneighb(z_i)$ has size at most $2\localityf$ since $|T| \leq 2\localityf$ for each $T\in \allsmalltrees$. Moreover, for any subset of small tree outputs $S\subseteq \bigcup_{T \in \allsmalltrees} T$, $|\treeneighb(S)| \leq 2\localityf|S|$.

    \begin{defn}[Forest Partition] \label{defn:forest_partition}
        $F_0, F_1, \dots, F_s \subseteq \{z_i\}_{i\in [N]}$ is a \emph{forest partition} if both of the following hold.
        \begin{enumerate}
            \item $F_0, \dots, F_s$ is a partition of all variables $\{z_i\}_{i\in [N]}$
            \begin{align}
                F_0 \uplus \dots \uplus F_s = \{z_i\}_{i\in [N]}
            \end{align}
            \item Each $F_i$ contains a union over a subset of trees from $\allsmalltrees \cup \{T_0\}$.
            \begin{align}
                \treeneighb(F_i) = F_i && \text{for each } i \in [s]
            \end{align}
        \end{enumerate}
    \end{defn}

    The next lemma shows that we can construct a forest partition with the property that, after a large fraction of the input bits to our $(\epsilon \log N)^{1/2}$ local function have been fixed, each of the remaining unfixed bits controls a single (independent) subset of trees in the partition.

    \begin{lem}\label{lem:partition_inputs_tree}
        There exists a forest partition $\forest{0}, \forest{1}, \dots, \forest{s}$
        for some $s \geq \Omega(N/\localityf^3)$, with $|\forest{i}| \leq O(\localityf^2)$ for each $i\in [s]$; and a partition of the input $u\in \{0,1\}^\ell$ into $u= (w, y)$, with $w\in \{0,1\}^s$ such that 
        \begin{align}
            f(w, y) \big\vert_{\forest{0}} &= h(y),\\
            f(w,y)\big\vert_{\{N+1\}} &= b(y),\\ 
            f(w,y)\big\vert_{\forest{i}} &= g_i(w_i, y) && \text{for each } i\geq 1,\\
            \text{and} \quad T_0 &\subseteq F_0.
        \end{align}
        For some functions $h:\{0,1\}^{\ell - s} \to \{0,1\}^{|\forest{0}|}, b:\{0,1\}^{\ell - s} \to \{0,1\}$, and $g_i : \{0,1\} \times \{0,1\}^{\ell - s} \to \{0,1\}^{|\forest{i}|}$ for each $i\in [s]$.
    \end{lem}
        We refer to $g_i(w_i, y)$ as the $i$th \emph{block} of the output, assigning values to the variables in $\forest{i}$, for $i\in [s]$. Note that if we fix the input $y$, each block $g_i(w_i, y)$ is a function only of the input bit $w_i$. Since the input $w\in \{0,1\}^s$ is uniformly random, the value of each of the blocks is independent conditioned on $y$. 
    \begin{proof}[Proof of \Cref{lem:partition_inputs_tree}]
        Consider the bipartite graph with the $\ell$ input variables to $f$ as the left vertices, and the $N+1$ output variables as the right vertices, where each input $j\in [\ell]$ and output $i\in [N+1]$ vertex share an edge iff the $i$th output bit of $f$, denoted as $f_i$ is a function of the $j$th input bit. We refer to this graph as the \emph{input-output dependency graph} of $f$. For each vertex $v$ in the dependency graph, let the \emph{neighborhood} of $v$, $\fneighb(v)$, be the set of vertices adjacent to $v$. Similarly, for any subset $S$ of vertices, let $\fneighb(S) := \bigcup_{v\in S}\fneighb(v)$. Since by assumption, $f$ is $\localityf$-local, the degree of the output vertices is at most $\localityf$.

        Let $L$ be the set of input vertices of the dependency graph for $f$ which are adjacent to the output vertices in $T_0$ or $b$, that is $L:= \fneighb(T_0 \cup \{b\})$ (or we could associate $b$ with the root $v_0$ in $T_0$). If we fix the inputs in $L$, then $b$, and the outputs in $T_0$ are also fixed. For this reason we refer to $L$ as the \emph{fixed} inputs, and the remaining inputs $U = \{u_i\}_{i\in [\ell]} \setminus L$ as the \emph{unfixed} inputs. 
        \begin{align}
            |L| \leq \localityf( |T_0|) \leq \localityf \left(|V(T_0)| + |E(T_0)|\right) \leq n - 3\localityf.
        \end{align}
        Therefore, there are at least $N - |L| \geq 2n -1 - (n - 2\localityf) \geq  n$ unfixed inputs $U$. Since $|V(T_0)| = \frac{n}{2\localityf} -1$, and $|E(T_0)| = |V(T_0)| - 1$.

        As mentioned above, the locality of $f$ implies that the degree of the output vertices in the dependency graph is at most $\localityf$. Using the following claim, we can also bound the degree of half of the input vertices in $U$.
        \begin{claim}
            There is a subset of inputs $\tilde{U} \subseteq U$ with size $|\tilde{U}| \geq |U|/2 \geq n/4$ such that the degree of the vertices in $\tilde{U}$ in the dependency graph of $f$ is at most $O(\localityf)$.
        \end{claim} 
        \begin{proof}
            Since there are at most $N \leq 2n$ output vertices, each of degree at most $\localityf$, there are at most $2n\localityf$ edges in the input/output dependency graph. Therefore, at least half of the vertices in $U$ have degree at most $4\localityf$ since otherwise there would be $|U|/2$ vertices with degree greater than $4\localityf$, and the total number of edges would be strictly greater than $\frac{|U|}{2} \cdot 4\localityf \geq \frac{n}{2} \cdot 4\localityf = 2\localityf n$ edges.
        \end{proof}
        
        Within these bounded degree input vertices $\tilde{U}$, we next find a subset $W$ such that each pair of vertices in $W$ are adjacent to disjoint trees.
        \begin{claim}\label{claim:size_of_independent_sources_tree}
            There exists a subset of inputs $W\subseteq \tilde{U}$ of size $|W| \geq \Omega(N/\localityf ^3)$ such that for each pair $u_i\neq u_j \in W$, the neighborhoods $\fneighb(u_i), \fneighb(u_j)$ intersect with disjoint trees. That is, for each $u_i\neq u_j\in W$, $\treeneighb(\fneighb(u_i)) \cap \treeneighb(\fneighb(u_j))  = \emptyset$.
        \end{claim}
        \begin{proof}
            We greedily build $W$ as follows: Initialize the set $V$ as the inputs $\tilde{U}$. While $V$ is non-empty, choose any $v\in V$, add it to $W$ and remove $\fneighb(\treeneighb(\fneighb(v)))$ from $V$. 
            
            Note that the size of $V$ decreases by at most $O(\localityf^3)$ in each iteration since for any subset of outputs $S$, $|\fneighb(S)| \leq \localityf |S|$, and $|\treeneighb(S)| \leq 2\localityf |S|$, and for any subset of inputs $S_{in}$, $|\fneighb(S_{in})| \leq O(\localityf)$.  Therefore, $|W| = |\tilde{U}|/O(\localityf^3) \geq  \Omega(n/\localityf^3) = \Omega(N/\localityf^3)$.
        \end{proof}
        
        We set $w$ as the input bits of $u$ which are indexed by $W$ from \Cref{claim:size_of_independent_sources_tree}, and let $y$ be the remaining bits of $u$. For each $i\in [s]$, let $\forest{i} = \treeneighb(\fneighb(w_i))$ and let $\forest{0}$ be the remaining $\{z_i\}$ variables:  $\forest{0} = \{z_i\}_{i\in [n]}\setminus (\bigcup_{i\in [s]} \forest{i})$. 
        
        We first show that $\forest{0}, \dots, \forest{s}$ is a \emph{forest partition} as defined in \Cref{defn:forest_partition}. By the definition of $\forest{0}$ it is clear that $\bigcup_{i=1}^s \forest{i} = \{z_i\}_{i\in [N]}$. Furthermore, these forests are disjoint since for each $i\neq j \in [s]$, $\forest{i} \cap \forest{j} = \treeneighb(\fneighb(w_i)) \cap \treeneighb(\fneighb(w_j)) = \emptyset$ by \Cref{claim:size_of_independent_sources_tree}, and since $\forest{0} \cap (\bigcup_{i\in [s]} \forest{i}) = \emptyset$ by definition. All that's left to show that this is a forest partition is that $\treeneighb(\forest{i}) = \forest{i}$ for each $i\in \{0, \dots, s\}$. This is clearly true for each $i\in [s]$ since $\treeneighb(\forest{i}) = \treeneighb(\treeneighb(\fneighb(w_i))) = \treeneighb(\fneighb(w_i)) = \forest{i}$. To show that $\treeneighb(\forest{0}) = \forest{0}$, suppose for the sake of contradiction that this is not the case, that there exists some $a \in \treeneighb(\forest{0})\setminus \forest{0}$. Since $\bigcup_{j=0}^s \forest{j} = \{z_i\}_{i\in [N]}$, $a$ is in some other forest $\forest{j}$ with $j\neq 0$. But this implies that $\treeneighb(\forest{j}) \cap \forest{0} \neq \emptyset$, and so $\forest{j} \cap \forest{0} \neq \emptyset$, a contradiction. Therefore, $\forest{0}, \forest{1}, \dots, \forest{s}$ is a forest partition as defined in \Cref{defn:forest_partition}.

        Next, we show that for each $i \in [s]$, $f(w, y) \big|_{\forest{i}}$ is a function of only $w_i$ and $y$. This is because for each $j\in[s]$, such that $j\neq i$, we have $\fneighb(w_j) \cap \forest{i} \subseteq \forest{j} \cap \forest{i} = \emptyset$. Similarly, the outputs $\forest{0}$ do not depend on any bits of $w$ since for each $i\in [s]$, $\fneighb(w_i) \cap \forest{0} \subseteq \forest{i} \cap \forest{0} = \emptyset$. 

        Since we initialized our set of fixed variables $L$ with $\fneighb(T_0 \cup \{b\})$, and we chose $W$ such that $W\cap L = \emptyset$, it follows that both $b$ and the outputs in $T_0$ can be written as functions of $y$. Furthermore, this implies that $T_0 \subseteq \forest{0}$. 

        All that's left to prove \Cref{lem:partition_inputs_tree} is to show $|F_i| \leq O(\localityf^2)$ for each $i\in[s]$. Note that for each $i\in [s]$, $|\forest{i}| = |\treeneighb(\fneighb(w_i))|$. Since $w_i$ was chosen from the subset of input variables that are not adjacent to $T_0$ in $f$'s dependency graph (those indexed by $U$), and have degree at most $O(\localityf)$ (indexed by $\tilde{U}\subseteq U$), it follows that $|\treeneighb(\fneighb(w_i))| \leq 2\localityf |\fneighb(w_i)|$ and $|\fneighb(w_i)| \leq O(\localityf)$. Therefore, $|\forest{i}| \leq O(\localityf^2)$ for each $i\in [s]$. 
        
    \end{proof}
    
    Next we consider how the $\pmmajmod{p}$ function evaluates on $(x, d)$.
    We partition the terms of the sum $S = \sum_{i=1}^{n-1} x_i (-1)^{\pathsum{d}_i}$ into $s+1$ parts according to the forest partition $\forest{0}, \forest{1}, \dots, \forest{s}$ from \Cref{lem:partition_inputs_tree}. 
    \begin{align}
        \blocksum{i} = \sum_{v_j \in V(F_i)} x_j (-1)^{\pathsum{d}_i} && \text{for each } i \in \{0,1, \dots, s\}.
    \end{align}
    Where $V(\forest{i})$ denotes the set of vertices $v_j \in V(\bintree[n])$ such that $x_j \in \forest{i}$ and $E(\forest{i})$ denotes the set of edges $e_j \in E(\bintree[n])$ such that $d_j \in \forest{i}$ for $i\in \{0,1, \dots, s\}$. Again, note that $v_0\notin V(F_0)$.
    We can rewrite the sum as $S = \sum_{i=0}^s \blocksum{i}$.

    Let's consider the sum $S$ for a particular assignment $z = (x,d) \in \{0,1\}^{N}$, where for each $i\in \{0,1,\dots,s\}$, $z\at{\forest{i}}$ denotes the assignment to $\forest{i}$. Note that $\blocksum{0}$ depends only on $z\at{\forest{0}}$, and each term $\blocksum{i}$ for $i\geq 1$ depends only on $z\at{\forest{0}}$ and $z\at{\forest{i}}$. 
    \begin{align}
        S(z) = \blocksum{0}(z\at{\forest{0}}) + \sum_{i=1}^s \blocksum{i}(z\at{\forest{i}}, z\at{\forest{0}})
    \end{align}
    This is because $x_j (-1)^{\pathsum{d}_i}$ depends on $x_j$ as well as each $d_{j'}$ where $e_{j'}$ is along the path from $v_0$ to $v_j$ in $\bintree$.

    \begin{defn}[Minimal Block]\label{defn:minimal_block}
        For some assignment $z\in \{0,1\}^{N}$, we say that the $i$th block is \emph{minimal} if 
        \begin{align}
            \blocksum{i}(z\at{\forest{i}}, z\at{\forest{0}})  = \min_{z'\at{\forest{i}}\in \{0,1\}^{|\forest{i}|}} \blocksum{i}(z'\at{\forest{i}}, z\at{\forest{0}}).
        \end{align}
    \end{defn}
    \begin{claim}\label{claim:unique_min_forest}
        For each fixed assignment to $z_{\forest{0}}$, and any $i\in [s]$, there is a unique minimal assignment to $z_{\forest{i}}$. That is, for each $z_{\forest{0}} \in \{0,1\}^{|\forest{0}|}$, there exists a $z^*_{\forest{i}} \in \{0,1\}^{|\forest{i}|}$ such that 
        \begin{align}
            \blocksum{i}(z^*_{\forest{i}}, z_{\forest{0}}) < \blocksum{i}(z_{\forest{i}}, z_{\forest{0}}) && \text{for each } \ z_{\forest{i}} \in \{0,1\}^{|\forest{i}|} \setminus \{z^*_{\forest{i}}\}.
        \end{align}
    \end{claim}
    \begin{proof}
        For each $i\in [s]$, the sum $\blocksum{i}$ can be broken into terms for each of the small trees $T_j \in \allsmalltrees$ in the forest $\forest{i}$.
        \begin{align}
            \blocksum{i} = \sum_{j \in [k] : T_j \subseteq \forest{i}} \treesum{T_j}
        \end{align}
        Where $\treesum{T_j} := \sum_{v_i \in V(T_j)} x_i (-1)^{\pathsum{d}_i}$.
        Note that the value each of $\treesum{T_j}$ for $j\in [s]$ depends on $z\at{\forest{0}}$ and the variables in $T_j$. Since each $T_j$ for $j\in [s]$ are disjoint, it is sufficient for us to show that for a fixed $z\at{\forest{0}}$, there is a unique minimal assignment to the variables of $T_j$ for each $j\in [s]$.

        For any two vertices $v_j \neq v_k \in V(\bintree[n])$, let $\path{j, k}\subseteq E(\bintree[n])$ be the subset of edges $\{e_1, \dots, e_{n-1}\}$ along the path from $v_j$ to $v_k$. Note that for any vertex $v_i$, $P(v_i)$ as defined in \Cref{defn:binary_tree} is equivalent to $\path{0,i}$. 
        Consider some $T\in \allsmalltrees$ with root $v_r$, and single-endpoint root edge $e_r$. We can rewrite $\treesum{T}$ as
        \begin{align}
            \treesum{T} &= \sum_{v_i \in V(T)} x_i \prod_{e_j \in \path{0,i} }  (-1)^{d_j}\\
            &= (-1)^{\pathsum{d}_r}\left( x_r + \sum_{v_i \in V(T)\setminus \{v_r\}} x_i \prod_{e_j \in \path{r,i} }  (-1)^{d_j}\right).
        \end{align}
        Note that $\pathsum{d}_r$ is a function of $z\at{\forest{0}}$ and $d_r$, and for a fixed $z\at{\forest{0}}$, we can fix $d_r$ such that $\pathsum{s}_r = -1$. Consider that we set $d_r$ in this way.
        \begin{align}
            \treesum{T}  &= - x_r + \sum_{v_i \in V(T)\setminus \{v_r\}} -x_i \prod_{e_j \in \path{r,i} }  (-1)^{d_j}
        \end{align}
        Now, $\treesum{T}$ is minimized if each of the $V(T)$ terms are minimized (value $-1$). This is achieved by setting $x_i = 1$ for each $v_i\in V(T)$ and $d_j = 0$ for each $e_j \in E(T) \setminus \{e_r\}$. Note that any other assignment to the variables will result in one of the terms being either $0$ or $1$, therefore this is the unique minimal assignment to the tree $T$. 
    \end{proof}

    Next, we design a statistical test similar to that in the proof of classical hardness of $(X, \majmod{p}\oplus \parity(X))$ (\Cref{thm:classical_LB_with_GHZ}) in \Cref{sec:classical_hardness_with_GHZ} with the additional set $T_M$ consisting of strings with a limited number of minimal blocks. We define the statistical test as follows.
    \paragraph{Statistical Test:} Let $N_0, N_M := 3N^{3\alpha}$ and $N_F := 2N^{3\alpha}$. The statistical test is $T := T_M \uplus T_0 \uplus T_F \uplus T_S$, where 
    \begin{align}
        T_M &:=  \{ z'\in \{0,1\}^{N+1} : \  \leq N_M \text{ blocks } i\in [s] \text{ of } z' \text{ are \emph{minimal}}\} \\
        T_0 &:=  \{ z'\in \{0,1\}^{N+1} : \  z'\at{\forest{i}} = 0^{|\forest{i}|} \text{ for } \leq N_0 \text{ blocks } i\in [s]\} \\
        T_F &:=  \{z' \in \{0,1\}^{N+1} : \exists (w,y) : f(w,y) = z' \text{ and } \geq N_F \text{ blocks } g_i(w_i, y) \text{ are } y\text{-fixed} \}\\
        T_S &:=  \{(z, b)\in \{0,1\}^{N}\times \{0,1\} : b \neq \pmmajmod{p}(z) \} \qquad \text{(``incorrect strings'')}
    \end{align}

    We will show that the function $f(U)$ passes the statistical test with probability at least $\frac{1}{2} - O(1/\log N)$ whereas the true distribution $D = (Z, \pmmajmod{p}(Z))$ passes with probability at most $1/N$ for sufficiently large $N$.

    \begin{claim} \label{claim:pmmajmod_statistical_test_pass}
        $\Pr[f(U) \in T] \geq \frac{1}{2} - O(1/\log N)$.
    \end{claim}
    \begin{proof}
        Using our partition of random input $u$ into $(w,y)$, our goal is to upper bound $\Pr_{w,y}[f(w,y) \in T]$, where the probability is taken over the randomness of $(w,y)$ chosen uniformly at random from $\{0,1\}^s \times \{0,1\}^{\ell -s}$. Since $\Pr_{w,y}[f(w,y) \in T] \geq \min_y \Pr_x[f(w,y) \in T]$, it is sufficient for us to upper bound $\Pr_w[f(w,y)\in T]$ for arbitrarily chosen $y\in \{0, 1\}^{\ell-s}$.

        Suppose we arbitrarily fix $y\in \{0,1\}^{\ell -s}$. If $\geq N_F$ blocks of $f(w,y)$ are $y$-fixed, then $f(w,y)\in T_F$ for each $w\in \{0,1\}^s$. Moreover, if at most $N_M$ blocks $g_i(w_i, y)$ are minimal for some choice of $w_i\in \{0,1\}$, then for each $w\in \{0,1\}^s$, $f(w, y) \in T_M$. Similarly, if at most $N_0$ blocks evaluate to zero $g_i(w_i, y) = 0^{|\forest{i}|}$ for some choice of $w_i \in \{0,1\}$, then for each $w\in \{0,1\}^s$, $f(w, y) \in T_0$. Therefore, we assume that less than $N_F$ blocks of $f$ are $y$-fixed, greater than $N_F$ of the forests of $f(w,y)$ take on their minimal value for some choice of $w$, and greater than $N_0$ blocks are all zeros for some choice of $w$. Therefore, the following two hold:
        \begin{enumerate}
            \item There are at least $N_M - N_F = N^{3\alpha}$ blocks $i\in [s]$ such that $\blocksum{i}(0, y) \neq \blocksum{i}(1, y)$.
            \item There are at least $N_0 - N_F = N^{3\alpha}$ blocks $i\in [s]$ such that $|g_i(0, y)| \neq |g_i(1, y)|$.
        \end{enumerate}
        Let $J\subseteq [s]$ be the indices of the blocks that change their respective terms of $S$, and let $K\subseteq [s]$ be the indices of the blocks with Hamming weight that change.
        \begin{align}
            J := \{i\in [s] : \blocksum{i}(0, y) \neq \blocksum{i}(1, y)\} && K := \{i\in [s] : |g_i(0, y)| \neq |g_i(1, y)|\}
        \end{align}
        We denote $|x, d|$ as the Hamming weight of the first $N$ output bits of $f(w,y)$, and recall that $b$ is the last bit of $f(w, y)$. Note that $|x,d|=|h(y)| + \sum_{i=1}^s |g_i(w_i, y)|$.  
        \begin{claim}\label{claim:S_hw_as_lincombs}
            Over the randomness of $w\in \{0,1\}^s$, the random variables $S$ and $|x, d|$ can be written as 
            \begin{align} \label{eq:S_hw_as_lincombs}
                S = a + \sum_{i\in J} a_i r_i, && |x, d| = e + \sum_{i \in K} e_i r_i && \text{where } r\sim \unif (\{0,1\}^{|J \cup K|}).
            \end{align}
            For some integers $a, e$, positive integers $a_1, \dots, a_{|J|} \leq O(\localityf^2) = O(\epsilon\log N)$, and nonzero integers $e_1, \dots, e_{|K|}$.
        \end{claim}
        \begin{proof}
            Note that over the randomness of $w\in \{0,1\}^s$, 
            for each $j'\notin J$ and $k'\notin K$, $\blocksum{j'}$ and $|g_{k'}(w_{k'}, y)|$ are fixed. Therefore, there exists some integers $\alpha, \beta$ such that 
            \begin{align}
                S = \alpha + \sum_{j\in J} \blocksum{j} && |x,d| = \beta + \sum_{k\in K} |g_k(w_k, y)|.
            \end{align}
            Moreover, each $\blocksum{j}$ for $j\in J$  are independent random variables which take on two different integer values with equal probability. Likewise, the $|g_k(w_k, y)|$ for $k\in K$ are independent random variables which take on two distinct values with equal probability. Although for $i\in J\cap K$, $\blocksum{i}$ and $|g_i(w_i, y)|$ are not independent. Thus for each $j\in J$ and $k\in K$, there exists integers $\alpha_0, \alpha_1, \beta_0, \beta_1$ such that $\alpha_0 \neq \alpha_1$,  $\beta_0\neq \beta_1$, and
            \begin{align}
                \blocksum{j} = \begin{cases}
                    \alpha_0 & \text{if } w_j = 0\\
                    \alpha_1 & \text{if }w_j = 1
                \end{cases}
                && |g_k(w_k, y)| = \begin{cases}
                    \beta_0 & \text{if } w_j = 0\\
                    \beta_1 & \text{if }w_j = 1
                \end{cases}
                && w \sim \unif(\{0,1\})^{|J \cup K|}.
            \end{align}
            For each $i\in J \cup K$, we will assign $r_i$ to either $w_i$ or $1-w_i$. Since each $w_i$ is independently uniformly random over $\{0,1\}$, so is each $r_i$. 

            Note that we can write the term $\blocksum{j}$ as  either $\blocksum{j} = \alpha_0 + (\alpha_1 - \alpha_0) w_j$, or $\blocksum{j} = \alpha_1 + (\alpha_0 - \alpha_1) (1- w_j)$. Thus, it is possible to rewrite $\blocksum{j}$ as $c + a_j r_j$ for some integer $c$ and positive integer $a_j$, by setting $r_j = w_j$ and $a_j = (\alpha_1 - \alpha_0)$ if $\alpha_1 > \alpha_0$ and setting $r_j = 1 - w_j$ and $a_j = (\alpha_0 - \alpha_1)$ if $\alpha_0 > \alpha_1$. Furthermore, the value of $a_j$ is  $|\alpha_0 - \alpha_1|$, and is at most $2 \cdot |V(F_j)| \leq \localityf \cdot 2^D = 2\localityf^2$ since the value of $|\blocksum{j}|$ is at most the number of vertices in $\forest{j}$. 
            Therefore, we can write $S = a + \sum_{i\in J} a_ir_i$ for some integer $a$ and positive integers $a_i \leq 2\localityf^2$ for $i\in J$. 

            For each $k\in K$, we can also write the term $|g_k(w_k, y)|$ as either $\beta_0 + (\beta_1 - \beta_0)x_k$ or $\beta_1 + (\beta_0 - \beta_1)(1- x_k)$. Therefore, regardless of whether $r_k$ was assigned as $x_k$ or $1-x_k$, the term can be written as $c + e_kr_k$ for some (not necessarily positive) integers $c$ and $e_k$.
            And, as desired, the entire Hamming weight sum can be written as $|x,d| = b + \sum_{i \in K} e_i r_i$ for some integers $b$ and $e_i$ for $i\in K$.
        \end{proof}
        Next, we consider how much the sums in \Cref{eq:S_hw_as_lincombs} depend on the same bits of $r$.
        Suppose that $|J \cap K| \leq \frac{1}{2} N^{3\alpha}$. Then $|J\setminus K| \geq \frac{1}{2}N^{3\alpha}$. If we fix $r_K$ arbitrarily, the value of $|x, d|$ is fixed, and therefore so is $\parity(x, d)$. Letting $c = \parity(x,d)$,  $a' = a + \sum_{i\in J\cap K} a_i r_i$, and $J' = J\setminus K$, we can simplify the probability that the output is ``incorrect'' over the randomness of $r_{J'}$ as follows.
        \begin{align}
            \Pr_{r_{J'}}\left[f(w, y) \in T_S\right] &= \Pr_{r_{J'}}[\MM{p}(S)\oplus \parity(x, d) \neq b]\\
            &= \Pr_{r_{J'}}\left[\MM{p}\left(a' + \sum_{i\in J'} a_i r_i\right) \neq c \oplus b\right]\\
            &= \Pr_{r_{J'}}\left[a' + \sum_{i\in J'} a_ir_i \in M_{c\oplus b \oplus 1} \mod p \right]
        \end{align}
        Where $M_0 = \{0, 1, \dots, (p-1)/2\}$ and $M_1 = \{(p+1)/2, \dots, p-1\}$. Since $|M_0|, |M_1| \geq (p-1)/2$, and $a_i$ is nonzero modulo $p$ (since $a_i \leq O(\epsilon\log N)$ for $i\in J$, and $p = \Theta(N^\alpha)$)) it follows from \Cref{cor:prob_sum_in_set} that 
        \begin{align}
            \Pr_{r_{J'}}\left[f(w, y) \in T_S\right] \geq \frac{p-1}{2p} - O(1/N) \geq 1/2 - O(1/p).
        \end{align}
        Where we used that $|J'| \geq \frac{1}{2}N^{3\alpha} \geq \Omega(p^3)$. Since the bits of $r_K$ were fixed arbitrarily, it holds that $\Pr_w[f(w, y) \in T_S] = \Pr_r[\MM{p}(S) \oplus \parity(x, d) \neq b] \geq 1/2 - O(1/p)$. Therefore we assume that $| J \cap K| > \frac{1}{2}N^{3\alpha}$. 
        
        If we fix all $r_i$ for $i\notin J\cap K$, the remaining non-fixed blocks $i \in J \cap K$ have possible assignments which give different values to both $|g_i(w_i, y)|$ and $S_i$. Letting $a' = a + \sum_{i \notin J\cap K} a + a_ir_i$, and $e' = \sum_{i \notin J \cap K} e_ir_i$, we simplify the probability that $f(w,y)$ is ``incorrect'' over the randomness of $r_{J\cap K}$ as follows.
        \begin{align}
            \Pr_{r_{J\cap K}}\left[f(w, y) \in T_S\right] = \Pr_{r_{J\cap K}}\left[\MM{p}\left(a' + \sum_{i\in J\cap K} a_i r_i\right)\oplus \intparity\left(e' + \sum_{i\in J \cap K}e_ir_i\right)\right]
        \end{align}
        Since $a_i\leq O(\localityf^2) \leq O(\epsilon\log N)$ for each $i\in [s]$ (by \Cref{claim:S_hw_as_lincombs}) and $|J \cap K | \geq \frac{1}{2}N^{3\alpha} = \Omega(p^3)$, it directly follows from \Cref{lem:MMP_of_sum} that
        \begin{align}
            \Pr_{r_{J\cap K}}\left[f(w, y) \in T_S\right]  \geq \frac{1}{2} - O(1/\log N)
        \end{align}
        Therefore, $\Pr_w[f(w, y) \in T_S] \geq \frac{1}{2} - O(1/\log N)$.
    \end{proof}
    \begin{claim} \label{claim:pmmmp_fail_stat_test}
        $\Pr[(Z, \pmmajmod{p}(Z)) \in T ] \leq 1/N$ for sufficiently large $N$.
    \end{claim}
    \begin{proof}
        This proof is almost identical to that of \Cref{claim:mmp_fail_stat_test}, which follows closely to the proof of Claim 3.3 in \cite{viola2012complexity}. The main difference in this proof accounts for the additional term $T_M$ in the statistical test -- so in addition to upper bounding the probability that $\mathcal{D} = (Z, \pmmajmod{p}(Z))$ is in  $T_0, T_S$, or $T_F$, we will also upper bound the probability that $\mathcal{D}\in T_M$.  Since $\mathcal{D}$ always outputs a ``correct'' string, $\Pr[\mathcal{D} \in T_S] = 0$. Thus, by the union bound it is sufficient for us to prove that $\Pr[\mathcal{D} \in T_0], \Pr[\mathcal{D} \in T_F], \Pr[\mathcal{D} \in T_M] \leq \frac{1}{3N}$.

        We start by showing that $\Pr[\mathcal{D} \in T_M] \leq \frac{1}{3N}$. To this end, we consider the probability that $\mathcal{D}\in T_M$ conditioned on the value of $Z\at{\forest{0}}$. Since $Z\at{\forest{0}} \in \{0,1\}^{|\forest{0}|}$ is uniformly random, 
        \begin{align}
        \Pr[\mathcal{D} \in T_M] = \frac{1}{2^{|\forest{0}|}} \sum_{t_0\in \{0,1\}^{|\forest{0}|}}  \Pr[\mathcal{D} \in T_M |Z\at{\forest{0}} = t_0].
        \end{align}
        Thus it is sufficient for us to show that $\Pr[\mathcal{D} \in T_M | Z\at{\forest{0}} = t_0] \leq \frac{1}{3N}$ for each $t_0 \in \{0,1\}^{|\forest{0}|}$.

        As shown in \Cref{claim:unique_min_forest}, for each forest $\forest{i}$ for $i\in [s]$, and some fixed $z\at{\forest{0}} \in \{0,1\}^{|\forest{0}|}$, there is a unique assignment for $z\at{\forest{i}}$ to minimize $\blocksum{i}(z\at{\forest{i}}, z\at{\forest{0}})$. Additionally, the minimality of each block is independent conditioned on the value of $Z\at{\forest{0}}$ since for each $i\in [s]$, $\blocksum{i}(Z)$ is a function of only $Z\at{\forest{i}}$ and $Z\at{\forest{0}}$.

        We lower bound the probability that any given forest is minimal conditioned on the value of $Z\at{\forest{0}}$. For any $i\in [s]$ and $t_0 \in \{0,1\}^{|\forest{0}|}$, we have
        \begin{align}
            \Pr_\mathcal{D}[\text{block } i \text{ is minimal } | Z\at{\forest{0}} = t_0] = \frac{1}{2^{|\forest{i}|}} \geq 2^{-O(\localityf^2)} = 2^{-O(\epsilon \log N)} \geq N^{-O(\epsilon)}. \label{eq:prob_block_minimal}
        \end{align}
        Where we used that $|\forest{i}| \leq   O(\localityf^2) \leq O(\epsilon\log n)$ for $i \in [s]$.

        Since the minimality of each forest are independent conditioned on the value of $Z\at{\forest{0}}$, for any subset of forests $U\subseteq [s]$, the probability that none of them are minimal conditioned on $Z\at{\forest{0}}$ is
        \begin{align}
            \Pr_\mathcal{D}[\text{all forests of } U \text{ are \emph{not} minimal} | Z\at{\forest{0}} = t_0] = \prod_{i\in U} \Pr[\text{forest } i \text{ is not minimal} | Z\at{\forest{0}} = t_0]. \label{eq:forest_conditional_independence}
        \end{align}
        Therefore, for each $t_0\in \{0,1\}^{|F_0|}$,
        \begin{align} 
            \Pr_\mathcal{D}[\mathcal{D} \in T_M | Z\at{\forest{0}} = t_0] &= \Pr_\mathcal{D} \left[\bigcup_{\substack{U\subseteq [s]: \\ |U| = s-N_M}} \left\{ \text{all forests of } U \text{ are \emph{not} minimal } \right\} \Bigg| Z\at{\forest{0}} = t_0 \right]\\
            &\leq \sum_{\substack{U\subseteq [s]: \\ |U| = s-N_M}} \Pr\left[ \text{all forests of } U \text{ are \emph{not} minimal } \Big\vert  Z\at{\forest{0}} = t_0  \right]\\
            &= \sum_{\substack{U\subseteq [s]: \\ |U| = s-N_M}} \prod_{i\in U} \Pr\left[\text{forest } i \text{ is not minimal} \Big| Z\at{\forest{0}} = t_0\right]\\
            &\leq \binom{s}{N_M} \left(1- N^{-O(\epsilon)}\right)^{s- N_M} \label{eq:cond_prob_DinT_M_start}
        \end{align}
        In the second line we used the union bound, the third line we used the independence of the block's minimality conditioned on $Z_{\forest{0}}$ (\Cref{eq:forest_conditional_independence}), the fourth line we used \Cref{eq:prob_block_minimal}.
        We can further simplify, beginning with standard bounds and then using that $\Omega(N/\localityf^3) \leq s \leq N$, $\localityf \leq (\epsilon \log N)^{1/2}$, and $N_M = 3N^{3\alpha}$
        \begin{align}
            &= \binom{s}{N_M} \left(\left(1- N^{-O(\epsilon)}\right)^{N^{O(\epsilon)}} \right)^{ N^{-O(\epsilon)}(s- N_M) } \\
            &\leq \left(\frac{s}{N_M}\right)^{N_M} \exp(-N^{-O(\epsilon)}(s - N_M))\\
            &= s^{N_M} \exp(-N^{-O(\epsilon)} s) \left(\frac{\exp(N^{-O(\epsilon)})}{N_M}\right)^{N_M}\\
            &\leq N^{3N^{3\alpha}} \exp\left(\frac{n^{1 - O(\epsilon)}}{ \log^{3/2}N}\right)\\
            &\leq \frac{1}{3N}  \label{eq:cond_prob_DinT_M_end}
        \end{align}
        for sufficiently large $N$ and small $\epsilon$ (such that $3\alpha < 1 - O(\epsilon)$). Therefore $\Pr[\mathcal{D} \in T_M]\leq \frac{1}{3N}$. 

        Next, we show using similar calculations that $\Pr[\mathcal{D} \in T_0] \leq \frac{1}{3N}$. Since each of the blocks $i\in [s]$, $Z\at{\forest{i}}$ is uniformly random, whether each of them is all zeros is independent. Therefore the probability that block $i\in [s]$ is all zeros is.
        \begin{align}
            \Pr[Z_{\forest{i}} = 0^{|\forest{i}|}] = 2^{-|\forest{i}|}  \leq 2^{-O(\localityf^2)} = N^{-O(\epsilon)}  && \text{for each } i \in [s]
        \end{align}
        Since $N_0 = 3N^{3\alpha}$, we can use the calculations from \Crefrange{eq:cond_prob_DinT_M_start}{eq:cond_prob_DinT_M_end} to bound $\Pr[\mathcal{D} \in T_0]$.
        \begin{align}
            \Pr[\mathcal{D} \in T_0 ] &\leq \sum_{\substack{T\subseteq [s]: \\ |T| = s-N_M}} \prod_{i\in T} \Pr\left[ Z\at{\forest{i}} \neq 0^{|\forest{i}|}\right]\\
            &\leq \binom{s}{N_0} \left(1 - N^{-O(\epsilon)}\right)^{s - N_0}\\
            &\leq \frac{1}{3N}
        \end{align}
        For sufficiently large $N$ and small $\epsilon$.

        All that's left is to show $\Pr[\mathcal{D} \in T_F] \leq \frac{1}{3N}$. For this we use the same calculations from the proof of \Cref{claim:mmp_fail_stat_test}, but in this scenario we have $\ell \leq N + N^{3\alpha}$, and the size of the support of $\mathcal{D}$ is $2^N$.
        \begin{align}
            \Pr[\mathcal{D} \in T_F] \leq \frac{|T_F|}{2^{N}} \leq \frac{2^{\ell - N_F}}{2^{N}} \leq 2^{N^{3\alpha} - 2N^{3\alpha}} \leq 2^{-N^{3\alpha}} \leq \frac{1}{3N}.
        \end{align}
        Where we used $\ell \leq N + N^{\delta}$, $\delta \geq 3\alpha$, and $N_F = 2N^{3\alpha}$.
        Therefore, applying the union bound we get
        \begin{align}
            \Pr[\mathcal{D} \in T] &\leq \Pr[\mathcal{D} \in T_S] + \Pr[\mathcal{D} \in T_M] + \Pr[\mathcal{D} \in T_0] + \Pr[\mathcal{D} \in T_F]\\
            &\leq 0 + \frac{1}{3N} + \frac{1}{3N} + \frac{1}{3N} = \frac{1}{N}
        \end{align}
    \end{proof}
\end{proof}

\appendix

\section{Implementing the \texorpdfstring{$\Urot{\theta}{m}$}{U m theta } Unitary. } 
\label{sec:U_m_theta_compiling}

The quantum circuits constructed in this paper involved $m$-qubit unitary gates, which we denoted by $\Urot{m}{\theta}$. These gates were chosen to implement unitary operations close to some desired non-unitary operation $\multiNUrot{m}{\theta}$. In the body of the paper we showed existence of these unitaries, but avoided a discussion of how to construct these gates out of a more elementary gate set. In this appendix we briefly outline one approach to answering this question, when the elementary gate set chosen contains arbitrary one qubit gates along with CNOT gates. While we do not give an explicit compilation of the $\Urot{m}{\theta}$ unitary in terms of these elementary gates, we outline the steps that can be used to find such a construction. Implementing this algorithm (or finding some other ad-hoc compilation of the $\Urot{m}{\theta}$ unitary) would be a necessary step before implementing the circuits described in this paper on a near-term quantum computer. Additionally, existence of this algorithm implies that the quantum circuits considered in this paper are uniform (meaning a description of them can be found in polynomial time) provided arbitrary one qubit gates and two qubit CNOT gates are allowed as elementary gates in the quantum circuit.  

The first thing to note is that an explicit description of the $\Urot{m}{\theta}$ can be obtained by starting with \Cref{lem:Frobenius_Urot_NUrot_bound} and then partitioning the set $\{0,1\}^m$ as described in the lemma. For completeness, one such possible definition is given below. 

\begin{defn}
    For any $m, \theta$ define $\Urot{m}{\theta}$ to be the unitary that acts on computational basis states $\ket{x} = \ket{x_1x_2...x_m}$ as
    \begin{align}
        \Urot{m}{\theta}\ket{x} &= \begin{cases}
        \multiNUrot{m}{\theta}\ket{x} & \text{ if } x_1 = 0  \\
        \sqrt{1- \sin^{2m}}\left(\multiNUrot{m}{\theta} \ket{x} + i^{m + 2 \abs{x}}\sin^{m}(\theta)\multiNUrot{m}{\theta}\ket{\overline{x}}\right) & \text{ if } x_1 = 1 
        \end{cases}
    \end{align}
\end{defn}
\noindent It is immediate from \Cref{lem:Frobenius_Urot_NUrot_bound} that $\Urot{m}{\theta}$ as defined above is unitary and satisfies $\norm{\Urot{m}{\theta} - \multiNUrot{m}{\theta}}_F \in O(\theta^{-m})$.

Now it remains to show how the $\Urot{m}{\theta}$ unitary described above can be compiled in terms of arbitrary one qubit gates and two qubit CNOT gates. This argument follows from a chain of results.\footnote{ Which the authors are very grateful to Michael Oliveira for pointing out to us. } First, as noted in page 12 of \cite{green2001counting}, the arguments of \cite{reck1994experimental} give a algorithm whose runtime is bounded as a function of $m$ that compiles any $m$ qubit unitary into a sequence of at most $O(m^3 4^m)$ two qubit gates. Taking $m$ to be a constant this gives an algorithm which has constant runtime and allows any $m$-qubit gate to be rewritten as a sequence of a constant number of two-qubit gates. Additionally, Section 5.1 of \cite{barenco1995elementary} shows how any two qubit gate can be rewritten as a length 5 sequence of one qubit gates and CNOT gates. Putting these results together we see that for constant $m$ and any $\theta$ there is a constant time algorithm that rewrites the $\Urot{m}{\theta}$ unitary as a constant length sequence of one qubit gates and CNOT gates. 

We close this section with two observations. Firstly, we point out that the compilation procedure described above requires arbitrary one qubit rotations. This is necessary -- as a straightforward counting argument shows that (even for constant $m$) it is impossible to compile all $\Urot{m}{\theta}$ unitaries exactly in constant depth with a finite sized set of elementary gates. Indeed, finding a quantum-classical sampling separation where the quantum circuit has constant depth and only involves gates drawn from a constant size gate set is one important open question left by this work. Secondly, we point out the procedure described in the previous paragraph for compiling the $\Urot{m}{\theta}$ unitary is unlikely to produce an ``optimal'' compilation. With careful thought it may be possible to find a more natural compilation technique that produces $\Urot{m}{\theta}$ gates while requiring many fewer elementary gates. Finding such a compilation would likely make an experimental implementation of the circuits described in this paper much more feasible. 

\section{Lower Bounds Against Classical Circuits with Unlimited Inputs but Bounded Fan-out and Fan-in}
\label{apx:unlimited-input}
\newcommand{\Inf}{\mathrm{Inf}}
One important limitation of our main result is the restriction on the number of input bits to the classical circuit. While the role that the number of inputs plays in the complexity of distributions is puzzling, we make some progress on this front by considering a tradeoff between bounded fan-out and the number of inputs.

So far, this document only considers classical circuits that have bounded fan-in, but unbounded fan-out. In contrast, the quantum circuits considered have both fan-in and fan-out bounded. So it is reasonable to consider how the two compare when the classical circuits are restricted to have bounded fan-out.  In this section, we consider classical circuits with an unlimited number of inputs, but that have bounded fan-out in addition to fan-in. With this exchange of constraints, we are able to maintain our circuit lower bound against $(X, \majmod{p}(X))$--- The distribution that can be approximated by a constant-depth quantum circuit with a $GHZ$ advice state. The general structure of the proof in this section gains inspiration from a mixture of \cite{viola2012complexity,viola2014extractors} -- although the details specific to this distribution call for novel techniques. 
\begin{defn}
    A function $f :\zo^\ell \to \zon$ is \emph{$\localityf$-local in input and output} if each output bit is a function of at most $\localityf$ of the input bits, and each input bit influences at most $\localityf$ of the output bits.
\end{defn}
In other words, $f$ is $\localityf$-local in input and output if the input-output dependency graph of $f$ has bounded degree $\localityf$ on both the input and output nodes.
\begin{thm}\label{thm:mm-lb-bounded-fanout}
    Suppose $f : \zo^\ell \to \zon$ is $\localityf$-local in both input and output. Let $U \sim \unif(\zo^\ell)$, and $X \sim \unif(\zon)$. Then for each $c\in (0,1/3)$ and prime  $p \leq \Theta(n^{c})$,  
    there exists an $\epsilon \in (0,1)$ such that if $\localityf = \epsilon\log^{1/2} n$ then
    \begin{align}
        \TVD\pbra{f(U), \pbra{X, \majmod{p}(X)}} \geq \frac{1}{2} - O(1/\log n).
    \end{align}
\end{thm}

\begin{proof}
    We denote the inputs to $f$ as variables $u_i\in \zo$ for each $i\in [\ell]$, the first $n$ output variables as $z_j \in \zo$ for each $j\in [n]$, and let $b \in \zo$ be the final output bit (the one that is supposed to be $\majmod{p}(z)$. Throughout, we will refer to $\{z_i\}\cup \{x_i\} \cup \{b\}$ as both variable and nodes of the input-output dependency graph for $f$. We use $N^k(V)$ to denote the set of nodes within distance $k$ from a node in $V$. We will refer to $N^k(V)$ as the $k$-neighborhood of $V$, and denote $N(\cdot) = N^1(\cdot)$. 

    For each function $g: \zo^m \to \RR$ and input index $i\in [m]$, we define the \emph{influence} of the $i$th input variable as $\Inf_i(g) := \Pr_{x\in \zo^m}\sbra{g(x) \neq g(x^i)}$, where we use $x^i$ to denote the bitstring $x$ with the $i$th bit flipped ($x^i = x_1 \dots x_{i-1} (1 - x_i) x_{i+1} \dots x_m$) 

    We begin by proving the \nameCref{thm:mm-lb-bounded-fanout} for a very special case.
    \begin{claim}\label{claim:subset-hw-unfixed}
        Suppose that there exists a subset of the first $n$ output variables $S\subseteq [n]$ of size $|S|\geq 3$ such that $\sum_{j\in S} z_j$ is constant (does not depend on inputs $u_i$)). Then for any Boolean function $b' : \zon \to \zo$, $\TVD\pbra{f(U), \pbra{X, b'(X)}} \geq \frac{5}{8}$
    \end{claim}
    \begin{proof}
        Without loss of generality, suppose $|S| = 3$ (if not, remove elements until it is). Let $c = \sum_{j\in S} z_i$, and let $T\subseteq \zon$ be the subset of strings such that the hamming weight on $S$ is consistent with $z$: $T := \cbra{z'\in \zon : \sum_{j \in S} z'_i = c}$. By assumption, $\Pr[z \in T] = 1$. On the other hand, for $X\sim \unif(\zon)$, 
        \[
        \Pr[X \in T] = \frac{|T|}{2^n} = \frac{2^{n-|S|} \cdot\binom{|S|}{c}}{2^n} = 2^{-|S|}\binom{|S|}{c} = \frac{1}{8} \cdot \binom{3}{c} \leq \frac{3}{8}.
        \] 
        The last inequality used the fact that for any $c\in \{0, 1, 2, 3\}$, $\binom{3}{c} \leq 3$. Using statistical test with set $T \times \zo$ to bound the total variation distance between $f(U)$ and $(X, b'(X))$ we see that
        \begin{align}
            \TVD(f(U), (X, b'(X))\geq \big\vert\Pr[z \in T] - \Pr[X \in T]\big\vert =  1 - \frac{3}{8} = \frac{5}{8}.
        \end{align}
    \end{proof}
    Therefore, we assume that the sum of any $3$ output variables $\{z_i\}$ is not fixed. In this case, we can show the following.
    \begin{claim}\label{claim:high_inf_set}
        If all sums of $k$ or more of the variables $\{z_i\}$ are not constant, then there exists a subset of inputs $W\subseteq \{u_i\}_{i\in [\ell]}$ of size $|W| \geq \Omega\pbra{(n-k)/\localityf^5}$ such that both of the following are true:
        \begin{enumerate}
            \item \label{item:apx-claim:high-influence} (High influence)\ $\Inf_i[f] \geq \frac{1}{2^{\localityf^2}}$ for each $u_i\in W$
            \item \label{item:apx-claim:indep} (Independent)  $N^2(u_i) \cap N^2(u_j) = \emptyset$ for each $u_i\neq u_j \in W$.
            \item \label{item:apx-claim:no-aff-b}(No affect on $b$) $b\notin N(W)$
        \end{enumerate}
    \end{claim}
    \begin{proof}
        We will construct $W \subseteq \{u_i\}_{i\in [\ell]}$, the set of high-influence inputs, iteratively. We initialize sets $W = \emptyset$ and $F = [n]\setminus N^2(b)$. The set $W$ corresponds to our ``chosen inputs'' and the set $F$ our ``free outputs.'' We will iteratively update these sets, maintaining the following invariant:
        \begin{align}
            N^2(W) \cap N^3\pbra{F} = \emptyset \label{eq:mm-lb-fanout-invar}.
        \end{align} 
        Where $W^\complement = [\ell]\setminus W$. For string $s\in \zo^m$, we denote the Hamming weight as $|s| = \sum_{i=1}^m s_i$, and for each $S\subseteq[m]$, the Hamming weight over $S$ as $|s|_S = \sum_{i\in S} s_i$. We further extend this notation so that $|f| : \zo^\ell \to [n]$ is the function mapping $u\in \zo^\ell$ to  $|f(u)|$, and  for $S\subseteq [n]$, $|f|_S : \zo^\ell \to [|S|]$ is the function mapping $u\in \zo^\ell$ to  $|f(u)|_S$
        
        Our algorithm proceeds as follows:
        
        \begin{center}
        \fbox{%
            \parbox{0.8\textwidth}{%
                While $|F| \geq k$:
                \begin{itemize}
                    \item Since $|F| \geq k$, $|f(u)|_F = \sum_{z_j\in F} z_j$ is not fixed. Therefore, there exists some input variable $u_i$ with $\Inf_{i}\pbra{|f |_F} >0$. Our invariant in \Cref{eq:mm-lb-fanout-invar} guarantees that $u_i$ is not already in $W$ and that $N(u_i) \subseteq F$. Therefore  $\Inf_i(|f|) >0$.
                    \item Update $W \gets W \cup \{x_j\}$ and $F \gets F \setminus N^5(u_i)$, ensuring that \Cref{eq:mm-lb-fanout-invar} is still satisfied.
                \end{itemize}
                Return $W$.
            }%
        }
        \end{center}
        
        We now analyze the algorithm. First note that the final $W$ has size $|W| \geq \frac{n-\localityf^2 -k}{\localityf^5} = \Omega(\frac{n-k}{\localityf^5})$ since it grows by one in each iteration, we start with $|F| \geq n - \localityf^2$, end when $|F|<k$ and decrease $F$ by at most $\localityf^5$ in each iteration.

        The algorithm guarantees that $\Inf_i(|f|) >0$ for each $u_i \in W$. Using the fact that $f$ is $\localityf$-local we will see that the influence of any input $u_i \in W$ on $|f|$ is actually at least $\frac{1}{2^{\localityf^2}}$. To this end, we observe that for input variable $u_i$ and output variable $z_j$, if $u_i \notin N(z_j)$, then $f(u)_j$ is independent of $u_i$, so $|f(u)| - |f(u^i)| = 0$. Therefore
        \begin{align}
            |f(u)| - |f(u^i)| = \sum_{j=1}^n f(u)_j - f(u^i)_j = \sum_{j\in N(u_i)} f(u)_j - f(u^i)_j
        \end{align}
        is a function of a subset of the variables in $N^2(u_i)$, of which there are at most $\localityf^2$. We denote these variables as $\wt{u}$. Therefore, we see that the influence of each variable is a factor of $\frac{1}{2^{\localityf^2}}$
        \begin{align}
            \Inf_i(|f|) = \Pr_{\wt{u}}\sbra{|f(u)| - |f(u^i)| \neq 0} \geq 2^{-\localityf^2} K && \text{for some } K\in \ZZ^+
        \end{align}
        Since our algorithm guarantees that $\Inf_i(|f|) > 0$ for each $u_i \in W$, we have that actually $\Inf_i(|f|) \geq \frac{1}{2^{\localityf^2}}$ for each $u_i \in W$. This proves \Cref{item:apx-claim:high-influence} in the \nameCref{item:apx-claim:high-influence}.
        
        Next, we prove \Cref{item:apx-claim:indep}: Assume for the sake of contradiction that there exist $u_i\neq u_j \in W$ such that $N^2(u_i) \cap N^2(u_j) \neq \emptyset$. Suppose without loss of generality that $u_i$ was added to $W$ first. Now consider the iteration that $u_j$ is added to $W$, so $u_i \in W$. Since $u_j$ was the chosen input this round, it must be that $u_j \in N(F)$. But using our invariant $N^2(W) \cap N^3(F) = \emptyset$ we reach the following contradiction.
        \begin{align}
            \emptyset \neq N^2(u_i) \cap N^2(u_j) \subseteq N^2(W) \cap N^3(F) = \emptyset.
        \end{align}
        Therefore, it must be the case that \Cref{item:apx-claim:indep} is satisfied.

        Finally, we observe that \Cref{item:apx-claim:no-aff-b} is satisfied since we initialize $F = [n] \setminus N^2(b)$ and never add variables to $F$. Specifically, note that we build $W$ by looking at input variables that influence the variables in $F$ which never contains any variables in $N^2(b)$, so $F$ is not influenced by $N(b)$ and so the final $W$ will not intersect $N(b)$.
    \end{proof}
    Combining \Cref{claim:subset-hw-unfixed} and \Cref{claim:high_inf_set} with $k=3$, we have a set $W$ of input variables of size $s:= |W| \geq \Omega(n/\localityf^5)$ satisfying \Cref{item:apx-claim:high-influence,item:apx-claim:indep,item:apx-claim:no-aff-b}. We partition the input variables $u = (x, y)$ Where $x\in \zo^s$ are the variables in $W$ and $y\in \zo^{n-s}$ are the variables in $W^\complement$. As done throughout this paper, \Cref{item:apx-claim:indep} and permutting the outputs of $f$, allows us to express $f(u) = f(x,y)$ in blocks $g_i(x_i, y)$ for each $i\in [s]$, $h(y)$ 
    \begin{align}
        f(x,y) = g_1(x_1, y) \circ g_2(x_2, y) \circ \dots \circ g_s(x_1, y) \circ h(y).
    \end{align}
    Let $\mu$ be the average influence of $x_i \in W$. So $\mu \geq \frac{s}{2^{\localityf^2}}$. For each $y\in \zo^{\ell - s}$, we refer to $y$ as \emph{bad} if $|g_i|$ is $y$-fixed for greater than $\frac{\mu}{2}$ of the $i\in [s]$. Otherwise, $y$ is good. Suppose we sample and fix the input variables $y\sim\unif{\zo^{\ell - s}}$. We proceed in two steps. First, we show that $y$ is good with high probability; secondly, we show that conditioned on $y$ being good, then with high probability the bit $b$ output by the circuit is wrong ($b \neq \majmod{p}(z)$) with high probability.  
    \begin{claim}
        $\Pr_y[y \text{ is good}] \geq 1 - O(1/\log n)$
    \end{claim}
    \begin{proof}
        
    \end{proof}
    
    \Cref{item:apx-claim:indep} guarantees that each of the blocks $g_i$ are independent -- in addition to depending on disjoint variables of $x$ they also depend on disjoint variables of $y$.
    \Cref{item:apx-claim:no-aff-b} ensures that the final output variable $b$ is in the $h(y)$ block, and \Cref{item:apx-claim:high-influence} ensures that if we choose a $y$ uniformly at random, then 
    \begin{align}
        \Pr_y[|g_i(x_i,y)| \text{ is not fixed}] = \Inf_{x_i}(|g_i|) \geq \frac{1}{2^{\localityf^2}}.
    \end{align} 
    Therefore, over the random choice of $y$, the expected number of blocks with unfixed Hamming weight is $\mu \geq \frac{s}{2^{\localityf^2}}= \geq \frac{n}{\localityf^5 2^{\localityf^2}}$.
    Since each of the $g_i$ blocks are independent, it follows the Chernoff bound that
    \begin{align}
        \Pr_y[y \text{ is bad}] \leq \exp(\frac{-\mu}{8}) \leq 2^{-\Omega\pbra{\frac{n}{\localityf^5 2^{\localityf^2}}}} \leq O(1/\log n).
    \end{align}
    \begin{claim}
        $\Pr[b \neq \majmod{p}(z) \oplus \parity(z)\  \vert \ y \text{ good}] \geq \frac{1}{2} - O(1/\log n)$
    \end{claim}
    \begin{proof}
        Conditioned on the event that $y$ is good, at least $t = \frac{\mu}{2} = \frac{n}{\localityf^5 2^{\localityf^2+1}} = n^{1- O(\epsilon^2)}$ blocks $g_i$ have unfixed hamming weight. That is, over the randomness of $x_i$, $|g_i(x_i, y)|$ is a random variable that takes on different integer values for $t$ of the $i\in [s]$. Since the size of each block is at most $\localityf$, these integers are between $0$ and $\localityf = \epsilon \log n$. Therefore, we can write the total hamming weight of the first $n$ outputs as  $|z| = a_0 + \sum_{i=1}^t a_i x_i$ for $a_1, \dots, a_t$ positive integers that are at most $\localityf \leq p/\log n$. If we set $\epsilon$ to be sufficiently small so that $t = n^{1 - O(\epsilon^2)}$ is on the order of $\Omega(p^3) = \Omega(n^{3c})$ (recall that $c\in (0,1)$, then we can apply \Cref{{lem:MMP_of_sum}}, which gives us that $\Pr[\MM{p}(|z|) \oplus \parity(|z|) \neq b] \geq \frac{1}{2} - O(1/\log n)$, completing the proof.
    \end{proof}
    Finally, we combine the last two Claims to conclude that there exists an $\epsilon\in (0,1)$ such that setting $\localityf = \epsilon \log^{1/2}(n)$, we get that
    \begin{align}
        \Pr[b \neq \majmod{p}(z) \oplus \parity(z)] 
        &= \Pr[y \text{ is good}]\cdot  \Pr[b \neq \majmod{p}(z) \oplus \parity(z)\  \vert \ y \text{ good}]\\
        &\geq \pbra{1 - O(1/\log n)}\pbra{\frac{1}{2} - O(1/\log n)}\\
        &= \frac{1}{2} - O(1/\log(n)).
    \end{align}
\end{proof}

\section{Lower Bound Against Classical Circuits with Biased Noise}
\label{sec:biased-noise}
As discussed in \Cref{sec:U_m_theta_compiling}, the quantum-classical separations proved in this paper involve quantum circuits with $m$ qubit unitary operations, denoted $\Urot{m}{\theta}$. In \Cref{sec:U_m_theta_compiling} we showed these operations could, in principle, be written as a constant length product of single qubit unitaries and CNOT gates. However, the single qubit unitaries involved in this decomposition were arbitrary; in particular they involve rotations by angles that could scale with the parameter $\theta$ which, in turn, scaled inverse polynomially with the problem size $n$. This should be compared with the classical lower bound which, at the moment, only holds against classical circuits given access to uniform random bits. Is the quantum-classical advantage considered in this paper only a consequence of the fact the quantum circuit can perform arbitrarily small rotations, corresponding to biased sources of randomness? In this section we show the answer to this question is NO, by extending the classical lower bound proved in the main paper to the case where the $\NCZ$ circuit has access to $kn + n^\delta$ random bits with each bit drawn from a Bernoulli distribution with entropy $1/k$. 

As a warmup, we recall the relationship between the bias of a Bernoulli random variable and its entropy. 

\begin{claim}
    Let $B_b$ be Bernoulli distributed random variable with bias $b$, meaning that $B_b$ takes value $0$ with probability $1/2 + b$ and value $1$ with probability $1/2 - b$. Then we have 
    \begin{align}
        1 - 4b^2 \leq H(B_b) \leq (1-4b^2)^{1/\ln(4)}
    \end{align}
\end{claim}

\begin{proof}
This result follows immediately from standard bounds on the binary entropy function. 
\end{proof}

We also let $B_b^{\otimes \ell}$ be $\ell$ i.i.d. Bernoulli distributed random variables with bias $b$. With this notation in hand we are ready to state the main theorem of this section. 
\begin{thm}
\label{thm:lower_bound_biased_noise}
        Let $B_b$ be a Bernoulli distributed random variable with bias $b$ and entropy $H(B_b) = 1/k$. For each $\delta < 1$, there exists an $\epsilon >0$ such that for all sufficiently large even integer $N$ and prime number $p = \Theta(N^\alpha)$ for $\alpha \in (\delta/3, 1/3)$: Let $f:\{0,1\}^\ell \to \{0,1\}^{N+1}$ be an $(\epsilon \log N)^{1/2}$-local function, with $\ell \leq kN + N^\delta$. Then $\TVD(f(B_b^{\otimes \ell}), (Z, \pmmajmod{p}(Z))) \geq 1/2 - O(1/\log N)$.
\end{thm}

First, we show that when the random variables provided as input to a bounded depth circuit are too biased the circuit cannot sample from any distribution of the form $(X, g(x))$ with uniform $X$. 

\begin{lem}
\label{lem:high_bias_bound}
    Let $f : \{0,1\}^\ell \rightarrow \{0,1\}^{N+1}$ be a $\localityf$ local function. Let $B_b^{\otimes \ell}$ be $\ell$ Bernoulli distributed random variables with bias $b$. Assume $b^2 \geq 1/4 \left(1 - (1/2\localityf)^{\ln(4)}\right)$ Then, for any function $g : \{0,1\}^N \rightarrow \{0,1\}$ we have 
    \begin{align}
        \TVD(f(B_b^{\otimes \ell}), (X,g(x))) \geq \frac{1}{2} - \frac{1}{N}
    \end{align}
\end{lem}

\begin{proof}
We consider the TVD between $X$ and the first $N$ output bits of $f(B_b^{\otimes \ell})$, which we denote by $Y$. By the assumption that $f$ is $\localityf$ local, we know that $Y$ is determined by at most $\localityf N$ bits of $B_b^{\otimes \ell}$. It follows that 
\begin{align}
    H(Y) &\leq (\localityf N)H(B_b) \leq (\localityf N) (1/2\localityf) = N/2. 
\end{align}
On the other hand, $X$ is uniform, so $H(X) = N$. We see $H(X) - H(Y) \geq N/2$. 

To convert this difference in entropies to a difference in TVD consider an joint distribution $(\tilde{X},\tilde{Y})$ where the marginal distributions of $\tilde{X}$ and and $\tilde{Y}$ match the distributions of $X$ and $Y$, but are sampled from an optimal coupling between $X$ and $Y$. Then we have
\begin{align}
    P(\tilde{X} \neq \tilde{Y}) = \TVD(X,Y) =: \Delta
\end{align}
where we have introduced the shorthand $\Delta$ for this distance. 
But we can then write the entropy of $\tilde{X}$ as
\begin{align}
H(\tilde{X}) \leq (1 - \Delta) H(\tilde{X} | \tilde{X} = \tilde{Y}) + \Delta H(\tilde{X}| \tilde{X} \neq \tilde{Y}) + H(\Delta)
\end{align}
where $H(\Delta)$ is the binary entropy function. But we can bound this as 
\begin{align}
(1 - \Delta) H(\tilde{X} | \tilde{X} = \tilde{Y}) +  \Delta H(\tilde{X}| \tilde{X} \neq \tilde{Y}) + H(\Delta) 
&\leq (1 - \Delta) H(\tilde{Y}) + \Delta N + H(\Delta) \\
&\leq (1 - \Delta) (N/2) + \Delta N + H(\Delta) \\
&= N (1 +  \Delta)/2 + 1
\end{align}
Since we also have $H(\tilde{X}) = H(X) = N$, we conclude
\begin{align}
    N(1 + \Delta)/2 + 1 &\geq N \\
    \implies \Delta &\geq 1/2 - 1/N 
\end{align}
as claimed. 
\end{proof}

We are now ready to begin the proof of \Cref{thm:lower_bound_biased_noise}. 
\begin{proof}[Proof (\Cref{thm:lower_bound_biased_noise}).]
First note that by symmetry it suffices to prove the result in the case where $b > 0$ and let $\localityf = (\epsilon \log N)^{1/2}$ be the locality of the function $f$. Then, by \Cref{lem:high_bias_bound} the proof of \Cref{thm:lower_bound_biased_noise} is immediate when $b \geq 1/2 (1 - (1/2\localityf)^{\ln(4)})^{1/2}$. So it just remains to prove the result when $b \leq 1/2 (1 - (1/2\localityf)^{\ln(4)})^{1/2}$. The proof in this case largely mirrors the proof of \Cref{thm:classical_LB_tree},  with some important differences. For ease of writing we illustrate just the differences below. 

To begin: Forest Partitions, Tree Neighborhoods, and Minimal Blocks (\Cref{defn:forest_partition,defn:minimal_block,defn:tree_neighbours}) can be defined identically to the definitions given in \Cref{sec:classical_hardness_main}.  Then, for any $\localityf = (\epsilon \log N)^{1/2}$ local function $f: \{0,1\}^\ell \rightarrow \{0,1\}^{N+1}$ note we can partition in the input $u \in \{0,1\}^\ell$ as $u = (w,y)$ for $w\in \zo^s$ and $y\in \zo^{\ell - s}$ as described in \Cref{lem:partition_inputs_tree}. Recall that this partition has the property that after the input variables $y$ are fixed, each bit in $w_i \in w$ will independently control the values of some forest $F_i$ in the forest partition of the output. Additionally, we will have that $\abs{w} = s \geq \Omega(\ell / \localityf^3) \geq \Omega(N / \localityf^3)$. Now, we need to define the following, somewhat relaxed, notion of an output $z' \in \{0, 1\}^{N+1}$ being $y$-fixed. As before, this is defined relative to a function $f$ and partition of the inputs to that function. First, define the set 
\begin{align}
    \mathcal{X} ( z') = \{ (w,y) : f(w,y) = z' \}. 
\end{align}
For any $(w,y)$ let $I(y)$ be the indicator variable taking value one when $N_F \geq 2N^{3\alpha}$ blocks $g_i(w_i,y)$ are $y$-fixed. Also, for any $(w,y) \in \{0,1\}^\ell$ let $p(w,y)$ be the probability of drawing that input from the distribution $B_b^{\otimes{\ell}}$. Then we say that an output $z'$ is \textit{likely $Y$-fixed} if 
\begin{align} \label{eq:likely_y-fixed}
    2^{N} \left(\sum_{(w,y) \in \mathcal{X} ( z') } p(w,y) I (y)\right) \geq 1 / \log(N).
\end{align}
Finally, we are ready to define the following statistical test.  
\paragraph{Statistical Test:} Let $N_0, N_M := 3N^{3\alpha}$ and $N_F := 2N^{3\alpha}$. The statistical test is $T' := T_M \uplus T_0 \uplus T_F' \uplus T_S$, where 
    \begin{align}
        T_M &:=  \{ z'\in \{0,1\}^{N+1} : \  \leq N_M \text{ blocks } i\in [s] \text{ of } z' \text{ are \emph{minimal}}\} \\
        T_0 &:=  \{ z'\in \{0,1\}^{N+1} : \  z'\at{\forest{i}} = 0^{|\forest{i}|} \text{ for } \leq N_0 \text{ blocks } i\in [s]\} \\
        T_F' &:=  \{z' \in \{0,1\}^{N+1} : z' \text{ is likely $Y$-fixed}  \}\\
        T_S &:=  \{(z, b)\in \{0,1\}^{N}\times \{0,1\} : b \neq \pmmajmod{p}(z) \} \qquad \text{(``incorrect strings'')}
    \end{align}
Note this is almost the same as the statistical test defined in \Cref{sec:classical_hardness_main}, but with the condition 
\begin{align}
     T_F &:=  \{z' \in \{0,1\}^{N+1} : \exists (w,y) : f(w,y) = z' \text{ and } \geq N_F \text{ blocks } g_i(w_i, y) \text{ are } y\text{-fixed} \}
\end{align}
replaced by the condition $T_F'$. It remains to show that a sample drawn from $f(B_b^{\otimes \ell})$ passes this test with probability $1/2 - O(1/\log(N))$, while a sample from $(Z, \pmmajmod{p}(Z))$ passes this test with probability $1/N$. We begin with the second claim. 

\begin{claim}
$\Pr[(Z, \pmmajmod{p}(Z)) \in T' ] \leq 1/N$ for sufficiently large $N$.
\end{claim}

\begin{proof}
This claim is identical to \Cref{claim:pmmmp_fail_stat_test}, except that the statistical test $T_F$ has been replaced by $T_F'$. Thus, all that is required to prove the claim is to show that $D := (Z, \pmmajmod{p}(Z)) \in T_F'$ with probability at most $1 / 3N$. To prove this, we consider a coupling between the distributions $D$ and $(W,Y) \sim B_b^{\otimes \ell}$ chosen so we have
\begin{align}\label{eq:pr-conditionedonTF}
    \Pr[D = f(W,Y) \text{ and at least } N_F \text{ blocks $g_i(W_i,Y)$ are $Y$-fixed} | D \in T_F'] \geq 1/\log(N). 
\end{align}
Note that, by definition of $T_F'$, this can be accomplished by first coupling each outcome $z$ for $D$ with outcomes $(w,y)$ for $(W,Y)$ satisfying $f(w,y) = z$ and at least $N_F$ blocks $g_i(W_i,Y)$ are $Y$-fixed (until either there are no more outcomes $(w,y)$ satisfying this condition or the sum of the probabilities of the coupled outcomes $(w,y)$ satisfying this condition would exceed $2^{-N}$) then coupling the remaining outcomes arbitrarily (while preserving probabilities). To walk through why this coupling implies \Cref{eq:pr-conditionedonTF}, note that the definition of ``likely $Y$-fixed'' in \Cref{eq:likely_y-fixed} ensures that for each $z\in T_F'$, the probability that $D = f(W, Y) = z$ and at least $N_F$ blocks are $Y$-fixed is at least $2^{-n}/\log n = \Pr[D = z] / \log n$. So if we consider the subset of probability mass of $(D, W, Y)$ where $D \in T_F'$, at least a $1/\log n$-fraction of the mass satisfies $D = f(W, Y)$ and $\geq N_F$ blocks are $Y$-fixed.

Then, conditioning on this event and using that $D$ is a distribution on at most $N+1$ binary variables it follows that 
\begin{align}
    H(D | D \in T_F') \leq \gamma H[f(W,Y) | \; \text{at least } N_F \text{ blocks $g_i(W_i,Y)$ are $Y$-fixed}] + (1- \gamma) (N+1) + H(\gamma) 
\end{align}
where $\gamma  = 1/\log(N)$, and we used the standard entropy inequality that $H(A) \leq \sum_b Pr[B = b] H(A| B = b) + H(B)$ for any distributions $A$ and $B$. But then we also know that if at least $N_F$ blocks $g_i(W_i,Y)$ are $Y$-fixed then $f(W,Y)$ is a function of at most $kN + N^\delta - N_F$ random variables each with entropy at most $1/k$. Thus we have 
\begin{align}
H[f(W,Y) | \; \text{at least } N_F \text{ blocks $g_i(W_i,Y)$ are $Y$-fixed}] \leq N + (N^\delta - N_F)/k
\end{align}
So it follows that 
\begin{align}
H(D | D \in T_F') &\leq N - \gamma (N_F - N^\delta)/k + (1-\gamma) + H(\gamma) \\
&\leq N - \gamma (N^{3\alpha})/k + 2.
\end{align}
Finally, since $D$ is uniformly distributed over bistrings of the form $(Z, \pmmajmod{p}(Z))$ we also know that 
\begin{align}
    \abs{T_F'} &= 2^{H(D | D \in T_F') } \\
    &= 2^{N - \gamma (N^{3\alpha})/k + 2}
\end{align}
and thus 
\begin{align}
\Pr[D \in T_F'] \leq 2^{- \gamma (N^{3\alpha})/k + 2}
\end{align}
But now we use that $b < 1/2( 1- (1/2d)^{\ln(4)})^{1/2}$ and 
hence $k = H(b) \leq d = (\epsilon \log(N))^{1/2}$ to conclude
\begin{align}
    \Pr[D \in T_F'] \leq 2^{2 - \gamma N^{3\alpha} / (\epsilon \log(N))^{1/2}} = 2^{2 - N^{3\alpha} / (\epsilon^{1/2} \log(N))^{3/2})}  \leq (1/3N)
\end{align}
for sufficiently large $N$, as desired.
\end{proof}
We now move on to the first claim. Our first step is to relate the notion of being likely $Y$-fixed introduced in this section to the stronger $y$-fixed criterion used in \Cref{sec:classical_hardness_main}.
\begin{claim}\label{claim:inT_FnotT_Fp}
Given $(W,Y) \sim B_b^{\otimes \ell}$ we have 
\begin{align}
    \Pr[f(W,Y) \notin T_F' \textbf{ and } Y \text{ fixes at least } N_F \text{ blocks } g_i(W_i,Y)] \leq 1/\log(N)
\end{align}
\end{claim}

\begin{proof}
Let $(\mathcal{W}, \mathcal{Y})$ be the set of all $(w,y) \in \{0,1\}^\ell$ with at least $N_F$ blocks being $y$-fixed but with $f(w,y) \notin T_F'$. Also, for any $z' \in \{0,1\}^{N+1}$ define $\mathcal{X}(z')$ as above, so $\mathcal{X}(z') := \{(w,y) : f(w,y) = z'\}$. Then, letting $p(w,y)$ be the probability of drawing input $(w,y)$ from the distribution $B_b^{\otimes \ell}$ we see:
\begin{align}
\sum_{(w,y) \in (\mathcal{W},\mathcal{Y})} p(w,y) &= \sum_{z' \notin T_F'} \left[ \sum_{(w,y) \in (\mathcal{W},\mathcal{Y}) \cap \mathcal{X}(z') } p(w,y) \right] \\
&\leq \sum_{z' \notin T_F'} 2^{-N} (1 / \log(N)) \leq 1/ \log(N)
\end{align}    
where we used that, by definition, any $z' \notin T_F$ has at most a $2^{-N} / \log(N)$ chance of being the image of a $(w,y)$ with more than $N_F$ blocks which are $y$-fixed, so
\begin{align}
    \sum_{(w,y) \in (\mathcal{W},\mathcal{Y}) \cap \mathcal{X}(z')} p(w,y) \leq 2^{-N}/\log(N)
\end{align}
for any $z' \notin T_F'$. 
\end{proof}

We now move on to the main Claim. 

\begin{claim}
$\Pr[f(B_b^{\otimes \ell}) \in T] \geq \frac{1}{2} - O(1/\log N)$.
\end{claim}
\begin{proof}
This proof follows similarly to the proof of \Cref{claim:pmmajmod_statistical_test_pass}, again with some additional complications.

We first define three subsets $A_1, A_2, A_3 \subseteq \zo^{\ell - s}$ of all possible values that $Y$ can take.
\begin{align}
    A_1 &:= \cbra{y \in \zo^{\ell - s} : \forall w\in \zo^s, z' = f(w,y) \text{ satisfies } z' \in T_0 \cup T_M \cup T_F}\\
    A_2 &:= \cbra{y \in \zo^{\ell - s} : y \text{ fixes at least } N_F \text{ blocks } g_i(w_i, y)}\\
    A_3 &:= \cbra{y \in \zo^{\ell - s} : y \text{ fixes } 
    < N_F \text{ blocks}, \textbf{ and } \exists w,w' \in \zo^s : f(w,y) \notin T_0, f(w',y)\notin T_M}.
\end{align}
It is easy to verify that $A_1 \cup A_2 \cup A_3 = \zo^{\ell - s}$. It is then sufficient for us to bound the $\Pr[f(W, Y) \in T \ \textbf{and} \ Y \in A_i]$ for each $i\in \{1,2,3\}$ since 
\begin{align}\label{eq:sum_of_probs_Ai}
    \Pr[f(W,Y) \notin T] \leq \sum_{i=1}^3 \Pr[f(W,Y) \notin T\  \textbf{and} \ Y \in A_i].
\end{align}
By the definition of $A_1$, we have that for each $y\in A_1$, and each $w\in \zo^s$, $f(w,y) \in T$. Therefore $\Pr[f(W,Y) \notin T\  \textbf{and} \ Y \in A_1]~=~0$. In \Cref{claim:inT_FnotT_Fp} we already proved that $\Pr[f(W,Y) \notin T\  \textbf{and} \ Y \in A_2] \leq O(1/\log N)$. All that remains is to bound $\Pr[f(W,Y) \notin T | Y \in A_3]$. The proof here follows the same steps as the proof of~\Cref{claim:pmmajmod_statistical_test_pass}, except that the random variables $x_i$ are no longer uniform but are now biased Bernoulli random variables. As mentioned previously, by \Cref{lem:high_bias_bound} we can assume this bias $b$ is upper bounded by 
\begin{align}
    b^2 \leq 1/4 \left(1 - (1/2 \sqrt{\epsilon \log(N)} )^{\ln(4)}\right).
\end{align}
To accommodate these biased random variables we require the following variant of \Cref{fact:puniform}.

\begin{fact}\label{fact:pbernoulli_bits}
Let $a_1, a_2, ..., a_t$ be nonzero integers modulo $p$ and let $X_1, X_2, ... X_t \in \{0,1\}^n$ be i.i.d. Bernoulli random variables sampled from a distribution with bias $b$. Then the total variation distance distance between $\sum_i a_i X_i \mod{p}$ and the uniform distribution over $\{0,1,..., p-1 \}$ is at most $\sqrt{p} \exp( - \Omega(  t (1 - 4b^2) / p^2))$. 
\end{fact}

\begin{proof}
We follow the same proof as outlined in Fact 3.2 of~\cite{viola2012complexity}. By Claim 33 of~\cite{bogdanov2010pseudorandom} we have that the total variation distance is at most 
\begin{align}
    \sqrt{p} \max_{a \neq 0} \left| \mathbb{E}_X \left[\exp (a 2 \pi i \sum_j a_j X_j)\right]  - \mathbb{E}_{U_p} \left[\exp(a 2 \pi i U_p) \right]\right|
\end{align}
where $U_p$ denotes the uniform distribution and the first expectation is taken over the Bernoulli random variables. For any $a \neq 0$ we have $\mathbb{E}_{U_p} \left[\exp(a 2 \pi i U_p) \right] = 0$. Also, by Lemma 13 of \cite{lovett2009pseudorandom} we have 
\begin{align}
    \mathbb{E}_X \left[\exp (a 2 \pi i \sum_j a_j X_j)\right] \leq \exp( - \Omega(  t (1 - 4b^2) / p^2)).
\end{align}
\end{proof}

Mirroring the use of~\Cref{fact:puniform}, we will now prove the straightforward corollary of \Cref{fact:pbernoulli_bits}. 

\begin{cor}\label{cor:prob_sum_in_set_biased}
    For each prime $p = \Theta(N^\alpha)$ with $\alpha <1$,  $t = \Omega(p^3)$, $a_0, a_1, \dots a_t$ nonzero integers modulo $p$, and $A \subseteq \{0, 1, \dots p-1\}$, let $X_1, ..., X_t$ be i.i.d. Bernoulli random variables with bias $b$ bounded above by $b^2 \leq 1/4( 1 - (1/2\sqrt{\epsilon \log(N)})^{\ln(4)} ) $. Then
    \begin{align}
        \frac{|A|}{p} - O(1/N) \leq \Pr_{X}\left[a_0 + \sum_{i=1}^t a_i X_i \in A\right] \leq \frac{|A|}{p} + O(1/N)
    \end{align}
\end{cor}

\begin{proof}
    Let $U_p$ be the uniform distribution over $\{0,1,..., p-1\}$. Then the result follows from the observation that 
    \begin{align}
        \TVD\left(U_p, a_0 + \sum_{i=1}^t a_i X_i\right) \leq N^{\alpha/2} \exp( - \Omega ( N^\alpha / \sqrt{\epsilon \log(N) } )) \leq O(1/N)
    \end{align}
\end{proof}
Note that this gives exactly the same asymptotic scaling as~\Cref{cor:prob_sum_in_set}. 
Then, following the same argument in \Cref{claim:pmmajmod_statistical_test_pass} and using \Cref{cor:prob_sum_in_set_biased} in place of \Cref{cor:prob_sum_in_set} gives the bound 
\begin{align}
    \Pr[f(W,Y) \notin T | Y \in A_3] &\leq 1/2 + O(1/\log(N)).
\end{align} Combining this with \Cref{eq:sum_of_probs_Ai}, we get that
\begin{align}
    \Pr[f(W,Y) \notin T] \leq O(1/\log n) + \Pr[Y \in A_3] \cdot  1/2 + O(1/\log(N)) = 1/2 + O(1/\log N).
\end{align}
This completes the proof.
\end{proof}
\end{proof}

\bibliographystyle{quantum}
\bibliography{ref}
\end{document}